\documentclass{article}

\usepackage[margin=1.25in]{geometry}
\usepackage[utf8]{inputenc} 
\usepackage[T1]{fontenc}    
\usepackage{amsmath,amsfonts,amsthm,amssymb}
\usepackage{hyperref}
\usepackage[numbers]{natbib}

\newtheorem{theorem}{Theorem}[section]
\newtheorem{corol}[theorem]{Corollary}
\newtheorem{proposition}[theorem]{Proposition}

\newtheorem{remark}{Remark}
\newtheorem{example}{Example}
\newtheorem{fact}{Fact}

\usepackage{graphicx} 
\usepackage{caption}
\usepackage{subcaption}
\usepackage{enumitem}
\usepackage{cleveref}

\graphicspath{{fig/}}

\newif\ifarxiv
\arxivtrue

\newif\ifdiff

\usepackage{changes}

\usepackage{url}            
\usepackage{booktabs}       

\usepackage{nicefrac}       
\usepackage{microtype}      
\usepackage{mathtools}
\usepackage{xfrac}
\usepackage{csquotes}
\usepackage{algorithm}
\usepackage{algpseudocode}
\usepackage{makecell}
\usepackage{bbm}
\newcommand{\tabitem}{~~\llap{\textbullet}~~}
\usepackage{tikz}
\usetikzlibrary{shapes.geometric, arrows, positioning, calc}

\usepackage[]{xcolor}

\usepackage{bibunits}

\usepackage{numprint}

\newcommand{\Ex}{\mathbb{E}}
\newcommand{\bd}{\boldsymbol{d}}
\newcommand{\bq}{\boldsymbol{q}}
\newcommand{\bphi}{\boldsymbol{\phi}}
\newcommand{\nf}{2^{\mathbb{N}_f}}

\defaultbibliography{main}
\defaultbibliographystyle{plainnat}

\title{Admissible online closed testing must employ e-values}

\author{%
  Lasse Fischer and Aaditya Ramdas\\
   University of Bremen and Carnegie Mellon University \\
  \texttt{fischer1@uni-bremen.de,  aramdas@cmu.edu}
}

\begin{document}

\maketitle
\begin{abstract}
 In contemporary research, data scientists often test an
 infinite sequence of hypotheses $H_1,H_2,\ldots $ one by one, and are required to make real-time decisions without knowing the future hypotheses or data. In this paper, we consider such an online multiple testing problem with the goal of providing simultaneous lower bounds for the number of true discoveries in data-adaptively chosen rejection sets. 
 Employing the recent online closure principle, we show that for this task it is necessary to use an anytime-valid test for each intersection hypothesis. This connects two distinct branches of the literature: online testing of multiple hypotheses (where the hypotheses appear online), and sequential anytime-valid testing of a single hypothesis (where the data for a fixed hypothesis appears online). Motivated by this result, we construct a new online closed testing procedure and a corresponding short-cut with a true discovery guarantee based on multiplying sequential e-values. This general but simple procedure gives uniform improvements over the state-of-the-art methods but also allows to construct entirely new and powerful procedures. 
 
\end{abstract}


\begin{bibunit}
\section{Introduction}

In the online multiple testing framework a potentially infinite stream of hypotheses $H_1, H_2, \ldots$ is tested one by one over time \citep{foster2008alpha, javanmard2018online}. At each time $t\in \mathbb{N}$, one must decide on the hypothesis $H_t$ without knowing the future hypotheses or data. This setting occurs in the tech industry \citep{kohavi2013online, ramdas2017online}, machine learning \citep{feng2021approval, zrnic2020power}, open data repositories as used in genomics \citep{aharoni2014generalized, mouse, 10002015global} and other data science tasks where flexible and real-time decision making is required. 

A common error metric for a chosen rejection set $R_t$ at time $t$ is the \emph{false discovery proportion}
\begin{align}
    \mathrm{FDP}(R_t)=\frac{\text{Number of true hypotheses in } R_t}{\text{Size of }R_t}. \label{eq:FDP}
\end{align}
Classical offline and online methods typically control its expectation, the false discovery rate (FDR) \citep{benjamini1995controlling}, below some level $\alpha\in (0,1)$. However, since the $\mathrm{FDP}$ may have high variance, controlling its expectation may not be enough.  

In a seminal work, \citet{goeman2011multiple}  proposed to control the tail probabilities of $\mathrm{FDP}(S)$ \emph{simultaneously over all possible sets} $S$ instead. More precisely, they suggest to provide an upper bound $\bq(S)$ for $\mathrm{FDP}(S)$ \ifdiff \added{such that 
$\bq(S) \geq \mathrm{FDP}(S) $ for all $S$ with probability at least $1-\alpha$.} \else such that 
$\bq(S) \geq \mathrm{FDP}(S)$ for all $S$ with probability at least $1-\alpha$. \fi

A major advantage of such simultaneous bounds on the FDP, compared to FDR control, is that the final rejection set(s) can be chosen post-hoc, meaning after looking at the data and calculating the bounds $\bq(S)$, without violating this error control. In other words, a scientist is allowed to query many (or all) such sets $S$, examine the reported bounds, and later choose one or a few final sets (and bounds) to report or follow up on. Since many more sets $S$ will be queried than will be rejected, we call these sets $S$ as \emph{query} sets (rather than, say, \emph{rejection} sets).

Bounds that hold with high probability are advantageous when an inflated FDP has severe consequences. However, there are also connections between FDR and simultaneous FDP controlling procedures \citep{goeman2019simultaneous,katsevich2020simultaneous,iqraa2024false}. Following \citet{goeman2021only}, we consider simultaneous \emph{lower bounds on the number of true discoveries} $\bd(S)\coloneqq (1-\bq(S))|S|$, which is equivalent to bounding the FDP from above but often more convenient. Their work focused on offline multiple testing, while we study this error metric in the online setting.

At each time $t$, an \textit{online true discovery procedure} $\bd$ allows the scientist to pick any (and possibly multiple) query sets $S\subseteq \{1,\ldots,t\}$ of interest. Then $\bd(S)$ instantly  provides a lower bound for the number of false hypotheses in $S$, which holds true with probability at least $1-\alpha$ simultaneously over all sets that have been queried or that might be queried in the future. Note that this simultaneity permits to stop (or continue) the testing process data-adaptively at any time, e.g.\ after the 50th discovery. This is an additional benefit compared to online FDR control, which is usually only provided at fixed times \citep{zrnic2021asynchronous, xu2022dynamic}. Figure~\ref{fig:flowchart_online} illustrates this process. 

\ifarxiv
\begin{figure}[tb]
\centering
\begin{tikzpicture}[node distance=2cm, every node/.style={draw, text width=6.5cm, align=center, rounded corners}, every path/.style={draw,-latex,thick,->,>=stealth}]
    \node (H) {(1) Scientist proposes $H_t$};
    \node (X) [right of=H, xshift=6cm] {(2) Scientist collects/observes data $X_t$};
    \node (S) [below of=H] {(3) Scientist chooses (possibly multiple) \\ query sets $S\subseteq \{1,\ldots,t\}$};
    \node (d) [right of=S, xshift=6cm] {(4) Statistician outputs $\boldsymbol{d}(S)$ };
    
    \path (H) -- (X);
    \path (X) -- (S);
    \path (S) -- (d);
\end{tikzpicture}
\caption{Illustration of using an online procedure with simultaneous true discovery guarantee. At each time $t$, (1) the scientist proposes a hypothesis $H_t$ for testing, possibly based on the data used for testing the previous hypotheses; (2) the scientist collects the data $X_t$ required for testing $H_t$ which may consist of new data and/or the reuse of old data; (3) the scientist chooses based on the data observed so far (possibly multiple) query sets $S\subseteq \{1,\ldots,t\}$ of interest --- if the scientist is interested in rejections, these query sets could also be interpreted as candidates for rejection; (4) the statistician employs the online procedure $\bd$ to provide lower bounds $\bd(S)$ for the number of true discoveries in the query sets $S$  that hold simultaneously with probability of at least $1-\alpha$. }
\label{fig:flowchart_online}
\end{figure}
\else
\begin{figure}[tb]
\centering
\begin{tikzpicture}[node distance=2cm, every node/.style={draw, text width=5.5cm, align=center, rounded corners}, every path/.style={draw,-latex,thick,->,>=stealth}]
    \node (H) {(1) Scientist proposes $H_t$};
    \node (X) [right of=H, xshift=6cm] {(2) Scientist collects/observes data $X_t$};
    \node (S) [below of=H] {(3) Scientist chooses (possibly multiple) \\ query sets $S\subseteq \{1,\ldots,t\}$};
    \node (d) [right of=S, xshift=6cm] {(4) Statistician outputs $\boldsymbol{d}(S)$ };
    
    \path (H) -- (X);
    \path (X) -- (S);
    \path (S) -- (d);
\end{tikzpicture}
\caption{Illustration of using an online procedure with simultaneous true discovery guarantee. At each time $t$, (1) the scientist proposes a hypothesis $H_t$ for testing, possibly based on the data used for testing the previous hypotheses; (2) the scientist collects the data $X_t$ required for testing $H_t$ which may consist of new data and/or the reuse of old data; (3) the scientist chooses based on the data observed so far (possibly multiple) query sets $S\subseteq \{1,\ldots,t\}$ of interest --- if the scientist is interested in rejections, these query sets could also be interpreted as candidates for rejection; (4) the statistician employs the online procedure $\bd$ to provide lower bounds $\bd(S)$ for the number of true discoveries in the query sets $S$  that hold simultaneously with probability of at least $1-\alpha$. }
\label{fig:flowchart_online}
\end{figure}
\fi

A recently popularized approach for online testing of a single hypothesis, where not the hypotheses but the data itself comes in sequentially, is the e-value. An e-value $E_t$ for a hypothesis $H_t$ is a nonnegative random variable which has expected value less than or equal to one if $H_t$ is true. The e-value (or its sequential extension, the e-process) is an alternative to the well-known p-value but more suitable for settings where early stopping of the sampling process or optional continuation is desired \citep{shafer2011test, grunwald2020safe, ramdas2023game}. 
In this paper, we exploit this sequential suitability of the e-value to design online multiple testing procedures with simultaneous true discovery guarantees. This connects two mostly separate areas of the literature: online testing of multiple hypotheses and sequential testing of a single hypothesis. \ifdiff \added{In Supplementary Material~\ref{sec:related_lit} we provide a more detailed description of these existing branches of literature.} \else In Supplementary Material~\ref{sec:related_lit} we provide a more detailed description of these existing branches of literature. \fi

\subsection{Our contribution}

\ifdiff \added{Our contribution can be divided into a theoretical and a methodological part.} \else  Our contribution can be divided into a theoretical and a methodological part. \fi

\ifarxiv
\paragraph{Our theoretical contribution.}
\else 
\ifdiff 
\added{\textbf{Our theoretical contribution.}}\hspace{0.3cm }
\else 
\textbf{Our theoretical contribution.}\hspace{0.3cm }
\fi
\fi
We provide new insights about the central role of e-values in online multiple testing. In particular, we show that in order to derive admissible online procedures with simultaneous true discovery guarantee, we must necessarily employ \emph{anytime-valid tests} for each intersection hypothesis, which in turn \emph{must} employ test martingales \citep{ramdas2020admissible}. By transitivity, the construction of these online true discovery procedures must rely on test martingales, which are sequential generalizations of e-values. Thus, e-values enter naturally into the construction of admissible online procedures. Even if an online procedure is constructed solely based on p-values, it must implicitly employ anytime-valid tests and can be reconstructed or improved using e-values. 

\ifarxiv
\paragraph{Our methodological contribution.}
\else 
\ifdiff 
\added{\textbf{Our methodological contribution.}}\hspace{0.3cm }
\else 
\textbf{Our methodological contribution.}\hspace{0.3cm }
\fi
\fi
Guided by our theoretical results, we construct \texttt{SeqE-Guard}, a powerful and computationally efficient online true discovery procedure based on multiplying sequential e-values. \ifdiff \added{This improves the existing methodology on several fronts.} \else 
This improves the existing methodology on several fronts. \fi
\ifdiff
\begin{enumerate}
    \item \added{Due to the flexibility provided by e-values,  \texttt{SeqE-Guard} is a very general procedure. We show that by plugging specific e-values into \texttt{SeqE-Guard}, the online methods of \citet{katsevich2020simultaneous} and \citet{iqraa2024false}, which are all online true discovery procedures we know of, can be uniformly improved. In addition, \texttt{SeqE-guard} improves some offline procedures (Section~\ref{sec:katse}).}
    \item \added{The \texttt{SeqE-Guard} algorithm operates according to simple rules and is therefore easy to analyze. For instance, \texttt{SeqE-Guard} reduces the complex existing methods by \citet{katsevich2020simultaneous} and \citet{iqraa2024false} with various parameters to simple e-values that can only take two or three different values (Table~\ref{tab:improvements}). This facilitates the interpretation and comparison of these methods.}
    \item \added{The existing online procedures only provide bounds for one query path, meaning a specific nondecreasing sequence of query sets. \texttt{SeqE-Guard} allows to analyze all possible query sets simultaneously, without sacrificing power (Corollary~\ref{corol:multiple_paths}).}
    \item \added{\texttt{SeqE-Guard} does not only improve existing procedures, but can also be used to construct entirely new methods, depending on the choice of the e-values. We investigate the use of \texttt{SeqE-Guard} with growth rate optimal (GRO) e-values \citep{shafer2021testing, grunwald2020safe} and propose a hedging approach to increase the power to detect false hypotheses (Section~\ref{sec:gro}). }
\end{enumerate}
\else
\begin{enumerate}
    \item Due to the flexibility provided by e-values,  \texttt{SeqE-Guard} is a very general procedure. We show that by plugging specific e-values into \texttt{SeqE-Guard}, the online methods of \citet{katsevich2020simultaneous} and \citet{iqraa2024false}, which are all online true discovery procedures we know of, can be uniformly improved. In addition, \texttt{SeqE-guard} improves some offline procedures (Section~\ref{sec:katse}).
    \item The \texttt{SeqE-Guard} algorithm operates according to simple rules and is therefore easy to analyze. For instance, \texttt{SeqE-Guard} reduces the complex existing methods by \citet{katsevich2020simultaneous} and \citet{iqraa2024false} with various parameters to simple e-values that can only take two or three different values (Table~\ref{tab:improvements}). This facilitates the interpretation and comparison of these methods.
    \item The existing online procedures only provide bounds for one query path, meaning a specific nondecreasing sequence of query sets. \texttt{SeqE-Guard} allows to analyze all possible query sets simultaneously, without sacrificing power (Corollary~\ref{corol:multiple_paths}).
    \item \texttt{SeqE-Guard} does not only improve existing procedures, but can also be used to construct entirely new methods, depending on the choice of the e-values. We investigate the use of \texttt{SeqE-Guard} with growth rate optimal (GRO) e-values \citep{shafer2021testing, grunwald2020safe} and propose a hedging approach to increase the power to detect false hypotheses (Section~\ref{sec:gro}).
\end{enumerate}
\fi

\subsection{Paper outline}

In Section~\ref{sec:online_true_control}, we define the online setting formally and recap concepts like coherence (Section~\ref{sec:coherence}) and (online) closed testing (Section~\ref{sec:online_closed_test}). Afterwards, we introduce a general approach to online true discovery guarantee based on test martingales and prove that every procedure must be constructed in that way (Section~\ref{sec:test_martingale}). 
In Section~\ref{sec:seq_e-value}, we consider online true discovery guarantee with sequential e-values and propose our \texttt{SeqE-Guard} algorithm for this task. By plugging specific sequential e-values into \texttt{SeqE-Guard} we immediately obtain uniform improvements of the state-of-the-art methods by \citet{katsevich2020simultaneous, iqraa2024false} (Section~\ref{sec:katse}). Separately, we consider \texttt{SeqE-Guard} with GRO e-values in Section~\ref{sec:gro}. In Section~\ref{sec:sim}, we perform simulations to compare the proposed methods and to quantify the gain in power obtained by our improvements\footnote{The code for the simulations is available at  \url{github.com/fischer23/online_true_discovery}.}.


\section{Online true discovery guarantee\label{sec:online_true_control}}

In this section, we introduce notation and recall concepts like simultaneous (online) true discovery guarantee and (online) closed testing. Then, we prove that any admissible procedure for delivering an online true discovery guarantee must rely on test martingales (sequential generalizations of e-values).

We consider the general online multiple testing setting described in \citep{fischer2024online}. Let $(\Omega, \mathbb{F})$, where $\mathbb{F}=(\mathcal{F}_i)_{i\in \mathbb{N}_0}$, be a filtered measurable space and $\mathcal{P}$ some set of probability distributions on $(\Omega, \mathbb{F})$. The $\sigma$-field $\mathcal{F}_i$ defines the information that can be used for testing hypothesis $H_i$ ($\mathcal{F}_0=\{\emptyset, \Omega\}$). Hence, in online multiple testing every hypothesis test is only allowed to use some partial information which is increasing over time. One can think of $\mathcal{F}_i$ as the data that is available at time $i$. However, we may also include external randomization or coarsen the filtration. For example, many existing works on online multiple testing consider $\mathcal{F}_i=\sigma(P_1,\ldots,P_i)$ \citep{foster2008alpha, javanmard2018online}, where each $P_j$ is a p-value calculated for hypothesis $H_j$. We would also like to point out that the hypotheses do not have to be prespecified. Hence, in practice one is allowed to data-adaptively construct each $H_i$ based on $\mathcal{F}_{i-1}$.

We assume that the data follows some unknown distribution $\mathbb{P} \in \mathcal{P}$. A null hypothesis $H\subseteq \mathcal{P}$ is a collection of probability distributions; we are effectively testing whether $\mathbb P \in H$. When $\mathbb P \in H$, we say that $H$ is true null, and otherwise we call it a false null. We define $I_0^\mathbb{P}\coloneqq \{i\in \mathbb{N}: \mathbb{P} \in H_i\}$ and $I_1^\mathbb{P}\coloneqq \mathbb{N}\setminus I_0^\mathbb{P}$ as the index sets of true and false null hypotheses, respectively. 

As defined in \citep{goeman2021only, fischer2024online}, a procedure with \textit{simultaneous true discovery guarantee} is a random function $\boldsymbol{d}:2^{\mathbb{N}_f}\to \mathbb{N} \cup \{0\}$, where $2^{\mathbb{N}_f}$ is the set of all finite subsets of $\mathbb{N}$ (analogously we use $2^{\mathbb{N}_{-f}}$ for the set of all infinite subsets of $\mathbb{N}$), such that for all $\mathbb{P} \in \mathcal{P}$:
\begin{align} \mathbb{P}(\boldsymbol{d}(S)\leq |S\cap I_1^{\mathbb{P}}| \text{ for all } S\in 2^{\mathbb{N}_f}) \geq 1-\alpha. \label{eq:simult_guarantee}\end{align}
Clearly, if $\boldsymbol{d}$ always outputs 0, it achieves simultaneous true discovery guarantee. Thus, implicitly, the larger $\boldsymbol{d}$ is, the better (here, larger is meant componentwise; $\boldsymbol{d} \geq \boldsymbol{d}'$ if $\boldsymbol{d}(S) \geq \boldsymbol{d}'(S)$ for all $S \in 2^{\mathbb{N}_f})$. \ifdiff \added{It should be noted that the sets $S$ in \eqref{eq:simult_guarantee} can be considered as fixed because of the quantification \enquote{for all S} (in contrast to classical FDR control), but the procedure $\bd$ depends on the data and thus the probability is taken with regard to $\bd$. However, the simultaneous guarantee of $\bd$  particularly allows to choose query sets $S$ arbitrarily based on the data while ensuring that $\boldsymbol{d}(S)\leq |S\cap I_1^{\mathbb{P}}|$ with probability at least $1-\alpha$.} \else  It should be noted that the sets $S$ in \eqref{eq:simult_guarantee} can be considered as fixed because of the quantification \enquote{for all S} (in contrast to classical FDR control), but the procedure $\bd$ depends on the data and thus the probability is taken with regard to $\bd$. However, the simultaneous guarantee of $\bd$  particularly allows to choose query sets $S$ arbitrarily based on the data while ensuring that $\boldsymbol{d}(S)\leq |S\cap I_1^{\mathbb{P}}|$ with probability at least $1-\alpha$. \fi

$\boldsymbol{d}$ is called  an \textit{online} true discovery procedure if $\boldsymbol{d}(S)$ is measurable with respect to $\mathcal{F}_{\max(S)}$ for all  $S\in 2^{\mathbb{N}_f}$ \citep{fischer2024online}. This ensures that at any time $t\in \mathbb{N}$ the procedure $\bd$ provides a lower bound for the number of false hypotheses \ifdiff \added{(or, equivalently, true discoveries)} \else (or, equivalently, true discoveries) \fi  in every set $S\subseteq \{1,\ldots,t\}$ with $\max(S)=t$ that remains the same no matter how many hypotheses will be tested in the future. \ifdiff \added{See also Figure~\ref{fig:flowchart_online} for an illustration of using online true discovery procedures.} \else See also Figure~\ref{fig:flowchart_online} for an illustration of using online true discovery procedures. \fi

Note that one does not have to consider an infinite number of hypotheses but could just stop at some finite time $i\in \mathbb{N}$ by setting $H_j=\mathcal{P}$ for all $j>i$ and $\bd(S)=\bd(S\cap \{1,\ldots,i\})$ for all $S \in 2^{\mathbb{N}_f}$ with $\max(S)>i$. With this, the online setting becomes classical offline testing in the case of $\mathcal{F}_1=\mathcal{F}_2=\ldots$ and thus online multiple testing can be seen as a true generalization of classical multiple testing \citep{fischer2024online}. Although we are mainly interested in the strict online case $\mathcal{F}_1\subset \mathcal{F}_2\subset \ldots$, this implies that all online procedures in this paper also apply in the offline setting.

\ifdiff \added{Further note that simultaneous true discovery procedures generalize many existing error rates, such as k-FWER or false discovery exceedance (FDX) \citep{goeman2021only}. In particular, this shows that requiring simultaneous control \textit{over all} sets $S\in \nf$ is not conservative. Intuitively, this is because a procedure $\bd$ that assigns lower bounds only to a subset $G$ of $\nf$ can be interpolated to a procedure on the entire set $\nf$ (see \citet{goeman2019simultaneous}).} \else Further note that simultaneous true discovery procedures generalize many existing error rates, such as k-FWER or false discovery exceedance (FDX) \citep{goeman2021only}. In particular, this shows that requiring simultaneous control \textit{over all} sets $S\in \nf$ is not conservative. Intuitively, this is because a procedure $\bd$ that assigns lower bounds only to a subset $G$ of $\nf$ can be interpolated to a procedure on the entire set $\nf$ (see \citet{goeman2019simultaneous}). \fi



In Figure~\ref{fig:flow_theorems}, we connect the results that will be given in the rest of the section and the related works \citep{ramdas2020admissible, fischer2024online}. The left-hand side provides a general approach to construct coherent online true discovery procedures based on test martingales and the right-hand side proves that every admissible online procedure must be constructed in that way. While most of the upper and lower connections are known, the main insights provided by this paper are the middle connections between increasing families of online intersection tests and anytime-valid tests, connecting two distinct branches of the literature. All terms and results will be clarified in the following subsections from top to bottom, starting with \textit{coherent online procedures}.

\ifarxiv
\begin{figure}[tb]
\begin{tikzpicture}[node distance=2cm, every node/.style={text width=5cm, align=center, rounded corners}, every path/.style={-latex,thick,->,>=stealth}]
    \node[draw] (COP) {Coherent online procedure $\bd$};
    \node[draw] (IFOIT) [below of=COP, yshift=-0.3cm] {Increasing family of online  intersection tests $\bphi$};
    \node[draw] (AVT) [below of =IFOIT, yshift=-0.3cm] {Anytime-valid test $(\psi_i)_{i\in I}$};
    \node[draw] (TM) [below of=AVT, yshift=-0.3cm] {Test martingale $(M_i^{\mathbb{P}})_{i\in I}$};
    
   \path (COP.340) edge[bend left=30] node[right, xshift=0.1cm, text width=6.5cm,align=left] {Theorem~\ref{prop:coherent_admissible}:\\
   $\phi_S=\mathbbm{1}\{\boldsymbol{d}(S)>0\} $ \\ $\phi_J=\sup\{\phi_S: S\subseteq J, S\in 2^{\mathbb{N}_f} \} $} (IFOIT.30);
    \path (IFOIT.150) edge[bend left=30] node[left, xshift=0.5cm, text width=6cm,align=left] {Online closure principle \citep{fischer2024online}:\\
    $\boldsymbol{d}^{\boldsymbol{\phi}}(S)= \inf\{|S\setminus I| : I\subseteq \mathbb{N}, \phi_I=0 \}$} (COP.200);
    
    \path (IFOIT.330) edge[bend left=30] node[right,yshift=0cm, xshift=0.1cm,align=left] {Theorem~\ref{theo:avTest}:\\
    $\psi^I_i=\phi_{I\cap \{1,\ldots,i\}}$} (AVT.20);
    \path (AVT.160) edge[bend left=30] node[left,yshift=0.1cm, text width=8.3cm, xshift=1.1cm,align=left] {Theorem~\ref{theo:av_test_family}:\\ $\phi_I=\inf\left\{\sup_{i\in I} \psi_i^J:J\cap \{1,\ldots,\sup(I)\}=I\right\}$ } (IFOIT.210);
    
    \path (AVT.340) edge[bend left=30] node[right, align=left, xshift=0.1cm,yshift=-0.05cm] {Fact~\ref{fact}:\\
    $M_i^{\mathbb{P}}= \dfrac{\mathbb{P}(\exists j\in I: \psi_j=1| \mathcal{F}_{i})}{\mathbb{P}(\exists j\in I: \psi_j=1)}$} (TM.20);
    \path (TM.160) edge[bend left=30] node[left, text width=5.5cm, align=left, xshift=0.25cm] {Ville's inequality: \\ $\psi_i=\inf_{\mathbb{P}\in H}\mathrm{sup}_{j\leq i} \mathbbm{1}\left\{M_j^{\mathbb{P}}\geq 1/\alpha\right\}$} (AVT.200);

\end{tikzpicture}
\caption{Illustration of the relation between online true discovery procedures, increasing families of online intersection tests, anytime-valid tests and test martingales. The left path from bottom to top provides a general approach for the construction of coherent online procedures with true discovery guarantee based on test martingales. The right path from top to bottom shows that coherent online procedure implicitly define test martingales. Taken together, the loop shows, in principle, how to improve any given coherent online procedure: we take the right path down, and then take the left path up. If the original procedure is admissible, following the loop  leaves the procedure unchanged.}
\label{fig:flow_theorems}
\end{figure}
\else 
\begin{figure}[tb]
\begin{tikzpicture}[node distance=2cm, every node/.style={text width=5.5cm, align=center, rounded corners}, every path/.style={-latex,thick,->,>=stealth}]
    \node[draw] (COP) {Coherent online procedure $\bd$};
    \node[draw] (IFOIT) [below of=COP, yshift=-0.6cm] {Increasing family of online  intersection tests $\bphi$};
    \node[draw] (AVT) [below of =IFOIT, yshift=-0.6cm] {Anytime-valid test $(\psi_i)_{i\in I}$};
    \node[draw] (TM) [below of=AVT, yshift=-0.6cm] {Test martingale $(M_i^{\mathbb{P}})_{i\in I}$};
    
   \path (COP.340) edge[bend left=30] node[right, xshift=0.1cm, text width=6.5cm,align=left] {Theorem~\ref{prop:coherent_admissible}:\\
   $\phi_S=\mathbbm{1}\{\boldsymbol{d}(S)>0\} $ \\ $\phi_J=\sup\{\phi_S: S\subseteq J, S\in 2^{\mathbb{N}_f} \} $} (IFOIT.30);
    \path (IFOIT.150) edge[bend left=30] node[left, xshift=0.5cm, text width=7cm,align=left] {Online closure principle \\ \citep{fischer2024online}:\\
    $\boldsymbol{d}^{\boldsymbol{\phi}}(S)= \inf\{|S\setminus I| : I\subseteq \mathbb{N}, \phi_I=0 \}$} (COP.200);
    
    \path (IFOIT.330) edge[bend left=30] node[right,yshift=0cm, xshift=0.1cm,align=left] {Theorem~\ref{theo:avTest}:\\
    $\psi^I_i=\phi_{I\cap \{1,\ldots,i\}}$} (AVT.20);
    \path (AVT.160) edge[bend left=30] node[left,yshift=0.1cm, text width=8.3cm, xshift=0.1cm,align=left] {Theorem~\ref{theo:av_test_family}:\\ $\phi_I=\inf\{\sup_{i\in I} \psi_i^J:J\cap \{1,\ldots,\sup(I)\}=I\}$ } (IFOIT.210);
    
    \path (AVT.340) edge[bend left=30] node[right, align=left, xshift=0.1cm,yshift=-0.05cm] {Fact~\ref{fact}:\\
    $M_i^{\mathbb{P}}\coloneqq \dfrac{\mathbb{P}(\exists j\in I: \psi_j=1| \mathcal{F}_{i})}{\mathbb{P}(\exists j\in I: \psi_j=1)}$} (TM.20);
    \path (TM.160) edge[bend left=30] node[left, text width=6.5cm, align=left, xshift=0.34cm] {Ville's inequality: \\ $\psi_i=\inf_{\mathbb{P}\in H}\mathrm{sup}_{j\leq i} \mathbbm{1}\left\{M_j^{\mathbb{P}}\geq 1/\alpha\right\}$} (AVT.200);

\end{tikzpicture}
\caption{Illustration of the relation between online true discovery procedures, increasing families of online intersection tests, anytime-valid tests and test martingales. The left path from bottom to top provides a general approach for the construction of coherent online procedures with true discovery guarantee based on test martingales. The right path from top to bottom shows that coherent online procedure implicitly define test martingales. Taken together, the loop shows, in principle, how to improve any given coherent online procedure: we take the right path down, and then take the left path up. If the original procedure is admissible, following the loop  leaves the procedure unchanged.}
\label{fig:flow_theorems}
\end{figure}
\fi

    

\subsection{Coherent online true discovery procedures\label{sec:coherence}}
An important property of multiple testing procedures is \textit{coherence} \citep{G, sonnemann1982allgemeine}. A true discovery procedure $\boldsymbol{d}$ is called \textit{coherent} \citep{goeman2021only}, if for all disjoint $S,U\in 2^{\mathbb{N}_f}$,
\begin{align}
\boldsymbol{d}(S)+\boldsymbol{d}(U)\leq \boldsymbol{d}(S \cup U) \leq \boldsymbol{d}(S) + |U|. \label{eq: coherence}
\end{align}
Coherence ensures consistent decisions or bounds of the multiple testing procedure and is therefore a desirable property. A procedure $\boldsymbol{d}$ is admissible if there is no other procedure  $\boldsymbol{\tilde{d}}$ that uniformly improves $\boldsymbol{d}$, 
where $\boldsymbol{\tilde{d}}$ is said to uniformly improve $\boldsymbol{d}$, if $\boldsymbol{\tilde{d}}\geq \boldsymbol{d}$ and $\mathbb{P}(\boldsymbol{\tilde{d}}(S)> \boldsymbol{d}(S))>0$ for at least one $\mathbb{P} \in \mathcal{P}$ and at least one $S\in 2^{\mathbb{N}_f}$. Equivalently, $\boldsymbol{d}$ is admissible if $\boldsymbol{\tilde{d}}\geq \boldsymbol{d}$ implies $\boldsymbol{\tilde{d}}= \boldsymbol{d}$.

\citet{goeman2021only} showed that in the offline setting, all admissible true discovery procedures must be coherent. 
However, it turns out that this result is not true in the online case. 
\begin{example}
    Consider a setting with only two hypotheses $H_1$ and $H_2$ with independent p-values $P_1$ and $P_2$ that are uniformly distributed under the null hypothesis. Let $\bd(\{1\})=1$, if $P_1\leq \alpha/2$ (and $0$ otherwise), $\bd(\{2\})=1$, if $P_2\leq \alpha/2$ (and $0$ otherwise) and 
    $$
    \bd(\{1,2\})=\begin{cases}
            2, & (P_1\leq \alpha/2 \land P_2 \leq \alpha) \lor (P_2\leq \alpha/2 \land P_1 \leq \alpha) \\
            1, & (P_1\leq \alpha/2 \land P_2 > \alpha) \lor (P_2\leq \alpha/2 \land P_1 > \alpha) \\
            0, & \text{otherwise}.
        \end{cases}
    $$
Then $\bd$ is an incoherent online procedure, since $\bd(\{1,2\})=2$  and $\bd(\{1\})=0$ if $P_2\leq\alpha/2$ and $\alpha/2<P_1\leq\alpha$. However, it is not possible to improve $\bd$ by a coherent online procedure $\tilde{\bd}\geq \bd$. To see this, note that in order to be coherent, $\tilde{\bd}(\{1\})$ needs to equal $1$, if $P_2\leq \alpha/2 \land P_1 \leq \alpha$. Since  $\tilde{\bd}$ is an online procedure, $\tilde{\bd}(\{1\})$ must not use information about $P_2$, so we must have $\tilde{\bd}(\{1\})=1$, if $P_1 \leq \alpha$. However, this implies that for all $\mathbb{P} \in H_1\cap H_2$ such that true discovery guarantee is not provided, $\mathbb{P}(\{\tilde{\bd}(\{1\}) = 0\} \cap \{\tilde{\bd}(\{2\}) =0 \}) \leq \mathbb{P}(\{P_1> \alpha\} \cap \{P_2>\alpha/2 \})= (1-\alpha)(1-\alpha/2)< 1-\alpha.$
\end{example}
Still, it is sensible to focus on coherent procedures as incoherent results are difficult to interpret.

\subsection{Online closed testing\label{sec:online_closed_test}}

The closure principle was originally proposed for FWER control \citep{marcus1976closed}. \ifdiff \added{\citet{genovese2004stochastic,genovese2006exceedance} proposed a similar approach for simultaneous true discovery guarantee, which was later formulated in terms of closed testing by \citet{goeman2011multiple, goeman2019simultaneous}.} \else \citet{genovese2004stochastic,genovese2006exceedance} proposed a similar approach for simultaneous true discovery guarantee, which was later formulated in terms of closed testing by \citet{goeman2011multiple, goeman2019simultaneous}. \fi
For each intersection hypothesis $H_I=\bigcap_{i\in I} H_i$, let $\phi_I$ be an intersection test and $\boldsymbol{\phi}=(\phi_I)_{I\subseteq \mathbb{N}}$ denote the family of intersection tests. Henceforth, it is understood that all tests are $\alpha$-level tests, meaning $\mathbb{P}(\phi_I=1)\leq \alpha$ for all $\mathbb{P}\in H_I$. \citet{genovese2004stochastic,genovese2006exceedance} showed that 
\begin{align}\boldsymbol{d}^{\boldsymbol{\phi}}(S)\coloneqq \inf\{|S\setminus I| : I\subseteq \mathbb{N}, \phi_I=0 \} \qquad (S\in 2^{\mathbb{N}_f})\label{eq:closed_procedure}\end{align} 
provides simultaneous true discovery guarantee over all $S\in 2^{\mathbb{N}_f}$. \ifdiff \added{This follows  from the fact that $\boldsymbol{d}^{\boldsymbol{\phi}}(S)>|S\cap I_1^{\mathbb{P}}|$ for any $S\in \nf$ implies $\phi_{I_0^{\mathbb{P}}}=1$, which happens with probability at most $\alpha$.} \else 
This follows  from the fact that $\boldsymbol{d}^{\boldsymbol{\phi}}(S)>|S\cap I_1^{\mathbb{P}}|$ for any $S\in \nf$ implies $\phi_{I_0^{\mathbb{P}}}=1$, which happens with probability at most $\alpha$. \fi

\ifdiff \added{Since we adopt the closed testing framework by \citet{goeman2011multiple, goeman2019simultaneous}, we also refer to $\boldsymbol{d}^{\boldsymbol{\phi}}$ as closed procedure.} \else Since we adopt the closed testing framework by \citet{goeman2011multiple, goeman2019simultaneous}, we also refer to $\boldsymbol{d}^{\boldsymbol{\phi}}$ as closed procedure. \fi
Technically, the aforementioned works only considered finitely many hypotheses, but the method and its guarantees extend to a countable number of hypotheses \citep{fischer2024online}, and so we present that version for easier connection to the online setting. 

\citet{goeman2021only} proved an important result: every (offline) coherent true discovery procedure is equivalent to or uniformly improved by a closed procedure of the form \eqref{eq:closed_procedure}. Therefore, the closure principle allows to construct and analyze all admissible true discovery procedures based on single tests for the intersection hypotheses which are usually much easier to handle.
\citet{fischer2024online} showed that the closed procedure $\boldsymbol{d}^{\boldsymbol{\phi}}$ is an online procedure, if the following assumptions are fulfilled:
\begin{enumerate}
    \item Every $\phi_I$, $I\subseteq \mathbb{N}$ is an \textit{online intersection test}, meaning $\phi_I$ is measurable with respect to $\mathcal{F}_{\sup (I)}$. \label{bull:online}
    \item The family of intersection tests $\boldsymbol{\phi}=(\phi_I)_{I\subseteq \mathbb{N}}$ is \textit{increasing}\footnote{ \citet{fischer2024online} used the term \textit{predictable} for~\eqref{eq:predictability}, but we use \textit{increasing} in order to avoid confusion with the measure-theoretic definition of predictability.}: for all $i\in \mathbb{N}$ and $I\subseteq \{1,\ldots, i\}$,
   \begin{align}\phi_I \leq \phi_{I\cup K} \text{ for all } K\subseteq \{k\in \mathbb{N}: k>i\}.\label{eq:predictability}\end{align}\label{bull:predictability}
\end{enumerate}

\ifdiff \added{This follows from the fact that the closed procedure $\bd^{\bphi}$ for increasing $\bphi$ can be written as}
\begin{align}
    \boldsymbol{d}^{\boldsymbol{\phi}}(S)\coloneqq \inf\{|S\setminus I| : I\subseteq \{1,\ldots, \max(S)\}, \phi_I=0 \} \qquad (S\in 2^{\mathbb{N}_f}). \label{eq:closed_procedure_increasing}
\end{align}
 \else 
This follows from the fact that the closed procedure $\bd^{\bphi}$ for increasing $\bphi$ can be written as
\begin{align}
    \boldsymbol{d}^{\boldsymbol{\phi}}(S)\coloneqq \inf\{|S\setminus I| : I\subseteq \{1,\ldots, \max(S)\}, \phi_I=0 \} \qquad (S\in 2^{\mathbb{N}_f}). \label{eq:closed_procedure_increasing}
\end{align}
\fi

A closed procedure $\bd^{\boldsymbol{\phi}}$, where $\boldsymbol{\phi}$ satisfies the conditions~\ref{bull:online} and~\ref{bull:predictability}, is called an \textit{online closed procedure}. Every online closed procedure is a coherent online procedure, which follows immediately from the same result in the offline case \citep{goeman2021only}.
 \citet{fischer2024online} proved that all online procedures with FWER control can be written as a closed procedure where the intersection tests satisfy these two conditions. Hence we know that the closure principle is admissible for offline true discovery control \citep{goeman2021only} and online FWER control \citep{fischer2024online}, but not yet for coherent online true discovery control. We close this gap by showing that any coherent online procedure can be recovered or improved by an online closed procedure.
 
\begin{theorem}\label{prop:coherent_admissible}
   Let $\bd$ be a coherent online true discovery procedure. Define \begin{align}\phi_S=\mathbbm{1}\{\boldsymbol{d}(S)>0\} \ \forall   S\in 2^{\mathbb{N}_f} \quad \text{and} \quad  \phi_J=\sup\{\phi_S: S\subseteq J, S\in 2^{\mathbb{N}_f} \} \ \forall  J\in2^{\mathbb{N}_{-f}}.\label{eq:closed_general}\end{align}
   Then $\bphi=(\phi_I)_{I\subseteq \mathbb{N}}$ is an increasing family of online intersection tests and $\bd^{\bphi}\geq \bd$.
\end{theorem}
\begin{proof}
      The simultaneous true discovery guarantee of $\bd$ implies that the $\phi_I$ define intersection tests. \ifdiff \added{Furthermore, since $\bd$ is coherent, it follows that $\phi_S \leq \phi_I$ if $S\subseteq I$, which particularly implies that $\bphi$ is increasing. In addition, } \else Furthermore, since $\bd$ is coherent, it follows that $\phi_S \leq \phi_I$ if $S\subseteq I$, which particularly implies that $\bphi$ is increasing. In addition, \fi $\phi_S$ is an online intersection test since $\bd$ is an online procedure. Therefore, it only remains to show that $\bd(S)\leq \bd^{\bphi}(S)$ for all $S\in 2^{\mathbb{N}_f}$. Suppose $\bd(S)=s$. Due to the coherence of $\bd$, we have $\phi_J=1$ for all $J\subseteq S$ with $|J|> |S|-s$. The coherence further implies that $\phi_I=1$ for all $I\subseteq \mathbb{N}$ with $|S\cap I|> |S|-s$ and hence $\bd^{\bphi}(S)\geq s$. 
\end{proof}

With this result, we can focus on online closed procedures when considering coherent online true discovery guarantees. In the following subsection, we introduce our main result and show that to construct online closed procedures, it is necessary to define anytime-valid tests and therefore one must rely on test martingales. This yields the connected graph in Figure~\ref{fig:flow_theorems} and motivates the construction of e-value based online closed procedures, which we will focus on afterwards. 

Note that Theorem~\ref{prop:coherent_admissible} particularly holds in the offline case $\mathcal{F}_1=\mathcal{F}_2=\ldots$ and therefore immediately yields the aforementioned result by \citet{goeman2021only} as a corollary.

\begin{remark}
    There are cases where the closed procedure based on the intersection tests defined in \eqref{eq:closed_general} dominates the original procedure. For example, consider three hypotheses and the coherent procedure $\bd$ with $\bd(\{1,2,3\})=1$,  $\bd(\{1,2\})=1$, $\bd(\{1,3\})=1$, $\bd(\{2,3\})=1$ and $\bd(\{i\})=0$ for $i\in \{1,2,3\}$. The corresponding closed procedure would give the same bounds except for further concluding that $\bd(\{1,2,3\})=2$. This makes sense, since if there is at least one true discovery in $\{1,2\}$, one in $\{1,3\}$ and one in $\{2,3\}$, there should be at least two true discoveries in $\{1,2,3\}$. 
\end{remark}

\subsection{Admissible coherent online procedures for true discovery guarantees must rely on test martingales\label{sec:test_martingale}}

\ifdiff \added{Anytime-valid tests are used to sequentially test a single null hypothesis with increasing data. A famous example is the sequential probability ratio test (SPRT) by \citet{wald1945sequential} for simple hypotheses, which rejects the null hypothesis as soon as the likelihood ratio of the alternative distribution against the null distribution is above some predefined threshold. In recent years, anytime-valid tests have received renewed attention with a particular focus on general composite hypotheses \citep{ramdas2023game}. For example, \citet{ramdas2020admissible} have shown that every anytime-valid test has a similar form to the aforementioned SPRT, in which the likelihood ratio is replaced by the more general concept of e-processes, which in turn are generalizations of test martingales.} 
\else Anytime-valid tests are used to sequentially test a single null hypothesis with increasing data. A famous example is the sequential probability ratio test (SPRT) by \citet{wald1945sequential} for simple hypotheses, which rejects the null hypothesis as soon as the likelihood ratio of the alternative distribution against the null distribution is above some predefined threshold. In recent years, anytime-valid tests have received renewed attention with a particular focus on general composite hypotheses \citep{ramdas2023game}. For example, \citet{ramdas2020admissible} have shown that every anytime-valid test has a similar form to the aforementioned SPRT, in which the likelihood ratio is replaced by the more general concept of e-processes, which in turn are based on test martingales. \fi 
In this section, we make connections between increasing families of online intersection tests and anytime-valid tests, which then reveal close relations of online procedures with true discovery guarantee and test martingales.

Suppose we have an anytime-valid test $(\psi^I_i)_{i\in I}$ for each intersection hypothesis $H_I$. Following~\cite{ramdas2023game},  $(\psi^I_i)_{i\in I}$ is an anytime-valid test for $H_I$, if $\psi^I_{i}$ is measurable with respect to $\mathcal{F}_i$ and we have $\mathbb{P}(\exists i\in I: \psi^I_i=1)\leq \alpha$ for all $\mathbb{P}\in H_I$.  The following theorem shows how we can use such anytime-valid tests to construct an increasing family of online  intersection tests.

\begin{theorem}\label{theo:av_test_family}
    For all $I\subseteq \mathbb{N}$ let $(\psi_i^I)_{i\in I}$ be an anytime-valid test for $H_I$. Then $\bphi=(\phi_I)_{I\subseteq \mathbb{N}}$, where
    \begin{align}
    \phi_I=\inf\left\{\sup_{i\in I} \psi_i^J:J\cap \{1,\ldots,\sup(I)\}=I\right\}, \label{eq:av_test_family}
    \end{align}
    is an increasing family of online intersection tests. In particular, for infinite $I\subseteq \mathbb{N}$, $\phi_I=\sup_{i\in I} \psi_i^I$.
\end{theorem}
\begin{proof}
    Since $(\psi_i^I)_{i\in I}$, $I\subseteq \mathbb{N}$, is an anytime-valid test for $H_I$, it immediately follows that $\phi_I$ is an online intersection test. Furthermore, for $i\in \mathbb{N}$, $I\subseteq \{1,\ldots,i\}$ and $K\subseteq \{k\in \mathbb{N}:k>i\}$ it holds that $J\cap \{1,\ldots,\sup(I\cup K)\}=I\cup K$ implies that $J\cap \{1,\ldots,\sup(I)\}=I$. Therefore,
    \begin{align*}
        \phi_I&=\inf\left\{\sup_{i\in I} \psi_i^J:J\cap \{1,\ldots,\sup(I)\}=I\right\} \\
        &\leq \inf\left\{\sup_{i\in I} \psi_i^J:J\cap \{1,\ldots,\sup(I\cup K)\}=I\cup K\right\} \\
        &\leq \inf\left\{\sup_{i\in I\cup K} \psi_i^J:J\cap \{1,\ldots,\sup(I\cup K)\}=I\cup K\right\} = \phi_{I\cup K},
    \end{align*}
    showing that $\bphi$ is increasing. 
\end{proof}

In the following theorem, we show a converse relationship, meaning that increasing families of online intersection tests implicitly define anytime-valid tests for the intersection hypotheses. In addition, we prove that every increasing family of online  intersection tests can be constructed by anytime-valid tests using \eqref{eq:av_test_family}.  

\begin{theorem}
    \label{theo:avTest}
    Let $\boldsymbol{\phi}=(\phi_I)_{I\subseteq \mathbb{N}}$ be an increasing family of online intersection tests.  Define \begin{align}\psi^I_i=\phi_{I\cap \{1,\ldots,i\}}.\label{eq:anytime_from_fam}\end{align} Then $(\psi^I_i)_{i\in I}$ is an anytime-valid test for $H_I$ for all $I\subseteq \mathbb{N}$. Furthermore, let $\tilde{\bphi}=(\tilde{\phi}_I)_{I\subseteq \mathbb{N}}$ be defined by $(\psi^I_i)_{i\in I}$, $I\subseteq \mathbb{N}$, using \eqref{eq:av_test_family}. Then $\tilde{\phi}_S=\phi_S$ for all $S\in \nf$ and $\bd^{\tilde{\bphi}}=\bd^{\bphi}$. 
    
\end{theorem}
\begin{proof}
    Since  $\phi_{I\cap \{1,\ldots,i\}}$ is an online intersection test, we have that $\psi^I_i$ is measurable with respect to $\mathcal{F}_i$. Furthermore, since $\boldsymbol{\phi}$ is increasing, it holds that
    $\mathbb{P}(\exists i\in I: \psi^I_i =1)= \mathbb{P}(\phi_I =1)\leq \alpha$ for all $\mathbb{P}\in H_I$. Hence, $(\psi^I_i)_{i\in I}$ is an anytime-valid test for $H_I$ with respect to the filtration $(\mathcal{F}_i)_{i\in I}$. Now let $\tilde{\bphi}$ be defined by \eqref{eq:av_test_family}. Then for all $S\in \nf$,
    \begin{align*}
        \tilde{\phi}_S&=\inf\left\{\sup_{i\in S} \psi_i^J:J\cap \{1,\ldots,\sup(S)\}=S\right\} \\
        &=\inf\left\{\sup_{i\in S} \phi_{J\cap \{1,\ldots,i\} }:J\cap \{1,\ldots,\sup(S)\}=S\right\} \\
        &=\sup_{i\in S} \phi_{S\cap \{1,\ldots,i\} }=\phi_S,
    \end{align*}
    where the last equality follows since $\bphi$ is increasing. Note that this implies $\bd^{\tilde{\bphi}}=\bd^{\bphi}$, since $\tilde{\bphi}$ and $\bphi$ are increasing and therefore $\bd^{\tilde{\bphi}}(S)$ and $\bd^{\bphi}(S)$, $S\in \nf$, can be determined solely based on $\tilde{\phi}_I$ and $\phi_I$, respectively, with $I\subseteq S$.
\end{proof}

Together, Theorems~\ref{prop:coherent_admissible} and~\ref{theo:avTest} imply that in order to define admissible coherent online true discovery procedures, we need to construct anytime-valid tests for the intersection hypotheses. In addition, they imply that every coherent online true discovery procedure implicitly defines nontrivial anytime-valid tests by \eqref{eq:closed_general} and \eqref{eq:anytime_from_fam} for each intersection hypothesis. These results are particularly interesting, since \citet{ramdas2020admissible} gave a precise characterization of anytime-valid tests, proving that every anytime-valid test can be reconstructed or uniformly improved using test martingales. To state their result more precisely, we need to introduce some terminology.




An anytime-valid test $(\tilde{\psi}_i)_{i\in I}$ uniformly improves $(\psi_i)_{i\in I}$, if $\tilde{\psi}_i\geq \psi_i$ for all $i\in I$ and $\mathbb{P}(\tilde{\psi}_i> \psi_i)>0$ for some $i\in I$ and $\mathbb{P}\in \mathcal{P}$. Furthermore, a nonnegative process $(M_i)_{i\in I \cup \{0\}}$ adapted to the filtration $(\mathcal{F}_i)_{i\in I \cup \{0\}}$, where $\mathcal{F}_{0}=\emptyset$, is a test (super)martingale for $\mathbb{P}$, if $M_0\stackrel{(\leq)}{=}1$ and $\mathbb{E}_{\mathbb{P}}[M_i|\mathcal{F}_{i-}]\stackrel{(\leq)}{=} M_{i-}$ for all $i\in I$, where $i-=\max\{j\in I\cup\{0\}:j<i\}$. We call $(M_i)_{i\in I\cup\{0\}}$ a test (super)martingale for a null hypothesis $H_I$, if the above holds for all $\mathbb{P}\in H_I$. In the following we formally state the result by \citet{ramdas2020admissible} (adapted to our setup) and provide a self-contained proof.

\begin{fact}\label{fact}
    Let $(\psi_i)_{i\in I}$ be an anytime-valid test for $H_I$, $I\subseteq \mathbb{N}$, and define for all $\mathbb{P}\in H_I$:
    $$
        M_i^{\mathbb{P}}\coloneqq \frac{\mathbb{P}(\exists j\in I: \psi_j=1| \mathcal{F}_{i})}{\mathbb{P}(\exists j\in I: \psi_j=1)} \quad (i\in I\cup \{0\})
    $$
    with $M_i^{\mathbb{P}}=1$ if the denominator equals zero.
    Then $(M_i^{\mathbb{P}})_{i\in I\cup \{0\}}$ is a test martingale for $\mathbb{P}$. Furthermore, let 
    $$
    \tilde{\psi}_i=\inf_{\mathbb{P}\in H_I} \sup_{j\leq i} \mathbbm{1}\{M_j^{\mathbb{P}}\geq 1/\alpha\} \quad (i\in I).
    $$
    Then $(\tilde{\psi}_i)_{i\in I}$ is an anytime-valid test for $H_I$ and either equals or uniformly improves $(\psi_i)_{i\in I}$.
\end{fact}
\begin{proof}
The tower property for conditional expectations implies that $(M_i^{\mathbb{P}})_{i\in I\cup \{0\}}$ is a test martingale for $\mathbb{P}$ and Ville's inequality shows that $(\tilde{\psi}_i)_{i\in I}$ is an anytime-valid test for $H_I$. Furthermore, $\psi_i=1$ implies that $\mathbb{P}(\exists j\in I: \psi_j=1| \mathcal{F}_{i})=1$ and therefore $M_i^{\mathbb{P}}\geq 1/\alpha$ for all $\mathbb{P}\in H_I$.
\end{proof}

Fact~\ref{fact} closes the loop illustrated in Figure~\ref{fig:flow_theorems}. This particularly shows that every coherent online true discovery procedure $\bd$ implicitly constructs test martingales for the intersection hypotheses. Furthermore, taking the left path in Figure~\ref{fig:flow_theorems} with these test martingales always yields a procedure which either equals or uniformly improves the original procedure $\bd$.

\section{Online true discovery guarantee with sequential e-values\label{sec:seq_e-value}}

In the previous section we showed that we need to construct anytime-valid tests, and thus test martingales, for the intersection hypotheses when constructing  coherent online procedures with a true discovery guarantee. This general martingale-based approach is illustrated in the left path from bottom to top in Figure~\ref{fig:flow_theorems}. Since every step involves taking the infimum over a large set, it seems computationally inefficient. However, in practice one can avoid this by using the same test martingale for all $\mathbb{P}\in H$, the same anytime-valid test $(\psi_i^J)_{i\in I}$ for all $J$ with $J\cap \{1,\ldots,\sup(I)\}=I$ and constructing $\bphi$ in a way that permits a short-cut for $\bd^{\bphi}$. In this section, we apply all this to derive a computationally efficient and powerful online true discovery procedure based on sequential e-values. We first introduce our general \texttt{SeqE-Guard} algorithm (Section~\ref{sec:seqE-Guard}), then show that it allows to uniformly improve all existing online true discovery procedures by using specific sequential e-values (Section~\ref{sec:katse}) and finally analyze its behavior when using GRO e-values (Section~\ref{sec:gro}).

\subsection{The \texttt{SeqE-Guard} algorithm\label{sec:seqE-Guard}}

Let $(M_t)_{t\in \mathbb{N}_0}$ be a test supermartingale with respect to $(\mathcal{F}_t)_{t\in \mathbb{N}_0}$ for some hypothesis $H$. We can break down $(M_t)_{t\in \mathbb{N}_0}$ into its individual factors $E_t=\frac{M_t}{M_{t-1}}$, $t>0$, with the convention $0/0=0$. Due to the supermartingale property,  $E_t$ is nonnegative, measurable with respect to $\mathcal{F}_t$ and \begin{align}\mathbb{E}_{\mathbb{P}}[E_t|\mathcal{F}_{t-1}]=\mathbb{E}_{\mathbb{P}}[M_t/M_{t-1}|\mathcal{F}_{t-1}]=\mathbb{E}_{\mathbb{P}}[M_t|\mathcal{F}_{t-1}]/M_{t-1} \leq 1 \quad (\mathbb{P}\in H). \label{eq:testmartingale}\end{align} 
Hence, each of these random variables $E_t$ is an e-value for $H$ conditional on the past, sometimes called a sequential e-value in the literature \citep{vovk2021values}. Thus, every test supermartingale can be written as the product of sequential e-values. Conversely, every product of sequential e-values defines a test supermartingale (just multiply $M_{t-1}$ on both sides of \eqref{eq:testmartingale}). For these reasons, sequential e-values are potentially the perfect tool to define online closed procedures, which is why we will focus on them in the following.

Assume that for each hypothesis $H_t$ an e-value $E_t$ is available and the e-values $(E_t)_{t\in \mathbb{N}}$ are sequentially valid with respect to $(\mathcal{F}_t)_{t\in \mathbb{N}_0}$,  meaning $E_t\in \mathcal{F}_t$ and $\Ex_{\mathbb{P}}[E_t|\mathcal{F}_{t-1}]\leq 1$ for all $\mathbb{P}\in H_t$. For each intersection hypothesis $H_I$, we construct a process $(W_I^t)_{t\in I\cup \{0\}}$ with $W_I^0=1$ and 
\begin{align}
W_I^t=\prod_{i\in I\cap \{1,\ldots, t\}} E_i \qquad (t\in I). \label{eq:seq_e}
\end{align}
Following the above argumentation, $(W_I^t)_{t\in I\cup \{0\}}$ is a test supermartingale for $H_I$ with respect to $(\mathcal{F}_t)_{t\in I\cup \{0\}}$. By Ville's inequality~\citep{howard2020time}, it follows that \begin{align}\phi_I=\mathbbm{1}\{\exists t\in I: W_I^t\geq 1/\alpha\}=\mathbbm{1}\left\{\sup_{t\in I} W_I^t \geq 1/\alpha\right\}\label{eq:intersection_test_mart}\end{align} is an intersection test. Furthermore, $\phi_I$ is an online intersection test and $\boldsymbol{\phi}=(\phi_I)_{I\subseteq \mathbb{N}}$ is increasing such that the closed procedure $\bd^{\bphi}$  \eqref{eq:closed_procedure} is indeed an online procedure (see Section~\ref{sec:online_closed_test}). Note that due to the supremum involved in Ville's inequality, the intersection tests $\phi_I$, $I\subseteq \mathbb{N}$, are not symmetric, and thus very different from the ones considered by \citet{tian2023large}.  

Sequential e-values arise naturally in a variety of settings \citep{shafer2021testing, grunwald2020safe, ramdas2023game, vovk2024merging}. For example, note that many online multiple testing tasks are performed in an adaptive manner. That means, the hypotheses to test, the design of a study or the strategy to calculate an e-value, depend on the data observed so far. This is very natural, since the hypotheses are tested over time and therefore the statistician performing these tests automatically learns about the context of the study and the true distribution of the data during the testing process. Indeed, to avoid such data-adaptive designs one would need to ignore all the previous data when making design decisions for the future testing process and therefore could be required to prespecify all hypotheses that are going to be tested, the sample size for these tests, decide which tests to apply and fix many further design parameters at the very beginning of the testing process. This would take away much of the flexibility of online multiple testing procedures. However, if one uses past information for the design of an e-value and calculates the e-value on the same data that design information is based on, one would need to have knowledge about the conditional distribution of the e-value given that design information. This is often not available. Therefore, one usually uses independent and fresh data for each of the e-values. In this way, no matter in which way an e-value depends on the past data, it remains valid conditional on it. Consequently, we obtain sequential e-values well suited to our online true discovery procedure.

\ifarxiv
\paragraph{Computational shortcuts.}
\else 
\textbf{Computational shortcuts.}\hspace{0.3cm }
\fi
One problem with procedures based on the (online) closure principle is that at each step $t\in \mathbb{N}$ up to $2^{t-1}$ additional intersection tests must be considered. However, for specific intersection tests this number can be  reduced drastically. In Algorithm~\ref{alg:general} we introduce a short-cut for the intersection tests in \eqref{eq:intersection_test_mart} that only requires one calculation per individual hypothesis but provides the same bounds as the entire closed procedure would. We call this Algorithm \texttt{SeqE-Guard}  (guarantee for true discoveries with sequential e-values).

The \texttt{SeqE-Guard} algorithm will provide a sequence of lower bounds $d_t$ on the number of true discoveries in an adaptively chosen sequence of query sets $S_t \subseteq \{1,\dots,t\}$, simultaneously over all $t$. 
The $S_t$ and $d_t$ sequences are (nonstrictly) increasing with $t$ since by default the statistician only decides at step $t$ whether to include the index $t$ within $S_t$ or not (based on $E_1,\dots,E_t$), but typically does not omit hypotheses that were deemed interesting at an earlier point (see Remark below). So it is understood below that $S_t \supseteq S_{t-1}$ for all $t$. As the multiple testing process continues, the statistician can report/announce one or more $(S_t,d_t)$ pairs that they deem interesting thus far, and our theorem below guarantees that all such announcements will be accurate with high probability.

\ifarxiv
\paragraph{A short description of \texttt{SeqE-Guard}.}
\else 
\textbf{A short description of \texttt{SeqE-Guard}.}\hspace{0.3cm }
\fi
At each step $t\in \mathbb{N}$, the statistician  decides based on the e-values $E_1,\ldots, E_t$ whether the index $t$ should be included in the query set $S_t$. If it is to be included, \texttt{SeqE-Guard} calculates the product of all e-values in $S_t$ (that were not already excluded) and all e-values smaller than $1$ that are not in $S_t$. If that product is larger than or equal to $1/\alpha$, we increase the true discovery bound by $1$ and \emph{exclude} the current largest e-value from the future analysis.  
In Supplementary Material~\ref{sec:example} we provide a simple calculation example for the  \texttt{SeqE-Guard} algorithm.

\begin{theorem}\label{theo:general_alg}
    Let $E_1, E_2,\ldots$ be sequential e-values for $H_1,H_2,\ldots$ . The true discovery bounds $d_t$ for $S_t$ defined by \texttt{SeqE-Guard} (Algorithm~\ref{alg:general}) are the same as the ones obtained by the online closure principle with the intersection tests in \eqref{eq:intersection_test_mart}.
    In particular, $S_t$ and $d_t$ satisfy for all $\mathbb{P}\in \mathcal{P}$: $\mathbb{P}(d_t\leq |S_t\cap I_1^{\mathbb{P}}| \text{ for all } t\in \mathbb{N})\geq 1-\alpha$. 
\end{theorem}
 \ifdiff \added{The main idea for the proof of Theorem~\ref{theo:general_alg} is that the bound $\bd^{\bphi}(S_t)$ is given by $|S_t\setminus I|$, where $I\subseteq \{1,\ldots,t\}$ is the index set that contains the most possible indices of $S_t$ while guaranteeing that $W_I^s<1/\alpha$ for all $s\leq t$ (see \eqref{eq:closed_procedure_increasing}). The \texttt{SeqE-guard} algorithm finds this set $I$ by including all indices $i\in \{1,\ldots,t\}\setminus S_t$ with $E_i<1$ (since those e-values reduce the product) and excluding the index $i\in I\cap S_t$ of the largest e-value in case of $W_I^t\geq 1/\alpha$. Since the full proof is quite long, we have deferred it to Supplementary Material~\ref{sec:proofs}.} \else 
 The main idea for the proof of Theorem~\ref{theo:general_alg} is that the bound $\bd^{\bphi}(S_t)$ is given by $|S_t\setminus I|$, where $I\subseteq \{1,\ldots,t\}$ is the index set that contains the most possible indices of $S_t$ while guaranteeing that $W_I^s<1/\alpha$ for all $s\leq t$ (see \eqref{eq:closed_procedure_increasing}). The \texttt{SeqE-guard} algorithm finds this set $I$ by including all indices $i\in \{1,\ldots,t\}\setminus S_t$ with $E_i<1$ (since those e-values reduce the product) and excluding the index $i\in I\cap S_t$ of the largest e-value in case of $W_I^t\geq 1/\alpha$. Since the full proof is quite long, we have deferred it to Supplementary Material~\ref{sec:proofs}. \fi

Note that closed procedures provide simultaneous true discovery guarantee simultaneously over all $S\in \nf$, while \texttt{SeqE-Guard} only gives a simultaneous lower bound on the number of true discoveries for a path of query sets $(S_t)_{t\in \mathbb{N}}$ with $S_1\subseteq S_2 \subseteq \ldots$ . However, at some time $t\in \mathbb{N}$ one could also obtain a lower bound for the number of true discoveries in any other set $S\subseteq \{1,\ldots,t\}$ that is not on the path. 
\ifdiff  \added{In particular, Theorem~\ref{theo:general_alg} shows that the bound $\bd^{\bphi}(S)$, with $\bphi$ from \eqref{eq:intersection_test_mart}, can be calculated in time linear in $\max(S)$ for any $S$ by setting $S_t=S\cap \{1,\ldots,t\}$.} \else In particular, Theorem~\ref{theo:general_alg} shows that the bound $\bd^{\bphi}(S)$, with $\bphi$ from \eqref{eq:intersection_test_mart}, can be calculated in time linear in $\max(S)$ for any $S$ by setting $S_t=S\cap \{1,\ldots,t\}$. \fi

We formulated \texttt{SeqE-Guard} for single query paths due to computational convenience and as we think this reflects the proceeding in many applications. This was also done in previous works on online true discovery guarantee \citep{katsevich2020simultaneous, iqraa2024false}. \ifdiff \added{In the following corollary, we capture the guarantee of \texttt{SeqE-Guard} for multiple query paths.} \else In the following corollary, we capture the guarantee of \texttt{SeqE-Guard} for multiple query paths. \fi 

\ifdiff
\begin{corol}\label{corol:multiple_paths}
    \added{Let $(d_t^{(1)})_{t\in \mathbb{N}}, (d_t^{(2)})_{t\in \mathbb{N}},\ldots$ be obtained by applying \texttt{SeqE-Guard} to multiple query paths $(S_t^{(1)})_{t\in \mathbb{N}},(S_t^{(2)})_{t\in \mathbb{N}},\ldots$ . Then, it holds for all $\mathbb{P}\in \mathcal{P}$: $\mathbb{P}(d_t^{(m)} \leq |S_t^{(m)}\cap I_1^{\mathbb{P}}| \text{ for all } t,m\in \mathbb{N})\geq 1-\alpha$.}
\end{corol}
\else
\begin{corol}\label{corol:multiple_paths}
    Let $(d_t^{(1)})_{t\in \mathbb{N}}, (d_t^{(2)})_{t\in \mathbb{N}},\ldots$ be obtained by applying \texttt{SeqE-Guard} to multiple query paths $(S_t^{(1)})_{t\in \mathbb{N}},(S_t^{(2)})_{t\in \mathbb{N}},\ldots$ . Then, it holds for all $\mathbb{P}\in \mathcal{P}$: $\mathbb{P}(d_t^{(m)} \leq |S_t^{(m)}\cap I_1^{\mathbb{P}}| \text{ for all } t,m\in \mathbb{N})\geq 1-\alpha$.
\end{corol}
\fi

It should be noted that if we apply \texttt{SeqE-Guard} to multiple query paths, we must use the same e-values for every query path we consider.



\ifarxiv
\begin{algorithm}[htbp]
\caption{\texttt{SeqE-Guard}: Online true discovery guarantee with sequential e-values} \label{alg:general}
 \hspace*{\algorithmicindent} \textbf{Input:} Sequence of sequential e-values $E_1,E_2, \ldots$ .\\
 \hspace*{\algorithmicindent} 
 \textbf{Output:} Query sets $S_1\subseteq S_2 \subseteq \ldots$ and true discovery bounds $d_1\leq d_2\leq \ldots$ .
\begin{algorithmic}[1]
\State $d_0=0$
\State $S_0=\emptyset$
\State $U=\emptyset$
\State $A=\emptyset$
\For{$t=1,2,\ldots$}
\State $S_t=S_{t-1}$
\State $d_{t}=d_{t-1}$
\State Statistician observes $E_t$ and chooses whether index $t$ should be included in $S_t$.  
\If{$t\in S_t$}
\State $A=A\cup \{t\}$
\If{$\prod_{i \in A\cup U} E_i\geq 1/\alpha$}
\State $d_{t}=d_{t-1}+1$
\State $A=A\setminus \{\text{index of largest e-value in } A\}$
\EndIf
\ElsIf{$E_t<1$}
\State $U=U\cup \{t\}$
\EndIf
\State \Return $S_t, d_t$
\EndFor
\end{algorithmic}
\end{algorithm}
\else 
\begin{algorithm}
\caption{\texttt{SeqE-Guard}: Online true discovery guarantee with sequential e-values} \label{alg:general}
  \begin{flushleft} \textbf{Input:} Sequence of sequential e-values $E_1,E_2, \ldots$ .\\
 \textbf{Output:} Query sets $S_1\subseteq S_2 \subseteq \ldots$ and true discovery bounds $d_1\leq d_2\leq \ldots$ .
 \end{flushleft}
\begin{algorithmic}[1]
\State $d_0=0$
\State $S_0=\emptyset$
\State $U=\emptyset$
\State $A=\emptyset$
\For{$t=1,2,\ldots$}
\State $S_t=S_{t-1}$
\State $d_{t}=d_{t-1}$
\State Statistician observes $E_t$ and chooses whether index $t$ should be included in $S_t$.  
\If{$t\in S_t$}
\State $A=A\cup \{t\}$
\If{$\prod_{i \in A\cup U} E_i\geq 1/\alpha$}
\State $d_{t}=d_{t-1}+1$
\State $A=A\setminus \{\text{index of largest e-value in } A\}$
\EndIf
\ElsIf{$E_t<1$}
\State $U=U\cup \{t\}$
\EndIf
\State \Return $S_t, d_t$
\EndFor
\end{algorithmic}
\end{algorithm}
\fi

\begin{example}[Nontrivial true discovery bounds with weak signals]
    Let $\alpha=0.05$ and suppose we test $1000$ hypotheses using sequential e-values. Further suppose that $300$ of those e-values equal $0.9$ and $700$ of those e-values equal $1.1$. Due to the small e-values, it seems like there is no or at most very small evidence against the null hypotheses. Indeed, standard multiple testing procedures like the e-BH procedure for FDR control \citep{wang2022false} would not make any rejections, \ifdiff \added{ regardless of the significance level $\alpha\in (0,1]$. } \else regardless of the significance level $\alpha\in (0,1]$. \fi  However, plugging these e-values into the \texttt{SeqE-Guard} algorithm yields a bound of $d_{1000}\geq 306$ for the query set $S_{1000}=\{i\leq 1000:E_i>1\}$, since $0.9^{300} 1.1^{395}\geq 20$. Consequently, the \texttt{SeqE-Guard} algorithm can claim with a probability of $0.95$ that there are at least $306$ false hypotheses among the $700$ hypotheses with e-value greater than $1$. This demonstrates the capability of the \texttt{SeqE-Guard} algorithm to detect small but frequent signals in the data.   
\end{example}

\subsection{Improving existing methods\label{sec:katse}}

Existing online procedures with simultaneous true discovery guarantees were proposed by \citet{katsevich2020simultaneous} and \citet{iqraa2024false}. Their setting involved observing one p-value for each hypothesis (as is standard in multiple testing) but in this section, we will show that these methods can be uniformly improved by \texttt{SeqE-Guard} by employing specific choices of the sequential e-values. Furthermore, we show that \texttt{SeqE-Guard} even uniformly improves some existing offline procedures.

\subsubsection{Improving existing online procedures\label{sec:improving_online}}
We start with the \texttt{online-simple} procedure by \citet{katsevich2020simultaneous}. Suppose that p-values $P_1,P_2,\ldots$ for the hypotheses $H_1,H_2,\ldots$ are available such that $P_i$ is measurable with respect to $\mathcal{F}_{i}$ and let $\alpha_1,\alpha_2, \ldots $ be nonnegative thresholds such that $\alpha_i$ is measurable with respect to $\mathcal{F}_{i-1}$.  It is assumed that the null p-values are valid conditional on the past, meaning $\mathbb{P}(P_i\leq x|\mathcal{F}_{i-1})\leq x$ for all $i\in I_0^{\mathbb{P}}$, $x\in [0,1]$. \citet{katsevich2020simultaneous} showed that \begin{align}\bd^{\text{os}}(S_t)=\left\lceil-ca+ \sum_{i=1}^t \mathbbm{1}\{P_i\leq \alpha_i\} - c\alpha_i\right\rceil \quad (t\in \mathbb{N}) \label{eq:online-simple}\end{align}
provides simultaneous true discovery guarantee over all sets $S_t=\{i\leq t: P_i\leq \alpha_i\}$, $t\in \mathbb{N}$, where $a>0$ is some parameter and $c=\frac{\log(1/\alpha)}{a\log(1+\log(1/\alpha)/a)}$. On closer examination of their proof, one can observe that they proved their guarantee by implicitly showing that \begin{align}E_i^{\text{os}}=\exp[\theta_c(\mathbbm{1}\{P_i\leq \alpha_i\} - c\alpha_i)] \quad (i\in \mathbb{N)} \label{eq:e-val_os} \end{align} define sequential e-values, where $\theta_c=\log(1/\alpha)/(ca)$. We now propose to simply plug these sequential e-values into \texttt{SeqE-Guard}; this leads to Algorithm~\ref{alg:online-simple}, which we will refer to as \texttt{closed online-simple} procedure in the following. To see this, note that
$$
\prod_{i\in I} E_i^{\text{os}}\geq 1/\alpha \Leftrightarrow  \sum_{i\in I}  \mathbbm{1}\{P_i\leq \alpha_i\} - c\alpha_i \geq \log(1/\alpha)/\theta_c = ca \quad (I\subseteq \mathbb{N}).
$$
 The superscript of the set $A_t^c$ in Algorithm~\ref{alg:online-simple} is used to reflect the fact that $A_t^c$ contains all indices of queried hypotheses that were excluded from the analysis until step $t$, and therefore can be seen as complement of the set $A$ in Algorithm~\ref{alg:general} with respect to $S_t$.

\ifarxiv
\begin{algorithm}
\caption{\texttt{Closed online-simple}} \label{alg:online-simple}
\hspace*{\algorithmicindent}\textbf{Input:} Sequence of p-values $P_1,P_2,\ldots$ and sequence of (potentially data-dependent) individual significance levels $\alpha_1,\alpha_2,\ldots$ .\\
 \hspace*{\algorithmicindent}\textbf{Output:} Query sets $S_1\subseteq S_2 \subseteq \ldots$ and true discovery bounds $d_1\leq d_2 \leq \ldots$ .
\begin{algorithmic}[1]
\State $d_0=0$
\State $S_0=\emptyset$
\State $A_0^c=\emptyset$
\For{$t=1,2,\ldots$}
\If{$P_t\leq \alpha_t$}
\State $S_t=S_{t-1}\cup \{t\}$
\Else
\State $S_t=S_{t-1}$
\EndIf
\If{$\sum_{i\in \{1,\ldots,t\}\setminus A_{t-1}^c}  \mathbbm{1}\{P_i\leq \alpha_i\} - c\alpha_i \geq  ca$} 
\State $d_t=d_{t-1}+1$
\State $A_{t}^c=A_{t-1}^c\cup \{\text{index of smallest individual significance level in } S_t\setminus A_{t-1}^c\}$
\Else
\State $d_{t}=d_{t-1}$
\State $A_{t}^c=A_{t-1}^c$
\EndIf
\State \Return $S_t, d_t$
\EndFor
\end{algorithmic}
\end{algorithm}
\else 
\begin{algorithm}
\caption{\texttt{Closed online-simple}} \label{alg:online-simple}
 \begin{flushleft}
\textbf{Input:} Sequence of p-values $P_1,P_2,\ldots$ and sequence of (potentially data-dependent) individual significance levels $\alpha_1,\alpha_2,\ldots$ .\\
 \textbf{Output:} Query sets $S_1\subseteq S_2 \subseteq \ldots$ and true discovery bounds $d_1\leq d_2 \leq \ldots$ .
 \end{flushleft}
\begin{algorithmic}[1]
\State $d_0=0$
\State $S_0=\emptyset$
\State $A_0^c=\emptyset$
\For{$t=1,2,\ldots$}
\If{$P_t\leq \alpha_t$}
\State $S_t=S_{t-1}\cup \{t\}$
\Else
\State $S_t=S_{t-1}$
\EndIf
\If{$\sum_{i\in \{1,\ldots,t\}\setminus A_{t-1}^c}  \mathbbm{1}\{P_i\leq \alpha_i\} - c\alpha_i \geq  ca$} 
\State $d_t=d_{t-1}+1$
\State $A_{t}^c=A_{t-1}^c\cup \{\text{index of smallest individual significance level in } S_t\setminus A_{t-1}^c\}$
\Else
\State $d_{t}=d_{t-1}$
\State $A_{t}^c=A_{t-1}^c$
\EndIf
\State \Return $S_t, d_t$
\EndFor
\end{algorithmic}
\end{algorithm}
\fi

Let $A_t^c$ and $d_t$ be defined as in \texttt{closed online-simple}, then 
\begin{align}
    d_t&\geq |A_{t-1}^c|+ \left\lceil-ca+ \sum_{i\in \{1,\ldots,t\}\setminus A_{t-1}^c} \mathbbm{1}\{P_i\leq \alpha_i\} - c\alpha_i  \right\rceil \nonumber \\
    &=  \left\lceil-ca+ \sum_{i=1}^t \mathbbm{1}\{P_i\leq \alpha_i\} - \sum_{i\in \{1,\ldots,t\}\setminus A_{t-1}^c} c\alpha_i  \right\rceil \geq \bd^{\text{os}}(S_t), \label{eq:improvement_os}
\end{align}
which shows that our \texttt{closed online-simple} method uniformly improves the \texttt{online-simple} procedure by \citet{katsevich2020simultaneous}. The improvement can be divided into two parts. First, the \texttt{closed online-simple} procedure is coherent, providing that $d_{t-1}\leq d_t$ for all $t\in \mathbb{N}$. Second, every time the bound $d_t$ is increased by one, the summand $-c\alpha_i$ is excluded from the bound, where $\alpha_i$ is the smallest significance level with index in $S_t\setminus A_{t-1}^c$. This shows that the (online) closure principle and \texttt{SeqE-Guard} automatically adapt to the number of discoveries and thus the proportion of false hypotheses.

Note that $E_i^{\text{os}}$ only takes two values. It takes $\exp[\theta_c(1 - c\alpha_i)]$ if $P_i\leq \alpha_i$ and $\exp[-\theta_cc\alpha_i]$ if $P_i> \alpha_i$, where $\mathbb{P}(P_i\leq \alpha_i|\mathcal{F}_{i-1})\leq \alpha_i$ for all $\mathbb{P}\in H_i$.  
Hence, we have for all $\mathbb{P}\in H_i$,
\begin{align}
\mathbb{E}_{\mathbb{P}}[E_i^{\text{os}}|\mathcal{F}_{i-1}]\leq \alpha_i \exp[\theta_c(1 - c\alpha_i)] + (1-\alpha_i)\exp[-\theta_cc\alpha_i]\eqqcolon u_i. \label{eq:u_i}
\end{align}
For example, $\alpha_i=\alpha=0.1$ and $a=1$ yield $u_i= 0.977$ ($u_i$ is increasing in $a$; for $a=3$ we obtain $u_i= 0.997$), which shows that $E_i^{\text{os}}$ is not admissible and thus can be improved.  A simple improvement can be obtained by plugging the e-value $\tilde{E}_i^{\text{os}}=E_i^{\text{os}}/u_i$ instead of $E_i^{\text{os}}$ into \texttt{SeqE-Guard}. Doing this at every step, we obtain a further improvement; we call this the \texttt{admissible online-simple} method.

\begin{proposition}\label{prop:improvement_os}
    By plugging in  the sequential e-values $(E_i^{\text{os}})_{i\in \mathbb{N}}$ \eqref{eq:e-val_os} into the \texttt{SeqE-Guard} algorithm, we obtain multiple uniform improvements (that can be applied together) over the \texttt{online-simple} method by \citet{katsevich2020simultaneous}:
    \begin{enumerate}
    \item The lower bound $d_t$ is nondecreasing in $t$ (\texttt{coherent online-simple}).
    \item Every time the bound $d_t$ is increased by one, the summand $-c\alpha_i$ is excluded from the bound, where $\alpha_i$ is the smallest threshold with index in $S_t\setminus A_{t-1}^c  $ (\texttt{closed online-simple}).
    \end{enumerate}
    Furthermore, since the expected value of $E_i^{\text{os}}$ is strictly smaller than $1$ under $H_0$, an additional improvement can be obtained:
    \begin{enumerate}
    \item[3.] The sequential e-values $E_i^{\text{os}}$, $i\in \mathbb{N}$, can be replaced by $\tilde{E}_i^{\text{os}}=E_i^{\text{os}}/u_i$, where $u_i<1$ is given by \eqref{eq:u_i} (\texttt{admissible online-simple}).
\end{enumerate}
\end{proposition}



\citet{katsevich2020simultaneous} introduced one further online procedure with simultaneous true discovery guarantee, the \texttt{online-adaptive} method. A uniform improvement can be obtained in the exact same manner as for the \texttt{online-simple} method above; we present this in Supplementary Material~\ref{appn:online-adaptive}. 
Inspired by the methods of \citet{katsevich2020simultaneous}, two further online procedures with simultaneous true discovery guarantee were proposed by \citet{iqraa2024false}. The first procedure is obtained by taking a union of  \texttt{online-simple} bounds for different parameters $a$. The second procedure exploits Freedman's inequality \citep{freedman1975tail} and a union bound. In Supplementary Material~\ref{sec:improvements_iqraa}, we show how both these recent methods can be uniformly improved by our e-value based approach. The proposed improvements are technically not instances of the \texttt{SeqE-Guard} algorithm, but can be obtained by the union of \texttt{SeqE-Guard} bounds (or using the average of multiple test martingales). In Table~\ref{tab:improvements} we provide the concrete sequential e-values yielding the mentioned improvements.

\ifarxiv
\begin{table}[h!]
    \centering
    \begin{tabular}{lll}
    \hline
        \textbf{Existing method} & \textbf{Sequential e-values} & \textbf{Algorithm} \\ \hline
           \texttt{online-simple} \citep{katsevich2020simultaneous} 
        & \makecell[l]{$E_i^{\text{os}}=\exp[\theta_c(\mathbbm{1}\{P_i\leq \alpha_i\} - c\alpha_i)]$ \\
        $a>0$ \\
        $c=\frac{\log(1/\alpha)}{a\log(1+\log(1/\alpha)/a)}$\\
        $\theta_c=\log(1/\alpha)/(ca)$} & \makecell[l]{\texttt{SeqE-Guard} }
        \\ \hline
        \texttt{online-adaptive} \citep{katsevich2020simultaneous} 
        & \makecell[l]{$E_i^{\text{ad}}=\exp[\theta_c(\mathbbm{1}\{P_i\leq \alpha_i\} - c\frac{\alpha_i}{1-\lambda_i}\mathbbm{1}\{P_i>\lambda_i\})]$ \\
        $\lambda_i\in [\alpha_i,1)$ \\
        $a>0$ \\
        $c=\frac{\log(1/\alpha)}{a\log(1+(1-\alpha^{B/a})/B)}$\\
        $B=\sup_{i\in \mathbb{N}} \frac{\alpha_i}{1-\lambda_i}$ \\
        $\theta_c=\log(1/\alpha)/(ca)$} & \makecell[l]{\texttt{SeqE-Guard} }
        \\ \hline
        \texttt{u-online-simple} \citep{iqraa2024false} 
        & \makecell[l]{$E_i^{\text{os},a}=\exp[\theta_{c_a}(\mathbbm{1}\{P_i\leq \alpha_i\} - c_a\alpha_i)]$ \\
        $c_a=\frac{\log(1/\alpha(a))}{a\log(1+\log(1/\alpha(a))/a)}$\\
        $\alpha(a)=\frac{6\alpha}{a^2\pi^2}$ \\
        $\theta_{c_a}=\log(1/\alpha(a))/(c_a a)$} & \makecell[l]{\texttt{SeqE-Guard} + \\  union over $a\in \mathbb{N}$ }
        \\ \hline
        \texttt{u-online-Freedman} \citep{iqraa2024false} 
        & \makecell[l]{$E_i^{\text{Freed},a}= \exp[\lambda_a(\mathbbm{1}\{P_i\leq \alpha_i\}-\alpha_i)$ \\ \hphantom{$E_i^{\text{Freed},a}=$}$-\psi(\lambda_a)\alpha_i(1-\alpha_i)]$ \\
        $\lambda_a=\log\left(1+\frac{\kappa_a}{a}\right)$\\
        $\psi(\lambda_a)=\exp(\lambda_a)-\lambda_a-1$ \\
        $\kappa_a=\sqrt{2a\log(1/\alpha(a))} + \frac{\log(1/\alpha(a))}{2}$ \\
        $\alpha(a)=\alpha\left(\frac{6}{\max(2\log_2(a),1)^2 (\pi^2 + 6)}\right)$} & \makecell[l]{\texttt{SeqE-Guard} + \\  union over $a=2^{j/2},$ \\  $j\in \mathbb{N}\cup \{0\}$ }
        \\ \hline
    \end{tabular}
    \caption{List of online procedures that can be uniformly improved by our \texttt{SeqE-Guard} algorithm. The left column shows the name of the existing method proposed in past work. The middle column gives the sequential e-value, which if plugged into the algorithm in the right column, yields a uniform improvement of the existing method in the left column. To the best of our knowledge this list includes all existing nontrivial online true discovery procedures. Thus, we claim that our approach uniformly improves all existing procedures. It should be noted that the shown sequential e-values may be further improved by ensuring that their expected value equals exactly $1$ in case of uniformly distributed p-values (this can be accomplished by dividing the shown e-values by their expected value but we omit it here to keep the expressions short).
    }
    \label{tab:improvements}
\end{table}
\else
\begin{table}[h!]
    \centering
    \resizebox{\textwidth}{!}{
    \begin{tabular}{lll}
    \hline
        \textbf{Existing method} & \textbf{Sequential e-values} & \textbf{Algorithm} \\ \hline
           \makecell[l]{\texttt{online-simple} \\ \citep{katsevich2020simultaneous} }
        & \makecell[l]{$E_i^{\text{os}}=\exp[\theta_c(\mathbbm{1}\{P_i\leq \alpha_i\} - c\alpha_i)]$ \\
        $a>0$ \\
        $c=\frac{\log(1/\alpha)}{a\log(1+\log(1/\alpha)/a)}$\\
        $\theta_c=\log(1/\alpha)/(ca)$} & \makecell[l]{\texttt{SeqE-Guard} }
        \\ \hline
        \makecell[l]{\texttt{online-adaptive} \\ \citep{katsevich2020simultaneous} }
        & \makecell[l]{$E_i^{\text{ad}}=\exp[\theta_c(\mathbbm{1}\{P_i\leq \alpha_i\} - c\frac{\alpha_i}{1-\lambda_i}\mathbbm{1}\{P_i>\lambda_i\})]$ \\
        $\lambda_i\in [\alpha_i,1)$ \\
        $a>0$ \\
        $c=\frac{\log(1/\alpha)}{a\log(1+(1-\alpha^{B/a})/B)}$\\
        $B=\sup_{i\in \mathbb{N}} \frac{\alpha_i}{1-\lambda_i}$ \\
        $\theta_c=\log(1/\alpha)/(ca)$} & \makecell[l]{\texttt{SeqE-Guard} }
        \\ \hline
        \makecell[l]{ \texttt{u-online-simple} \\ \citep{iqraa2024false} }
        & \makecell[l]{$E_i^{\text{os},a}=\exp[\theta_{c_a}(\mathbbm{1}\{P_i\leq \alpha_i\} - c_a\alpha_i)]$ \\
        $c_a=\frac{\log(1/\alpha(a))}{a\log(1+\log(1/\alpha(a))/a)}$\\
        $\alpha(a)=\frac{6\alpha}{a^2\pi^2}$ \\
        $\theta_{c_a}=\log(1/\alpha(a))/(c_a a)$} & \makecell[l]{\texttt{SeqE-Guard} + \\  union over $a\in \mathbb{N}$ }
        \\ \hline
        \makecell[l]{\texttt{u-online-Freedman} \\ \citep{iqraa2024false} }
        & \makecell[l]{$E_i^{\text{Freed},a}= \exp[\lambda_a(\mathbbm{1}\{P_i\leq \alpha_i\}-\alpha_i)$ \\ \hphantom{$E_i^{\text{Freed},a}=$}$-\psi(\lambda_a)\alpha_i(1-\alpha_i)]$ \\
        $\lambda_a=\log\left(1+\frac{\kappa_a}{a}\right)$\\
        $\psi(\lambda_a)=\exp(\lambda_a)-\lambda_a-1$ \\
        $\kappa_a=\sqrt{2a\log(1/\alpha(a))} + \frac{\log(1/\alpha(a))}{2}$ \\
        $\alpha(a)=\alpha\left(\frac{6}{\max(2\log_2(a),1)^2 (\pi^2 + 6)}\right)$} & \makecell[l]{\texttt{SeqE-Guard} + \\  union over $a=2^{j/2},$ \\  $j\in \mathbb{N}\cup \{0\}$ }
        \\ \hline
    \end{tabular}}
    \caption{List of online procedures that can be uniformly improved by our \texttt{SeqE-Guard} algorithm. The left column shows the name of the existing method proposed in past work. The middle column gives the sequential e-value, which if plugged into the algorithm in the right column, yields a uniform improvement of the existing method in the left column. To the best of our knowledge this list includes all existing nontrivial online true discovery procedures. Thus, we claim that our approach uniformly improves all existing procedures. It should be noted that the shown sequential e-values may be further improved by ensuring that their expected value equals exactly $1$ in case of uniformly distributed p-values (this can be accomplished by dividing the shown e-values by their expected value but we omit it here to keep the expressions short).
    }
    \label{tab:improvements}
\end{table}
\fi

\ifdiff 
\subsubsection{\added{Improving existing offline procedures}\label{sec:improving_offline}}
\added{Interestingly, even though the \texttt{SeqE-Guard} algorithm was specifically designed for the \textit{online} multiple testing setting, it also allows to improve some \textit{offline} methods. \citet{katsevich2020simultaneous} introduced true discovery procedures for pre-ordered and interactive paths. That means, even though all data are available at the beginning, the hypotheses are ordered based on some side information or orthogonal data.  For instance, this includes simultaneous true discovery bounds for the knockoff framework \citep{barber2015controlling, candes2018panning}.} 

\added{These offline procedures are based on the same idea as the \texttt{online-simple} and \texttt{online-adaptive} method (see Lemma~1 in \citet{katsevich2020simultaneous}). Given ordered p-values $P_{(1)},P_{(2)},\ldots$ and nonnegative thresholds $\alpha_1,\alpha_2,\ldots$, they proposed the general bounds} 
\begin{align*}
    \bd^{\text{KR}}(S_t)=-c_{\alpha} a+ \sum_{i\leq t} \mathbbm{1}\{P_{(i)}\leq \alpha_i\} - c_\alpha h_i(P_{(i)}) \quad (t\in \mathbb{N})
\end{align*}
\added{for the query path $S_t=\{i\leq t: P_{(i)}\leq \alpha_i\}$, $t\in \mathbb{N}$, where $(h_i)_{i \in \mathbb{N}}$ are functions on $[0,1]$, $a>1$ some regularization parameter and $c$ a constant. They implicitly proved that $\bd^{\text{KR}}$ provides simultaneous true discovery guarantee over all $S_t$, $t\in \mathbb{N}$, if $E_{(1)}^{\text{KR}},E_{(2)}^{\text{KR}},\ldots$ define sequential e-values, where}
\begin{align}
E_{(i)}^{\text{KR}}=\exp\left[ \frac{\log(1/\alpha)}{ac} \left(\mathbbm{1}\{P_{(i)}\leq \alpha_i \} -c h_i(P_{(i)}) \right)\right] \quad (i\in \mathbb{N}), \label{eq:e-values_KR_gen}
\end{align}
\added{and gave specific choices of $(h_i)_{i \in \mathbb{N}}$, $(\alpha_i)_{i \in \mathbb{N}}$ and $c$ that satisfy this requirement. It is apparent that this general approach recovers the \texttt{online-simple} method for $(i)=i$ and $h_i(P_{(i)})=\alpha_i$. Furthermore, it can easily be concluded that we can improve these general bounds $\bd^{\text{KR}}$ by plugging the e-values $E_{(1)}^{\text{KR}}, E_{(2)}^{\text{KR}},\ldots $ into \texttt{SeqE-Guard}, as we illustrated with the \texttt{online-simple} method in Section~\ref{sec:improving_online}. Since the main focus of this paper is on the online setting, we will not discuss each of the improvements individually.}

\else

\subsubsection{Improving existing offline procedures\label{sec:improving_offline}}
Interestingly, even though the \texttt{SeqE-Guard} algorithm was specifically designed for the \textit{online} multiple testing setting, it also allows to improve some \textit{offline} methods. \citet{katsevich2020simultaneous} introduced true discovery procedures for pre-ordered and interactive paths. That means, even though all data are available at the beginning, the hypotheses are ordered based on some side information or orthogonal data.  For instance, this includes simultaneous true discovery bounds for the knockoff framework \citep{barber2015controlling, candes2018panning}. 

These offline procedures are based on the same idea as the \texttt{online-simple} and \texttt{online-adaptive} method (see Lemma~1 in \citet{katsevich2020simultaneous}). Given ordered p-values $P_{(1)},P_{(2)},\ldots$ and nonnegative thresholds $\alpha_1,\alpha_2,\ldots$, they proposed the general bounds 
\begin{align*}
    \bd^{\text{KR}}(S_t)=-c_{\alpha} a+ \sum_{i\leq t} \mathbbm{1}\{P_{(i)}\leq \alpha_i\} - c_\alpha h_i(P_{(i)}) \quad (t\in \mathbb{N})
\end{align*}
for the query path $S_t=\{i\leq t: P_{(i)}\leq \alpha_i\}$, $t\in \mathbb{N}$, where $(h_i)_{i \in \mathbb{N}}$ are functions on $[0,1]$, $a>1$ some regularization parameter and $c$ a constant. They implicitly proved that $\bd^{\text{KR}}$ provides simultaneous true discovery guarantee over all $S_t$, $t\in \mathbb{N}$, if $E_{(1)}^{\text{KR}},E_{(2)}^{\text{KR}},\ldots$ define sequential e-values, where
\begin{align}
E_{(i)}^{\text{KR}}=\exp\left[ \frac{\log(1/\alpha)}{ac} \left(\mathbbm{1}\{P_{(i)}\leq \alpha_i \} -c h_i(P_{(i)}) \right)\right] \quad (i\in \mathbb{N}), \label{eq:e-values_KR_gen}
\end{align}
and gave specific choices of $(h_i)_{i \in \mathbb{N}}$, $(\alpha_i)_{i \in \mathbb{N}}$ and $c$ that satisfy this requirement. It is apparent that this general approach recovers the \texttt{online-simple} method for $(i)=i$ and $h_i(P_{(i)})=\alpha_i$. Furthermore, it can easily be concluded that we can improve these general bounds $\bd^{\text{KR}}$ by plugging the e-values $E_{(1)}^{\text{KR}}, E_{(2)}^{\text{KR}},\ldots $ into \texttt{SeqE-Guard}, as we illustrated with the \texttt{online-simple} method in Section~\ref{sec:improving_online}. Since the main focus of this paper is on the online setting, we will not discuss each of the improvements individually.
\fi 



\subsection{Adaptively hedged GRO e-values\label{sec:gro}}

The most common strategy to calculate e-values in practice is based on variants of the GRO-criterion \citep{shafer2021testing, grunwald2020safe, waudby2023estimating}, which dates back to the Kelly criterion \citep{kelly1956new, breiman1961optimal}. Here, we assume that each null hypothesis $H_i$ comes with an alternative $H_i^A$. Suppose $H_i^A$ contains a single distribution $\mathbb{Q}_i$. Then the growth rate optimal (GRO) e-value $E_i^{\mathrm{GRO}}$ is defined as the e-value $E_i$ that maximizes the growth rate under $\mathbb{Q}_i$, given by  $\mathbb{E}_{\mathbb{Q}_i}[\log(E_i)]$, over all e-values for $H_i$. 
If $H_i^A$ is composite, a prior over its distributions can be used, defining the GRO e-value via the corresponding mixture distribution. If the null hypothesis is simple, the GRO e-value is given by the likelihood ratio (LR) of the alternative over the null distribution \citep{shafer2021testing}. If $H_i$ is composite, the GRO e-value takes the form of a LR of the alternative against a specific (sub-) distribution \citep{grunwald2020safe, larsson2024numeraire}. The GRO e-value is particularly powerful when many e-values are combined by multiplication \citep{shafer2021testing} and therefore seems to be a reasonable choice for our \texttt{SeqE-Guard} algorithm. Furthermore, since the growth rate is the standard measure of performance for e-values \citep{ramdas2023game}, there might be applications where we do not have access to the data to calculate our own e-value, but just get a GRO e-value for each individual hypothesis. For example, this could be the case in meta-analyses, where each study just reported the GRO e-value. In the following, we will discuss how the GRO concept transfers to online true discovery guarantee based on sequential e-values.

A naive approach would be to just plugin the GRO e-values into \texttt{SeqE-Guard}. However, the GRO e-values only maximize the growth rate under their alternatives. This makes sense when testing a single hypothesis. In our setting, the product of GRO e-values would only maximize the growth rate of $W_t=\prod_{i=1}^t E_i$ if all hypotheses are false --- a very unlikely scenario. Indeed, GRO e-values can be small if the null hypothesis is true, so directly using GRO e-values in \texttt{SeqE-Guard} can lead to low power. Hence, when multiplying e-values for different hypotheses one needs to incorporate the possibility that a hypothesis is true. One approach is to hedge the GRO e-values by defining $$\tilde{E}_i^{\mathrm{GRO}}=1-\lambda_i+\lambda_iE_i^{\mathrm{GRO}},$$ where $\lambda_i\in [0,1]$ is measurable with respect to $\mathcal{F}_{i-1}$. To see that $\tilde{E}_1^{\mathrm{GRO}}, \tilde{E}_2^{\mathrm{GRO}}, \ldots$ indeed define sequential e-values if the GRO e-values are sequential, just note that $\mathbb{E}_{P}[\tilde{E}_i^{\mathrm{GRO}}|\mathcal{F}_{i-1}]=1-\lambda_i+\lambda_i\mathbb{E}_{\mathbb{P}}[E_i^{\mathrm{GRO}}|\mathcal{F}_{i-1}]\leq 1$ for all $\mathbb{P} \in H_i$. 

In order to derive a reasonable choice for $\lambda_i$, we consider a specific Bayesian two-groups model. Suppose $H_i=\{\mathbb{P}_i\}$ and $H_i^A=\{\mathbb{Q}_i\}$ are both simple hypotheses and which of these hypotheses is true is random, where $\tau_i$ gives the probability that the alternative hypothesis $H_i^A$ is true. A reasonable approach would be to choose the e-value $E_i$ for $H_i$ that maximizes the growth rate under the true distribution $\mathbb{E}_{(1-\tau_i)\mathbb{P}_i+\tau_i\mathbb{Q}_i}[\log(E_i)]$. 
\begin{proposition}\label{prop:gro}
    Suppose $\mathbb{Q}_i$ is absolutely continuous with respect to $\mathbb{P}_i$. Then the e-value $E_i$ that maximizes the growth rate under the true distribution $(1-\tau_i)\mathbb{P}_i+\tau_i\mathbb{Q}_i$ is given by $1-\tau_i+\tau_iE_i^{\mathrm{GRO}}$, where $E_i^{\mathrm{GRO}}$ maximizes the growth rate under the alternative $\mathbb{Q}_i$. 
\end{proposition}
\begin{proof}
    Due to \citet{shafer2021testing}, the e-value maximizing the growth rate under $(1-\tau_i)\mathbb{P}_i+\tau_i\mathbb{Q}_i$ is given by the likelihood ratio $$\frac{d[(1-\tau_i)\mathbb{P}_i+\tau_i\mathbb{Q}_i]}{d\mathbb{P}_i}=1-\tau_i+\tau_i\frac{ d \mathbb{Q}_i}{d\mathbb{P}_i}=1-\tau_i+\tau_iE_i^{\mathrm{GRO}},$$
    where $\frac{ d \mathbb{Q}_i}{d\mathbb{P}_i}$ denotes the Radon-Nikodym derivative.
\end{proof}

Proposition~\ref{prop:gro} shows that in the case of simple null and alternative hypotheses the optimal e-value under the true distribution is the same as if we maximize the growth rate under the alternative and then hedge the resulting e-value according to the probability that the alternative is true. We think it is possible that something analogous holds for composite null hypotheses as well. However, even if this is not the case, it seems to be a reasonable strategy in general.

Therefore, our approach is to specify an estimate $\hat{\tau}_i$ for the probability $\tau_i$ that $H_i$ is false, where $\hat{\tau}_i$ can either depend on prior information or the past data, and then set $\lambda_i=\hat{\tau}_i$. We propose to set  \begin{align}
\hat{\tau}_i=\frac{1/2+\sum_{j=1}^{i-1} \mathbbm{1}\{E_j^{\mathrm{GRO}}>1\}}{i} \qquad (i\in \mathbb{N}),\label{eq:weight_gro}
\end{align}
if no prior information is available. The idea is that  $E_j^{\mathrm{GRO}}>1$ indicates data favoring the alternative over the null distribution, since the GRO e-value can be interpreted as a (generalized) LR.

We illustrate the behavior of $\hat{\tau}_i$ defined in \eqref{eq:weight_gro} for different scenarios in Figure \ref{fig:sim_tau}. We consider the simulation setup described in Section \ref{sec:sim}, but with $n=100$ and proportion of false hypotheses $\pi_A\in \{0.2,0.5,0.8\}$. If the signal is weak, $\hat{\tau}_i$ lies between the true proportion and $0.5$ and if the signal is strong, the estimate $\hat{\tau}_i$ is close to the true proportion. We think that this behavior is desirable, since a tendency towards $0.5$ is not that hurtful if the alternative and null distribution are close as the GRO e-values have a small variance in this case.    

\begin{figure}[tb]
\centering
\includegraphics[width=0.8\textwidth]{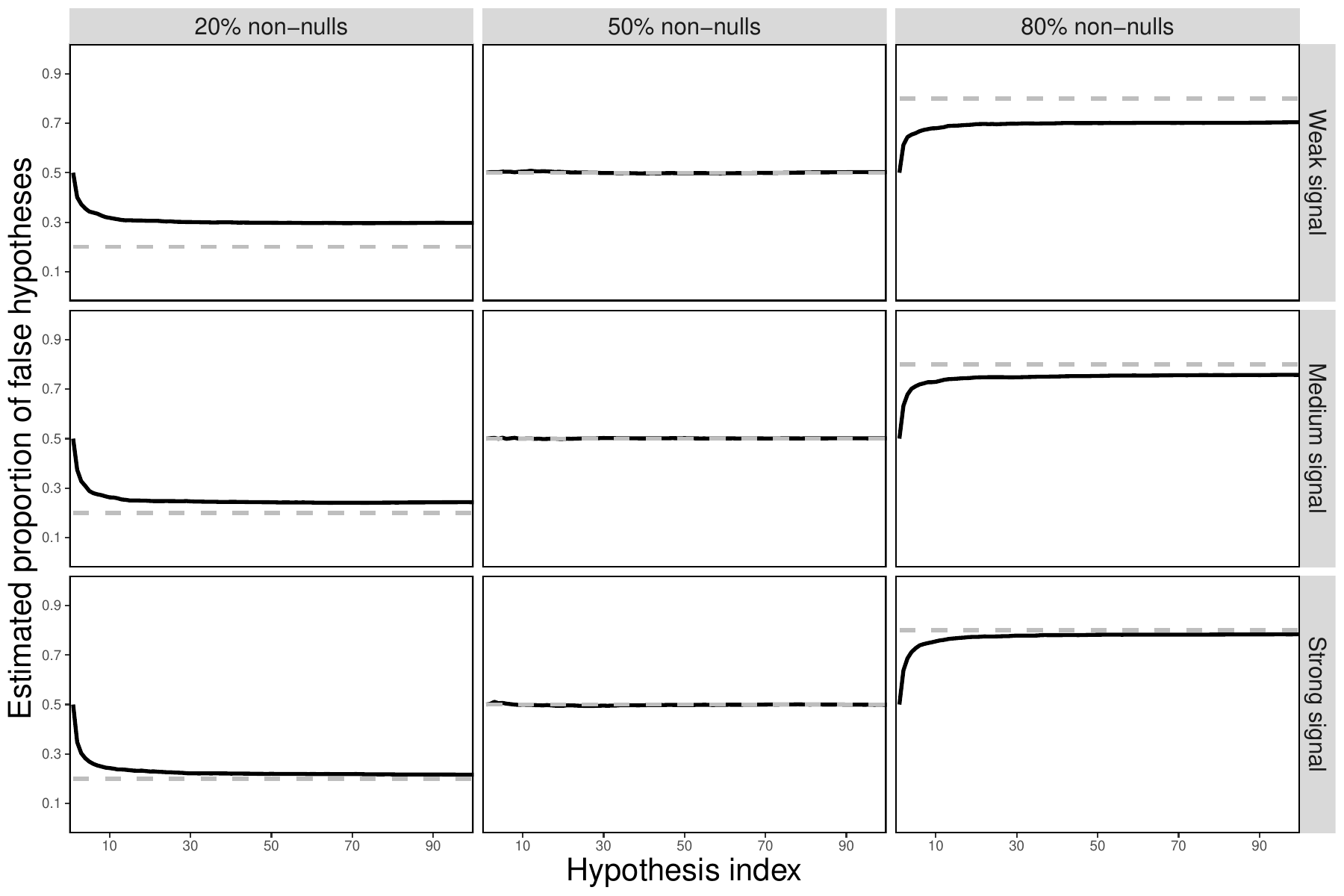}
\caption{Comparing $\hat{\tau}_i$ (solid line) defined in \eqref{eq:weight_gro} to the true $\tau_i$ (dashed line) for different scenarios. The simulation design is described in Section \ref{sec:sim}. The estimate $\hat{\tau}_i$ has a tendency to bias towards $50\%$ but it estimates the proportion of false hypotheses very well when the signal is strong.  \label{fig:sim_tau} }\end{figure}

    Hedging e-values before multiplying them does not only apply to GRO e-values. Such strategies for merging sequential e-values have been used by many preceding authors like \citet{waudby2023estimating, vovk2024merging}. However, the above argumentation provides a reasonable choice for the parameter $\lambda_i$ in our setting. We analyze this approach experimentally in Section \ref{sec:sim_gro}.

\section{Simulations\label{sec:sim}}

In this section we numerically calculate the true discovery proportion (TDP) bound, which is defined as the true discovery bound for $S_t$ divided by the size of $S_t$. We compare TDP bounds obtained by applying \texttt{SeqE-Guard} to the different sequential e-values proposed in the previous sections. In Subsection~\ref{sec:sim_KR}, we compare the \texttt{online-simple} method by \citet{katsevich2020simultaneous} with its uniform improvement. In Subsection~\ref{sec:sim_gro}, we demonstrate how hedging GRO e-values improve the true discovery bound. In Subsection~\ref{sec:sim_all}, we compare the proposed e-values to decide which is best suited for practice. \ifdiff \added{An application to real data can be found in the Supplementary Material~\ref{sec:IMPC}.} \else An application to real data can be found in  Supplementary Material~\ref{sec:IMPC}. \fi

We consider the same simulation setup in all subsections. We sequentially  test $n=1000$ null hypotheses $H_i$, $i\in \{1,\ldots,n\}$, of the from $H_i:X_i \sim \mathcal{N}(0,1)$ against the alternative  $H_i^A:X_i \sim \mathcal{N}(\mu_A,1)$ for some $\mu_A>0$, where $X_1,\ldots, X_{n}$ are independent data points or test statistics. The probability that the alternative hypothesis is true is set by a parameter $\pi_A\in (0,1)$ and the desired guarantee is set to $\alpha=0.1$. For all comparisons we consider a grid of simulation parameters $\mu_A\in \{2, 3, 4\}$ and $\pi_A\in \{0.1, 0.3, 0.5\}$, where we refer to $\mu_A=2$ as weak signal, $\mu_A=3$ as medium signal and $\mu_A=4$ as strong signal. The p-values are calculated by $\Phi(-X_i)$, where $\Phi$ is the CDF of a normal distribution. The raw GRO e-values are given by the likelihood ratio $E_i^{\mathrm{GRO}}=p_{\mu_A}(X_i)/p_0(X_i)$, where $p_{\mu_A}$ and  $p_{0}$ are the densities of a normal distribution with variance $1$ and mean $\mu_A$ and $0$, respectively. The query sets $S_t$, $t\in \{1,\ldots,t\}$, are defined as $S_t=\{i\in \{1,\ldots, t\}: P_i\leq \alpha\}$. All of the results in the following are obtained by averaging over $1000$ independent trials.

\subsection{Comparing the \texttt{online-simple} method \citep{katsevich2020simultaneous} with its improvements\label{sec:sim_KR}}

In Section~\ref{sec:katse} we showed that the \texttt{online-simple method} by \citet{katsevich2020simultaneous} can be uniformly improved by the \texttt{closed online-simple} procedure (Algorithm~\ref{alg:online-simple}). We also showed that this closed procedure can be further uniformly improved by the \texttt{admissible online-simple} procedure, which ensures that the expected value of each sequential e-value is exactly one. In this section, we aim to quantify the gain in power for making true discoveries by using these improvements instead of the \texttt{online-simple} method. Although \citet{katsevich2020simultaneous} proposed $a=1$ as default parameter, we found $a=3$ to perform better which is why we use it here. 


The results are illustrated in Figure~\ref{fig:sim_KR}. It can be seen that the \texttt{closed online-simple} procedure leads to a substantial improvement of the \texttt{online-simple} procedure in all cases. Of the queried hypotheses, the former approximately identifies $10\%-20\%$ more as false. 
The additional improvement of the admissible online-simple method is quite small in this case, however, it is potentially larger for smaller parameters $a$ (see Section~\ref{sec:katse}).

In Section~\ref{sec:improvements_iqraa} we also provide simulations for our improvements over the methods proposed by \citet{iqraa2024false}. The behavior of both, the initially proposed methods and our improvements, is quite similar as in the \texttt{online-simple} case.


\begin{figure}[tb]
\centering
\includegraphics[width=0.8\textwidth]{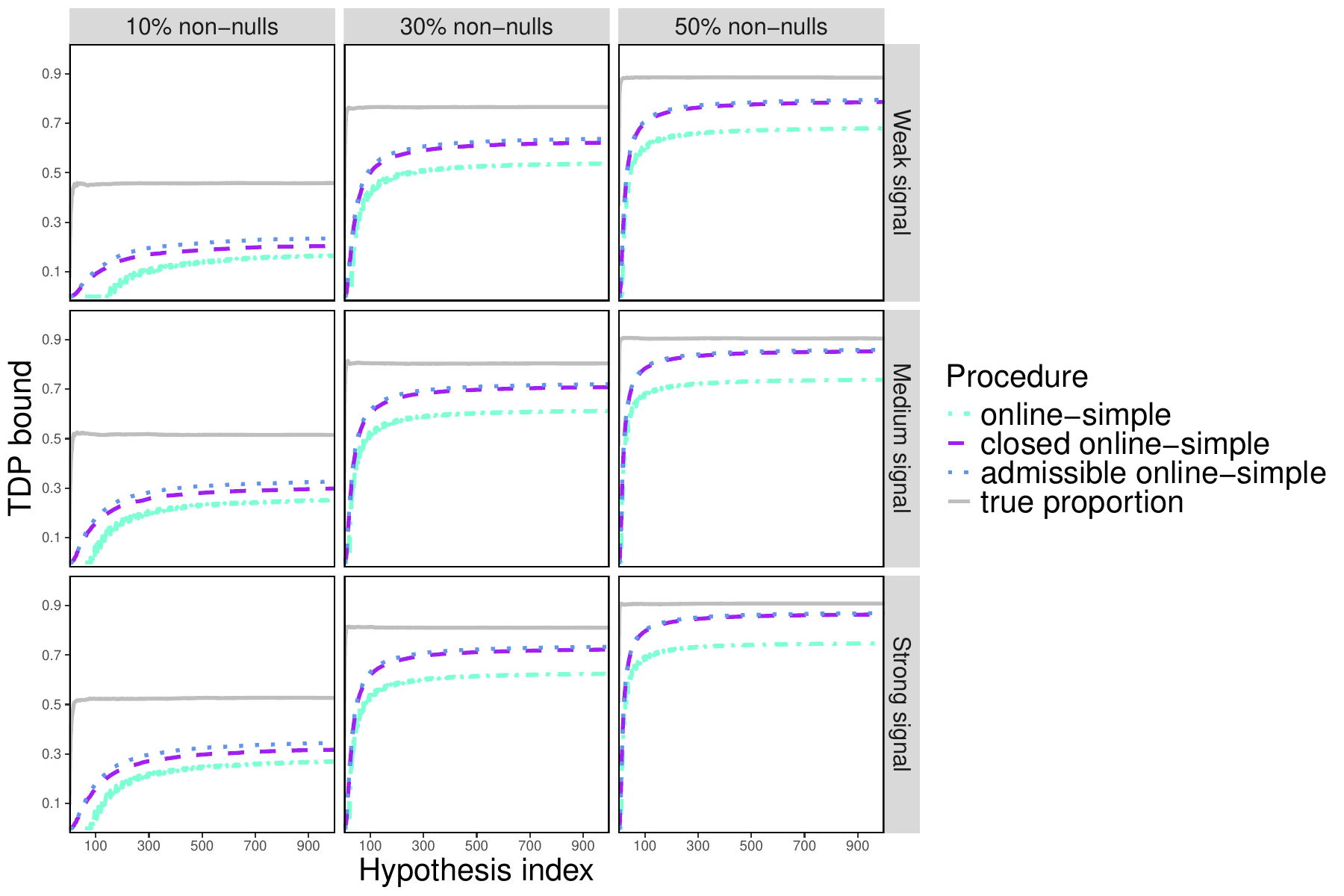}
\caption{True discovery proportion bounds obtained by the \texttt{online-simple} method \citep{katsevich2020simultaneous}, the \texttt{closed online-simple} method and the \texttt{admissible online-simple} method. 
The \texttt{closed online-simple} method leads to substantially larger bounds than the original \texttt{online-simple} method. 
An additional improvement can be obtained by the \texttt{admissible online-simple} method, which is closest to the true proportion (top line) in all figures. \label{fig:sim_KR} }\end{figure}

\subsection{GRO e-values\label{sec:sim_gro}}

In Section~\ref{sec:gro} we argued that the raw GRO e-values should be hedged to account for the probability that a null hypothesis is true. Now we compare the true discovery bounds obtained by applying \texttt{SeqE-Guard} to raw and hedged GRO e-values.
For the hedged GRO e-values we chose the predictable parameter proposed in \eqref{eq:weight_gro}. The results are illustrated in Figure~\ref{fig:sim_GRO}, showing that GRO e-values lead to very low bounds which can be increased substantially by hedging. 


\begin{figure}[tb]
\centering
\includegraphics[width=0.8\textwidth]{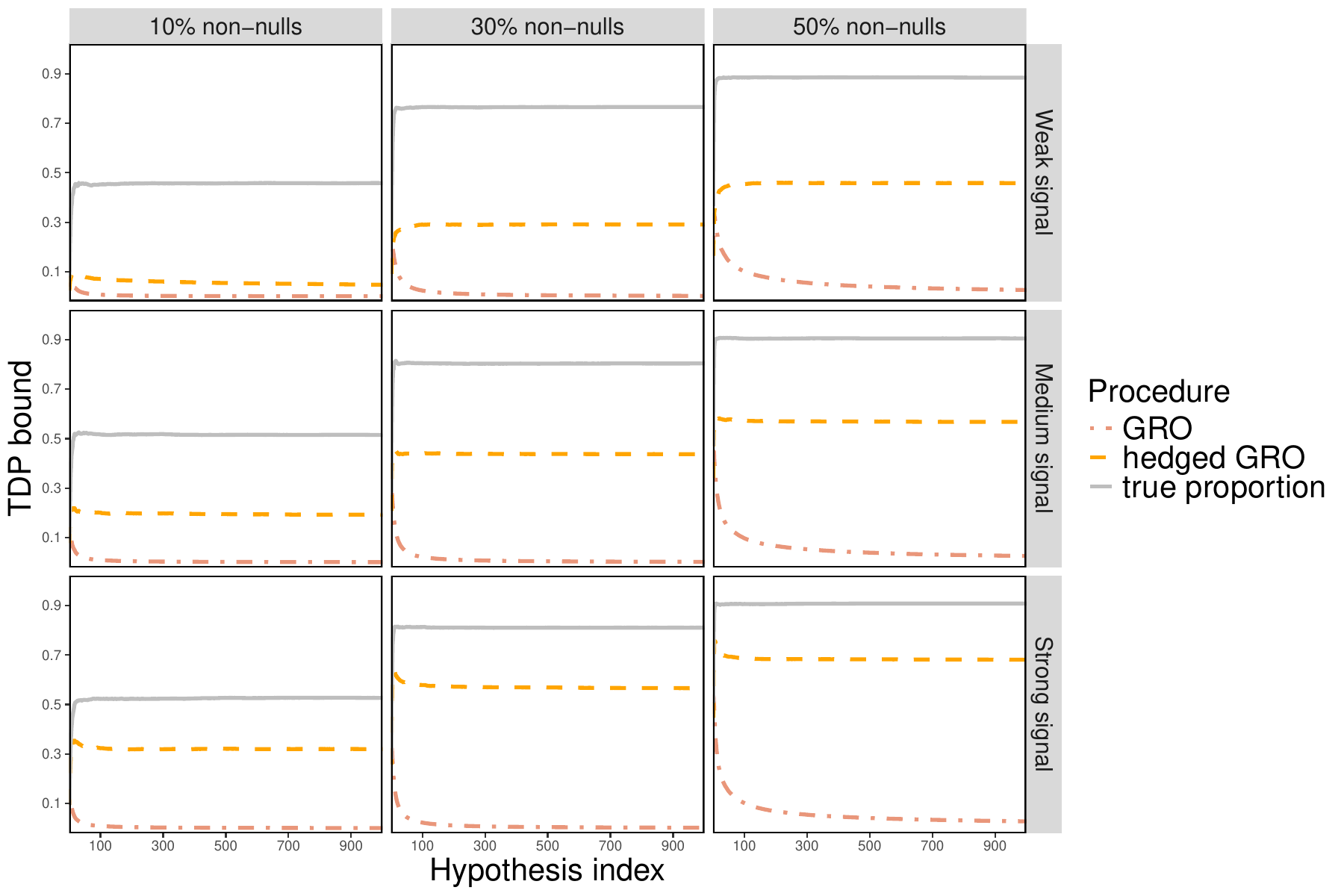}
\caption{True discovery proportion bounds obtained by applying \texttt{SeqE-Guard} to GRO e-values and hedged GRO e-values. The hedged GRO e-values improve the GRO e-values substantially. \label{fig:sim_GRO} }\end{figure}

\subsection{Which sequential e-values should we choose?\label{sec:sim_all}}

In this section, we compare the \texttt{SeqE-Guard} procedure when applied with the best versions of the proposed sequential e-values to derive recommendations for practice.  More precisely, we compare the \texttt{admissible online-simple} method with the hedged GRO e-value. 


\begin{figure}[tb]
\centering
\includegraphics[width=0.8\textwidth]{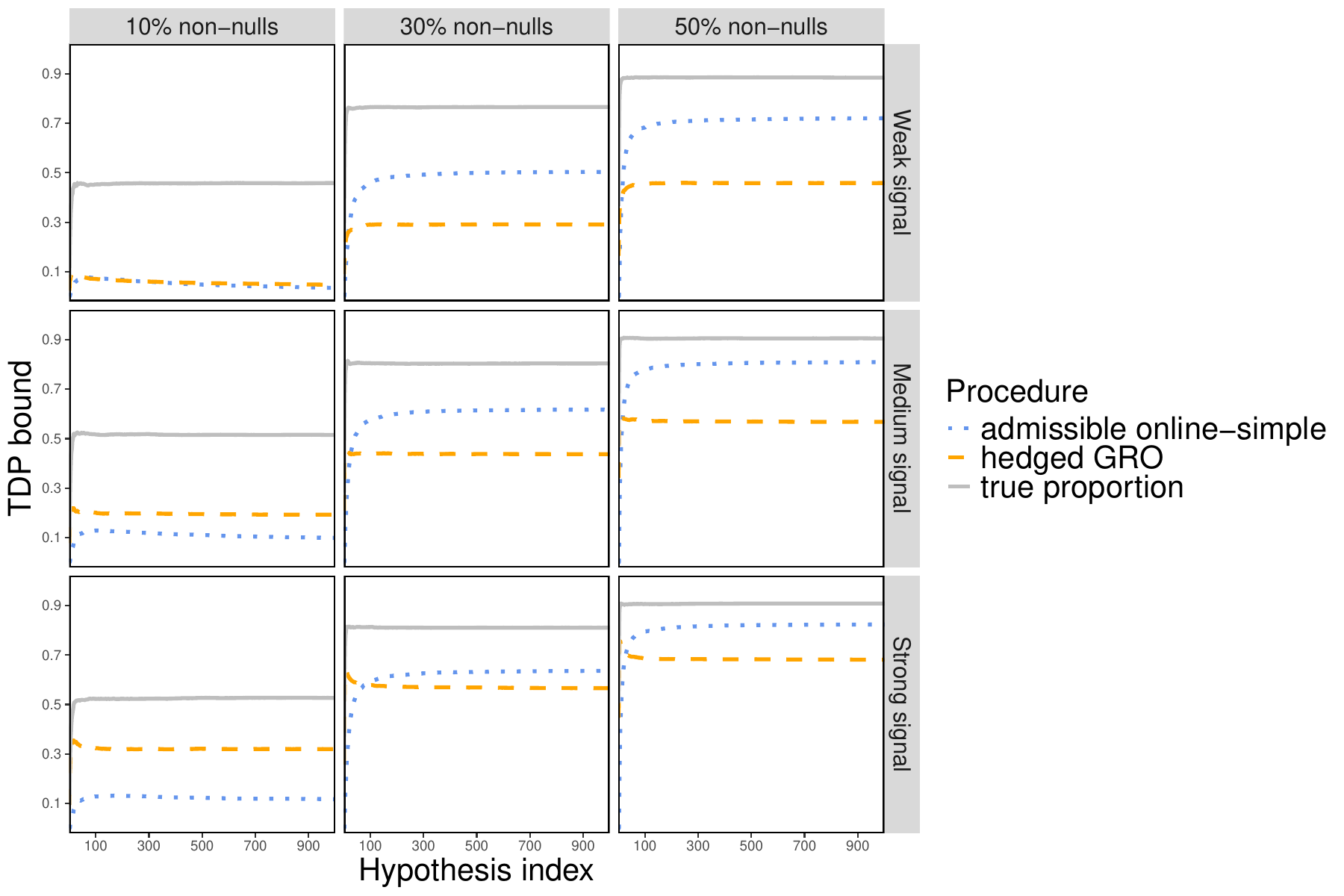}
\caption{True discovery proportion bounds obtained by the \texttt{admissible online-simple} method and applying \texttt{SeqE-Guard} to hedged GRO e-values. If the proportion of false hypotheses is small, the hedged GRO e-values perform well, particularly if the signal is strong. If the proportion of false hypotheses is large, the \texttt{admissible online-simple} method should be preferred. However, the   \texttt{admissible online-simple} method adapts to the specific query path, while the hedged GRO e-values are better suited for exploring multiple query paths. \label{fig:sim_best} }\end{figure}

The results are depicted in Figure~\ref{fig:sim_best}. 
The procedures perform quite differently in the various settings. When the proportion of false hypotheses is small, the \texttt{SeqE-Guard} algorithm performs best with hedged GRO e-values, particularly,  if the signal is strong. However, if the proportion of false hypotheses is large, the \texttt{admissible online-simple} method clearly outperforms the \texttt{SeqE-Guard} with hedged GRO e-values. Hence, if we expect sparse but strong signal, applying \texttt{SeqE-Guard} with hedged GRO e-values is the best choice. In contrast, for dense but weak signals, the \texttt{admissible online-simple} method should be preferred. 


Furthermore, it should be noted that the only reasonable query path for the \texttt{online-simple} method and its improvements is given by $S_t=\{i\leq t: P_i\leq \alpha_i\}$ (which we use in the simulations for $\alpha_i=\alpha$). To see this,  note that all e-values $E_i^{\text{os}}$, $i\leq t$, with $i\notin S_t$ are smaller than $1$ and therefore their inclusion in $S_t$ would not increase the lower bound for the number of true discoveries. Furthermore, excluding e-values $E_i^{\text{os}}$, $i\leq t$, with $P_i\leq \alpha_i$ would be difficult to communicate, as those e-values reached there maximum possible value. Hence, hedged GRO e-values are better suited  for an exploratory analysis where the scientist might be interested in several different query paths (see Corollary~\ref{corol:multiple_paths}). For example, \texttt{SeqE-Guard} with hedged GRO e-values can provide a (nontrivial) query path with online FWER control (by including $t\in S_t$ iff this implies $d_t=d_{t-1}+1$), while simultaneously providing (nontrivial) real-time lower bounds for the number of false hypotheses among all hypotheses with e-values greater than, say, $2$. For such an exploratory proceeding, we would recommend the hedged GRO e-values, since they showed good performance (Figure~\ref{fig:sim_GRO}) without adapting to the query path at all. In Supplementary Material~\ref{sec:boosting}, we show how GRO e-values could be adapted to the chosen query path using a boosting approach \citep{wang2022false} that uniformly improves the bounds for that specific query path (but may reduces the bounds for other query paths). 
In Supplementary Material~\ref{sec:calib} we show how one could investigate multiple query paths when only p-values are available, using other p-to-e calibrators than the simple binary calibrator applied for the \texttt{online-simple} method.

\section{Discussion}
In this paper, we proposed a new closed testing based online true discovery procedure for sequential e-values and derived a general short-cut that only requires one calculation per hypothesis. Although the \texttt{SeqE-Guard} algorithm is restricted to sequential e-values, it is a general procedure for the task of online true discovery guarantee, since there are many different ways to construct sequential e-values. In particular, it yields uniform improvements of the state-of-the-art methods by \citet{katsevich2020simultaneous} and \citet{iqraa2024false}, although they were not explicitly constructed using e-values.

From a theoretical point of view this paper gives new insights about the role of e-values in multiple testing by showing that every admissible coherent online true discovery procedure must be based on sequential e-values. From a practical point of view, we constructed a powerful and flexible multiple testing procedure, which allows to observe hypotheses one-by-one over time and make fully data-adaptive decisions about the hypotheses and stopping time while bounding the number of true discoveries or equivalently, the false discovery proportion. On the way, we introduced a new idea for hedging of GRO e-values in multiple testing. 

\ifdiff \added{Interestingly, even though we constructed the \texttt{SeqE-guard} algorithm primarily for the online setting, it also allows to uniformly improve some existing offline procedures. Future work could examine the implications of our results for the offline setting in more detail.} \else
Interestingly, even though we constructed the \texttt{SeqE-guard} algorithm primarily for the online setting, it also allows to uniformly improve some existing offline procedures. Future work could examine the implications of our results for the offline setting in more detail. \fi

\ifarxiv
\subsection*{Acknowledgments}
 The authors thank Etienne Roquain, \ifdiff \added{two anonymous reviewers, an associate editor and a joint editor for their helpful comments.} \else two anonymous reviewers, an associate editor and a joint editor for their helpful comments. \fi LF acknowledges funding by the Deutsche Forschungsgemeinschaft (DFG, German Research Foundation) – Project number 281474342/GRK2224/2. AR was funded by NSF grant DMS-2310718.
\else

\section{Acknowledgments}
 The authors thank Etienne Roquain, \ifdiff \added{two anonymous reviewers, an associate editor and a joint editor for their helpful comments.} \else two anonymous reviewers, an associate editor and a joint editor for their helpful comments. \fi LF acknowledges funding by the Deutsche Forschungsgemeinschaft (DFG, German Research Foundation) – Project number 281474342/GRK2224/2. AR was funded by NSF grant DMS-2310718.  

 \section{Competing interests}
No competing interest is declared.

\section{Supplementary material}
Supplementary material is available online at \textit{Journal of the Royal Statistical Society: Series B}.

\fi

\putbib
\end{bibunit}


\clearpage

\begin{center}
    {\LARGE \bfseries Supplementary material for ``Admissible online closed testing must employ e-values''} \\[1ex] %
\end{center}

\ifarxiv
\renewcommand{\thesection}{S.\arabic{section}}
\renewcommand{\theequation}{S.\arabic{equation}}
\renewcommand{\thetheorem}{S.\arabic{theorem}}
\renewcommand{\theassumption}{S.\arabic{assumption}}
\renewcommand{\theproperty}{S.\arabic{property}}
\renewcommand{\thefact}{S.\arabic{fact}}
\renewcommand{\theproposition}{S.\arabic{proposition}}
\renewcommand{\thecorol}{S.\arabic{corol}}
\renewcommand{\thelemma}{S.\arabic{lemma}}
\renewcommand{\theremark}{S.\arabic{remark}}
\renewcommand{\theexample}{S.\arabic{example}}
\renewcommand{\thefigure}{S.\arabic{figure}}
\renewcommand{\thetable}{S.\arabic{table}}
\setcounter{section}{0}
\setcounter{equation}{0}
\setcounter{theorem}{0}
\setcounter{assumption}{0}
\setcounter{property}{0}
\setcounter{fact}{0}
\setcounter{remark}{0}
\setcounter{example}{0}
\setcounter{figure}{0}
\setcounter{table}{0}
\else
\renewcommand{\thesection}{S.\arabic{section}}
\renewcommand{\theequation}{S.\arabic{equation}}
\renewcommand{\thetheorem}{S.\arabic{theorem}}
\renewcommand{\thefact}{S.\arabic{fact}}
\renewcommand{\theproposition}{S.\arabic{proposition}}
\renewcommand{\theremark}{S.\arabic{remark}}
\renewcommand{\theexample}{S.\arabic{example}}
\renewcommand{\thefigure}{S.\arabic{figure}}
\renewcommand{\thetable}{S.\arabic{table}}
\setcounter{section}{0}
\setcounter{equation}{0}
\setcounter{theorem}{0}
\setcounter{fact}{0}
\setcounter{remark}{0}
\setcounter{example}{0}
\setcounter{figure}{0}
\setcounter{table}{0}
\fi

\begin{bibunit}

\section{Related literature\label{sec:related_lit}}

Our work mixes ingredients from different subfields of sequential and multiple hypothesis testing. 

The e-value has recently emerged as a fundamental concept in composite hypothesis testing and underlying a universal approach to anytime-valid inference \citep{wasserman2020universal,shafer2021testing,vovk2021values, grunwald2020safe,ramdas2020admissible}, but the roots can be traced back to the works of \citet{ville1939etude, wald1945sequential} and Robbins~\citep{darling1967confidence}. A recent overview of the e-value literature is given by \citet{ramdas2023game} and \citet[Chapter 1]{ramdas2024hypothesis}. 

Interest in e-values has grown rapidly in recent years, including in particular multiple testing with e-values. \citet{wang2022false} introduced and analyzed the e-BH procedure, an e-value variant of the popular Benjamini-Hochberg (BH) procedure \citep{benjamini1995controlling} for FDR control. \citet{vovk2021values} explored the possibility of combining several e-values by averaging and multiplication. They also used this to derive multiple testing procedures with familywise error rate (FWER) control by applying the closure principle \citep{marcus1976closed} with these combination rules. The FWER is a strict error criterion defined as the probability of rejecting any true null hypothesis. \citet{vovk2023confidence, vovk2024true} extended these ideas to obtain procedures with a true discovery guarantee. All the aforementioned approaches consider classical \emph{offline} multiple testing.

 \emph{Online} multiple testing   initially focused on procedures for p-values \citep{foster2008alpha, aharoni2014generalized, javanmard2018online, ramdas2017online}. An overview of this literature was recently provided by \citet{robertson2023online}. \citet{xu2024online} is the sole paper to consider online multiple testing with e-values, focusing on FDR control for dependent e-values.

A related line of work investigates simultaneous true discovery guarantees by closed testing, mostly focusing on offline settings with p-values. The closure principle was initially proposed and analyzed for FWER control \citep{marcus1976closed, sonnemann1982allgemeine, romano2011consonance}. However, \citet{goeman2011multiple} noted that the same principle can be applied to obtain simultaneous true discovery bounds, a more general and less conservative task than FWER control, although a similar approach was proposed earlier by \citet{genovese2004stochastic,genovese2006exceedance}. Many works have since built on these
\citep{goeman2019simultaneous,vesely2023permutation,hemerik2018false,li2024simultaneous}. 
Importantly, \citet{goeman2021only} proved that all admissible procedures for bounding the true discovery proportion must employ closed testing. 

The  recent work of \citet{fischer2024online} showed how the closure principle can  also be used for \emph{online} multiple testing. However, their investigation of admissibility and construction of concrete procedures is restricted to FWER control. The current work extends their ideas to lower bounds on the true discovery proportion.

Another related work to ours is by \citet{katsevich2020simultaneous}. They proposed various p-value based true discovery procedures for structured, knockoff and also online settings exploiting martingale techniques. \citet{iqraa2024false} modified and improved some of their methods with a focus on m-consistency, a property that relates true discovery procedures to FDR. Our work uniformly improves the methods by \citet{katsevich2020simultaneous} and \citet{iqraa2024false} for the online setting.

In Table~\ref{tab:messages}, we give a brief comparison of our theoretical results to the related works. Our paper is the first connecting online multiple testing to anytime-valid testing of a single hypothesis. Furthermore, our setting generalizes the ones considered by \citet{goeman2021only} (since the online setting generalizes the offline setting) and \citet{fischer2024online} (since true discovery guarantee generalizes FWER control). Consequently, our results immediately apply in their settings, but their results do not imply ours. 

\begin{table}[h!]
    \centering
    \begin{tabular}{ll}
    \hline
        \textbf{Paper} & \textbf{Key theoretical messages} \\ \hline
           \citet{goeman2021only}
        & \makecell[l]{\tabitem Admissible closed testing is sufficient for admissibility of \\ 
        \hphantom{\tabitem}monotone  stacks of true discovery procedures  \\ \tabitem All admissible true discovery procedures are closed procedures} 
        \\ \hline
        \citet{fischer2024online}
        & \makecell[l]{\tabitem Increasing families of online intersection tests are sufficient \\ 
        \hphantom{\tabitem}to obtain online closed procedures \\ 
           \tabitem All admissible FWER controlling online procedures are \\ 
        \hphantom{\tabitem}online closed procedures} 
        \\ \hline
         \citet{ramdas2020admissible}
        & \makecell[l]{\tabitem Sufficient conditions for admissibility of anytime-valid tests \\ 
        \hphantom{\tabitem}based on test martingales \\ 
           \tabitem All admissible anytime-valid tests must \\ 
        \hphantom{\tabitem}employ test martingales} 
        \\ \hline
        This paper
        & \makecell[l]{\tabitem Online closed procedures must employ anytime-valid tests \\ 
        \hphantom{\tabitem}for the intersection hypotheses \\ 
           \tabitem All admissible coherent online true discovery procedures \\ \hphantom{\tabitem}are online closed procedures.} 
        \\ \hline
    \end{tabular}
    \caption{Comparison of the key messages of closely related theoretical works. Our main claim, the first shown above, connects the work by \citet{ramdas2020admissible} with the online multiple testing literature \citep{fischer2024online, katsevich2020simultaneous, iqraa2024false} and thus allows to apply the powerful anytime-valid testing concept of test martingales to online multiple testing. 
    Furthermore, our second claim implies and unifies the second claims of \citet{goeman2021only} and \citet{fischer2024online}, since the online setting generalizes the offline setting and true discovery guarantee generalizes FWER control (and all admissible procedures are coherent in these settings).}
    \label{tab:messages}
\end{table}

\section{Simple example: Lower bounds for the number of false hypotheses\label{sec:example}}

A simple but interesting special case of the general online discovery process (see Figure~\ref{fig:flowchart_online}) occurs if we choose $S=\{1,\ldots,t\}$ at each time $t\in \mathbb{N}$. That means we observe a stream of hypotheses $H_1,H_2,\ldots$ with sequential e-values $E_1,E_2,\ldots$  and want to provide a real-time lower bound $\bd(\{1,\ldots,t\})=d_t$ for the number of false hypotheses among $H_1,\ldots,H_t$ which holds true with high probability simultaneously over all times $t$:
\begin{align}
\mathbb{P}(d_t\leq |\{i\leq t: H_i\text{ false}\}| \text{ for all }t\in \mathbb{N})\geq 1-\alpha \quad \text{ for some } \alpha\in [0,1]. \label{eq:bound_intro}
\end{align}

In this case, the \texttt{SeqE-Guard} algorithm consists at each time $t$ of two simple steps: (1) multiply the e-values up to step $t$; (2) if the product is greater or equal than $1/\alpha$, increase the lower bound by $1$ and exclude the largest e-value from the future analysis. 
More precisely, set $d_0=0$ and $A=\emptyset$, and then do for $t=1,2,\ldots:$
\begin{enumerate}
    \item Set $A=A\cup \{t\}$ and calculate $\Pi=\prod_{i\in A} E_i$. \label{bull:intro_1}
    \item If $\Pi\geq 1/\alpha$, then update $d_t=d_{t-1}+1$ and $A=A\setminus \{\text{index of largest e-value in }A\}$; otherwise, set $d_t=d_{t-1}$.
\end{enumerate}

For example, suppose $\alpha=0.05$ and the first five e-values are
$$
E_1=5,\quad E_2=4,\quad E_3=0.8, \quad E_4=0.5, \quad E_5=14.
$$
At time $t=2$ the product (of $E_1$ and $E_2$) equals $20$ and therefore we can set $d_2=1$ and then exclude $E_1$ from the future analysis. Then, at time $t=5$, the product (of $E_2, E_3,E_4$ and $E_5$) is again greater than $20$ and therefore we can increase the lower bound and set $d_5=2$. Hence, in this case we can confidently claim (with probability $0.95$) that there is at least one false hypothesis among $H_1$ and $H_2$ and at least two false hypotheses among $H_1,\ldots,H_5$. This claim remains valid regardless of how many hypotheses are tested in the future and what the e-values for these hypotheses look like. 


Of course the claims above are too imprecise for many applications, as they only state that two of the five hypotheses are false, but not which ones. However, this is only a simple example. In general, the users of our \texttt{SeqE-Guard} algorithm can specify (based on the data) any subset of hypotheses in which they are interested and the \texttt{SeqE-Guard} algorithm will provide a lower bound for the number of false hypotheses in the subset that is valid simultaneously over all times and possible subsets. For example, a user might only be interested in the number of false hypotheses among $H_1$, $H_2$ and $H_5$, since due to their small e-values $H_3$ and $H_4$ are unlikely to be false anyway. The \texttt{SeqE-Guard} algorithm would then still provide a lower bound of $2$ for this subset, since excluding e-values smaller than $1$ from the query set won't decrease the bound (see Algorithm~\ref{alg:general}). This bound would be much more informative than the same lower bound for all five hypotheses, although we cannot make claims about the truth of individual hypotheses either. However, note that this is the same for FDR procedures, as they just state that the expected proportion of true hypotheses among the rejected ones is below $\alpha$.

\section{Uniform improvements of existing methods}
\subsection{Uniform improvement of the \texttt{online-adaptive} method by \citet{katsevich2020simultaneous}\label{appn:online-adaptive}}

Let p-values $P_1,P_2,\ldots$ and significance levels $\alpha_1,\alpha_2, \ldots $ be defined as for the \texttt{online-simple} algorithm (see Section~\ref{sec:katse}), and the null p-values be valid conditional on the past. Furthermore, let $(\lambda_i)_{i\in \mathbb{N}}$ be additional parameters such that $\lambda_i\in [\alpha_i,1)$ is measurable with respect to $\mathcal{F}_{i-1}$ and $B:=\sup_{i\in \mathbb{N}} \frac{\alpha_i}{1-\lambda_i}<\infty$. The \texttt{online-adaptive} bound by  \citet{katsevich2020simultaneous}
$$d^{\text{ad}}(S_t)=\left\lceil-ca+ \sum_{i=1}^t \mathbbm{1}\{P_i\leq \alpha_i\} - c \frac{\alpha_i}{1-\lambda_i} \mathbbm{1}\{P_i>\lambda_i\} \right\rceil$$
provides simultaneous true discovery guarantee over all sets $S_t=\{i\leq t: P_i\leq \alpha_i\}$, $t\in \mathbb{N}$, where $a>0$ is some regularization parameter and $c=\frac{\log(1/\alpha)}{a\log(1+(1-\alpha^{B/a})/B)}$. Note that $c$ has a different value than for the \texttt{online-simple} algorithm. Similar as demonstrated in Section~\ref{sec:katse},  \citet{katsevich2020simultaneous} proved the error guarantee by showing that $E_i^{\text{ad}}=\exp[\theta_c(\mathbbm{1}\{P_i\leq \alpha_i\} - c\frac{\alpha_i}{1-\lambda_i}\mathbbm{1}\{P_i>\lambda_i\})]$, $i\in \mathbb{N}$, are sequential e-values, where $\theta_c=\log(1/\alpha)/(ca)$. Note that
$$
\prod_{i\in I} E_i^{\text{ad}}\geq 1/\alpha \Leftrightarrow  \sum_{i\in I}  \mathbbm{1}\{P_i\leq \alpha_i\} - c\frac{\alpha_i}{1-\lambda_i}\mathbbm{1}\{P_i>\lambda_i\} \geq \log(1/\alpha)/\theta_c = ca \quad (I\in 2^{\mathbb{N}_f}).
$$
With this, one can define a uniform improvement of the \texttt{online-adaptive} algorithm in the exact same manner as for the \texttt{online-simple} algorithm. Note that the \texttt{online-adaptive} method already adapts to the proportion of null hypotheses using the parameter $\lambda_i$ and therefore cannot be further improved by the (online) closure principle in that direction. However, it still leads to a real uniform improvement by transforming it into a coherent procedure. 

\begin{proposition}\label{prop:improvement_ad}
    The \texttt{SeqE-Guard} algorithm applied with the sequential e-values $(E_i^{\text{ad}})_{i\in \mathbb{N}}$ uniformly improves the \texttt{online-adaptive} method by \citet{katsevich2020simultaneous}.
\end{proposition}

Furthermore, the e-values $E_i^{\text{ad}}$ are inadmissible if $\alpha_i/(1-\lambda_i)$ is not constant for all $i\in \mathbb{N}$ and thus can be improved. For this, note that $E_i^{\text{ad}}=\exp(\theta_c) $, if $P_i\leq \alpha_i$,  $E_i^{\text{ad}}=1 $, if $\alpha_i<P_i\leq \lambda_i$ and $E_i^{\text{ad}}=\exp(-\theta_cc \alpha_i/(1-\lambda_i)) $, if $P_i>\lambda_i$. Hence, for all $\mathbb{P}\in H_i$, we can provide a tight upper bound for the expectation of $E_i^{\text{ad}}$ by 
\ifarxiv
\begin{align*}
\mathbb{E}_{\mathbb{P}}[E_i^{\text{ad}}|\mathbb{F}_{i-1}]&=\exp(\theta_c)\mathbb{P}(P_i\leq \alpha_i|\mathcal{F}_{i-1})\\ &+\mathbb{P}(\alpha_i<P_i\leq \lambda_i|\mathcal{F}_{i-1})+\exp\left(-\theta_cc \frac{\alpha_i}{1-\lambda_i}\right)\mathbb{P}(P_i>\lambda_i|\mathcal{F}_{i-1}) \\
&= (\exp(\theta_c)-1)\mathbb{P}(P_i\leq \alpha_i|\mathcal{F}_{i-1})+\mathbb{P}(P_i\leq \lambda_i|\mathcal{F}_{i-1})\left[1-\exp\left(-\theta_cc \frac{\alpha_i}{1-\lambda_i}\right)\right]\\ &+\exp\left(-\theta_cc \frac{\alpha_i}{1-\lambda_i}\right) \\
&\leq (\exp(\theta_c)-1)\alpha_i+\lambda_i\left[1-\exp\left(-\theta_cc \frac{\alpha_i}{1-\lambda_i}\right)\right]+\exp\left(-\theta_cc \frac{\alpha_i}{1-\lambda_i}\right)\\ &=:u(\alpha,\alpha_i,\lambda_i,B,a),
\end{align*}
\else
\begin{align*}
&\mathbb{E}_{\mathbb{P}}[E_i^{\text{ad}}|\mathcal{F}_{i-1}]\\&=\exp(\theta_c)\mathbb{P}(P_i\leq \alpha_i|\mathcal{F}_{i-1})+\mathbb{P}(\alpha_i<P_i\leq \lambda_i|\mathcal{F}_{i-1})+\exp\left(-\theta_cc \frac{\alpha_i}{1-\lambda_i}\right)\mathbb{P}(P_i>\lambda_i|\mathcal{F}_{i-1}) \\
&= (\exp(\theta_c)-1)\mathbb{P}(P_i\leq \alpha_i|\mathcal{F}_{i-1})+\mathbb{P}(P_i\leq \lambda_i|\mathcal{F}_{i-1})\left[1-\exp\left(-\theta_cc \frac{\alpha_i}{1-\lambda_i}\right)\right]+\exp\left(-\theta_cc \frac{\alpha_i}{1-\lambda_i}\right) \\
&\leq (\exp(\theta_c)-1)\alpha_i+\lambda_i\left[1-\exp\left(-\theta_cc \frac{\alpha_i}{1-\lambda_i}\right)\right]+\exp\left(-\theta_cc \frac{\alpha_i}{1-\lambda_i}\right)\\ &=:u(\alpha,\alpha_i,\lambda_i,B,a),
\end{align*}
\fi
which can easily be calculated for given $\alpha$, $ \alpha_i$, $\lambda_i$, $B$ and $a$. If $B=\alpha_i/(1-\lambda_i)$, then $u(\alpha,\alpha_i,\lambda_i,B,a)=1$. However, in practice $\alpha_i/(1-\lambda_i)$ may vary over time such that there are indices $j\in \mathbb{N}$ with $\alpha_j/(1-\lambda_j)<B$. In this case, the e-value $E_j^{\text{ad}}$ becomes conservative. For example, if $\alpha=0.1$, $\alpha_i=0.1$, $\lambda_i=0.5$, $B=0.4$ and $a=1$, then $u(\alpha,\alpha_i,\lambda_i,B,a)=0.966$.

\subsection{Uniform improvements of the methods by \citet{iqraa2024false}\label{sec:improvements_iqraa}}

We derive uniform improvements of the online procedures \texttt{u-online-simple} and \texttt{u-online-Freedman} of \citet{iqraa2024false}. It should be noted that the uniform improvements presented here are not instances of our \texttt{SeqE-Guard} algorithm, however, they are obtained by taking the union of many \texttt{SeqE-Guard} bounds or by taking the mean of many test martingales (instead of a single one) for each intersection test, and hence are closely related to \texttt{SeqE-Guard}. 

\subsubsection{Uniform improvement of the \texttt{u-online-simple} method\label{sec:u-online-simple}}

\citet{iqraa2024false} have introduced a modified version of the \texttt{online-simple} method by \citet{katsevich2020simultaneous}. In the following, we show how this method can be uniformly improved by e-value based online closed testing.

\sloppy Let p-values $P_1,P_2,\ldots$ and significance levels $\alpha_1,\alpha_2, \ldots $ be defined as for the \texttt{online-simple} algorithm (see Section~\ref{sec:katse}), and the null p-values be valid conditional on the past. Instead of fixing the parameter $a>0$ in advance, \citet{iqraa2024false} combined \texttt{online-simple} bounds for different parameters using a union bound. More precisely, for each $a\in \mathbb{N}$ let $\alpha(a)=\frac{6\alpha}{a^2\pi^2}$ and $c_a=\frac{\log(1/\alpha(a))}{a\log(1+\log(1/\alpha(a))/a)}$. By \eqref{eq:online-simple}, we have that
for all $\mathbb{P}\in \mathcal{P}$:
\begin{align}
    & \mathbb{P}(\bd_a^{\text{os}}(S_t)\leq |S_t\cap I_1^{\mathbb{P}}| \text{ for all } t\in \mathbb{N})\geq 1-\alpha(a), \label{eq:os-guarantee} \\ 
    & \text{ where } \bd_a^{\text{os}}(S_t)=\left\lceil-c_aa+ \sum_{i=1}^t \mathbbm{1}\{P_i\leq \alpha_i\} - c_a\alpha_i\right\rceil \quad (t\in \mathbb{N}), \label{eq:os-a} 
\end{align}
with $\alpha_1,\alpha_2,\ldots$ being nonnegative thresholds and $S_t=\{i\leq t: P_i\leq \alpha_i\}$. \citet{iqraa2024false} proposed  
\begin{align}
    \bd^{\text{u-os}}(S_t)=\max_{a\in \mathbb{N}} \bd_a^{\text{os}}(S_t)\quad t\in \mathbb{N},
\end{align}
whose true discovery guarantee follows by applying a union bound to \eqref{eq:os-guarantee}. We refer to the procedure $\bd^{\text{u-os}}$ as \texttt{u-online-simple} method in this paper.

In Section~\ref{sec:katse} we have shown that the bound in \eqref{eq:os-a} can be uniformly improved by Algorithm~\ref{alg:online-simple} for all $a>0$ and $\alpha(a)\in (0,1)$. Hence, a simple uniform improvement of the \texttt{u-online-simple} method can be obtained by taking the maximum of these improved bounds. 

More precisely, let $E_i^{\text{os},a}=\exp[\theta_{c_a}(\mathbbm{1}\{P_i\leq \alpha_i\}-c_a\alpha_i)]$
and define 
$
W_I^{t,a}\coloneqq \prod_{i\in I\cap \{1,\ldots,t\}} E_i^{\text{os},a}.
$
Then the \texttt{closed u-online-simple} method is given by $\bd^{\text{cu-os}}(S_t)=\max_{a\in \mathbb{N}} \bd^{\bphi^a}(S_t)$, where $\bd^{\bphi^a}$ is the online closed procedure obtained by the intersection tests 
$$
\phi_I^a\coloneqq \mathbbm{1}\left\{\exists t\in I: W_I^{t,a}\geq 1/\alpha(a)\right\} = \mathbbm{1}\left\{\exists t\in I: \frac{6}{a^2\pi^2} W_I^{t,a}\geq 1/\alpha\right\}, \quad a\in \mathbb{N}, I\subseteq \mathbb{N}.
$$
Since $\bd^{\bphi^a}(S_t)\geq \bd_a^{\text{os}}(S_t)$ for all $a\in \mathbb{N}$ (see Section~\ref{sec:katse}), we also have $\bd^{\text{cu-os}}(S_t)\geq \bd^{\text{u-os}}(S_t)$. 

However, we can improve $\bd^{\text{cu-os}}$ even further. For this, note that $\phi_I^a$ can be uniformly improved by
$$
\phi_I^{\text{m-os}}\coloneqq \mathbbm{1}\left\{\exists t\in I: \sum_{a\in \mathbb{N}} \frac{6}{a^2\pi^2} W_I^{t,a}\geq 1/\alpha\right\}, 
$$
simultaneously for all $a$. Since (weighted) means of test supermartingales are test supermartingales again, Ville's inequality implies that $\phi_I^{\text{m-os}}$ is an  intersection test and the online closed procedure $\bd^{\text{m-os}}\coloneqq \bd^{\bphi^{\text{m-os}}}$ provides simultaneous true discovery guarantee. We refer to $\bd^{\text{m-os}}$ as \texttt{m-online-simple} procedure in the following. Since $\phi_I^{\text{m-os}}\geq \phi_I^a$ for all $a\in \mathbb{N}$, we have
$$
\bd^{\text{m-os}}(S_t)\geq \bd^{\text{cu-os}}(S_t) \geq \bd^{\text{u-os}}(S_t) \text{ for all } t\in \mathbb{N}.
$$

Note that $E_i^{\text{os},a'}\geq E_j^{\text{os},a'}$ for some $a'\in \mathbb{N}$ and $i\neq j$ implies that $E_i^{\text{os},a}\geq E_j^{\text{os},a}$ for all $a\in \mathbb{N}$ and $E_i^{\text{os},a'}\geq 1$ for some $a'\in \mathbb{N}$ and $i\in \mathbb{N}$ implies that $E_i^{\text{os},a}\geq 1$ for all $a\in \mathbb{N}$. With this, we can derive a short-cut (see Algorithm~\ref{alg:mean-online-simple}) for $\bd^{\text{m-os}}(S_t)$, $t\in \mathbb{N}$, in the same manner as we did in Algorithms~\ref{alg:general}~and~\ref{alg:online-simple}. 

\begin{proposition}\label{prop:improvement_uos}
    The \texttt{m-online-simple} procedure (Algorithm~\ref{alg:mean-online-simple}) and the \texttt{cu-online-simple} procedure uniformly improve the \texttt{u-online-simple} method by \citet{iqraa2024false}.
\end{proposition}

\begin{algorithm}
\begin{flushleft} \caption{\texttt{m-online-simple}} \label{alg:mean-online-simple}
 \textbf{Input:} Sequence of p-values $P_1,P_2,\ldots$ and sequence of (potentially data-dependent) individual significance levels $\alpha_1,\alpha_2,\ldots$ .\\ 
 \textbf{Output:} Query sets $S_1\subseteq S_2 \subseteq \ldots$ and true discovery bounds $d_1\leq d_2 \leq \ldots$ .
 \end{flushleft}
\begin{algorithmic}[1]
\State $d_0=0$
\State $S_0=\emptyset$
\State $A_0^c=\emptyset$
\For{$t=1,2,\ldots$}
\If{$P_t\leq \alpha_t$}
\State $S_t=S_{t-1}\cup \{t\}$
\Else
\State $S_t=S_{t-1}$
\EndIf
\If{$\sum_{a\in \mathbb{N}} \frac{6}{a^2\pi^2} \prod_{i\in \{1,\ldots,t\}\setminus A_{t-1}^c} E_i^{\text{os},a}  \geq  1/\alpha$} 
\State $d_t=d_{t-1}+1$
\State $A_{t}^c=A_{t-1}^c\cup \{\text{index of smallest individual significance level in } S_t\setminus A_{t-1}^c\}$
\Else
\State $d_{t}=d_{t-1}$
\State $A_{t}^c=A_{t-1}^c$
\EndIf
\State \Return $S_t, d_t$
\EndFor
\end{algorithmic}
\end{algorithm}

\sloppy In Figure~\ref{fig:sim_m-os}, we compare the \texttt{u-online-simple} method \citep{iqraa2024false} with our \texttt{m-online-simple} procedure (Algorithm~\ref{alg:mean-online-simple}). The simulation setup is the same as described in Section~\ref{sec:sim}. The plots show that the improvement obtained by the \texttt{m-online-simple} procedure is substantial.

\begin{figure}[tb]
\centering
\includegraphics[width=0.8\textwidth]{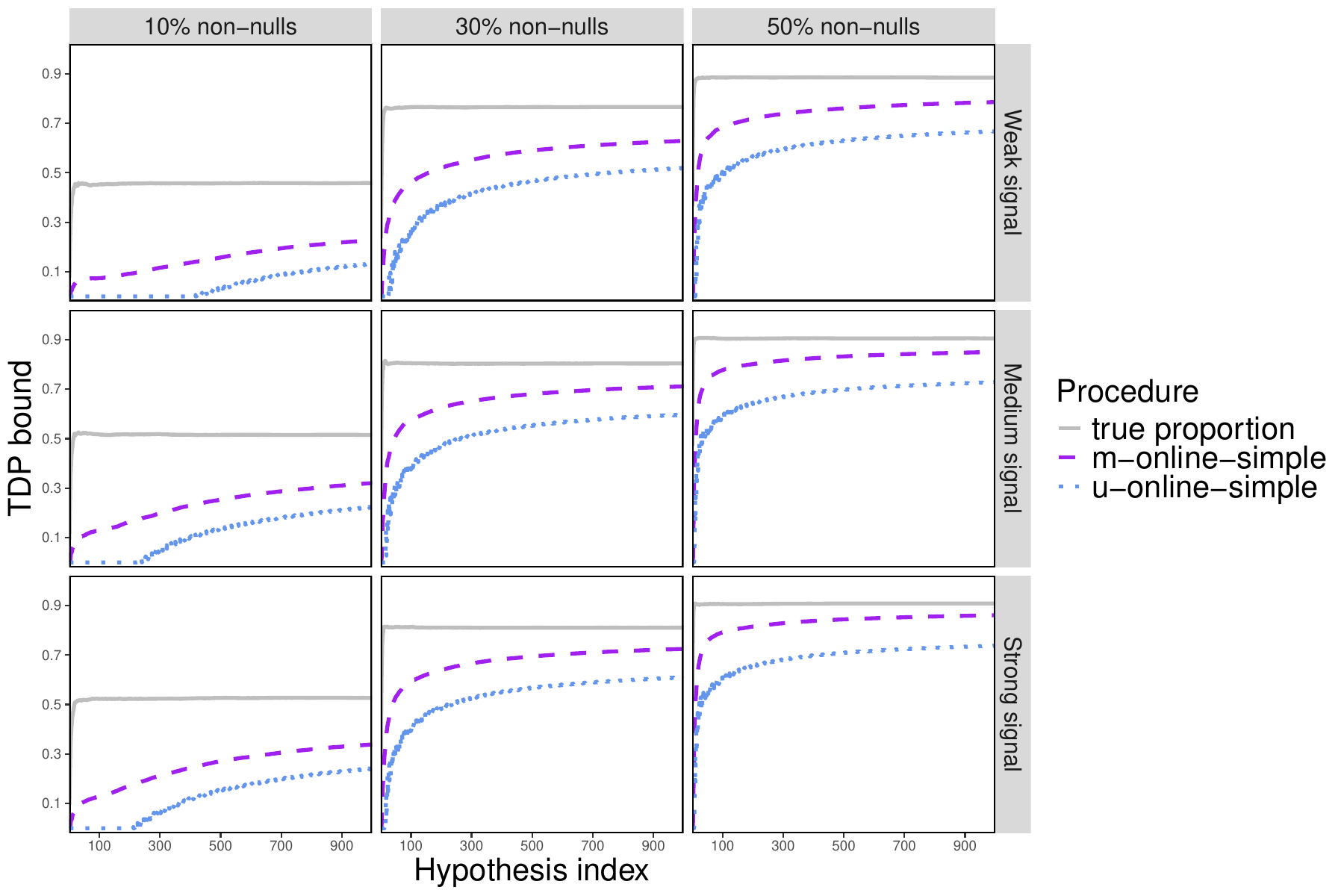}
\caption{True discovery proportion bounds obtained by the \texttt{u-online-simple} method \citep{iqraa2024false} and our \texttt{m-online-simple} method. The \texttt{m-online-simple} procedure improves the \texttt{u-online-simple} method substantially. \label{fig:sim_m-os} }\end{figure}

\subsubsection{Uniform improvement of the \texttt{u-online-Freedman} method\label{sec:online-freedman}}

Similarly as the \texttt{u-online-simple} method (Section~\ref{sec:u-online-simple}), the \texttt{u-online-Freedman} procedure is obtained by taking the (weighted) union of many individual bounds. However, instead of using the \texttt{online-simple} bounds by \citet{katsevich2020simultaneous}, the \texttt{u-online-Freedman} procedure is based on Freedman's inequality \citep{freedman1975tail}. 

\sloppy Let p-values $P_1,P_2,\ldots$ and significance levels $\alpha_1,\alpha_2, \ldots $ be defined as for the \texttt{online-simple} algorithm (see Section~\ref{sec:katse}), and the null p-values be valid conditional on the past. \citet{iqraa2024false} showed (see Corollary 41 in \citep{iqraa2024false}) that 
for all $\mathbb{P}\in \mathcal{P}$:
\begin{align}
    & \mathbb{P}(\bd_a^{\text{Freed}}(S_t)>  |S_t\cap I_1^{\mathbb{P}}| \text{ and } A_t\leq a \text{ for some } t\in \mathbb{N})\leq \alpha(a), \label{eq:freed-guarantee} \\ 
    & \text{ where } \bd_a^{\text{Freed}}(S_t)=1+\left\lfloor -\kappa_a + \sum_{i=1}^t \mathbbm{1}\{P_i\leq \alpha_i\} - \alpha_i \right\rfloor, \label{eq:freed-bound} \\ 
    & \text{ and } B_t=\sum_{i=1}^t \alpha_i(1-\alpha_i),
\end{align}
\sloppy with $S_t=\{i\leq t: P_i\leq \alpha_i\}$, $\kappa_a=\sqrt{2a\log(1/\alpha(a))} + \frac{\log(1/\alpha(a))}{2}$ and $\alpha(a)=\alpha\left(\frac{6}{\max(2\log_2(a),1)^2 (\pi^2 + 6)}\right)$ for some parameter $a=2^{j/2}, j\in \mathbb{N}\cup \{0\}$. The \texttt{u-online-Freedman} procedure is then obtained by applying a union to \eqref{eq:freed-guarantee} over all $j\in \mathbb{N}\cup \{0\}$ (see Corollary 42 in \citep{iqraa2024false}). In the following, we show that the \texttt{SeqE-Guard} algorithm allows to uniformly improve the bound $\bd_a^{\text{Freed}}$ for each $a$. 

 \citet{howard2020time} improved Freedman's inequality by a test martingale approach. Our improvement of the bound $\bd_a^{\text{Freed}}$ is based on the same technique, which additionally uses the \texttt{SeqE-Guard} algorithm. For this, we define the sequential e-values
 \begin{align}
     E_i^{\text{Freed},a}\coloneqq \exp(\lambda_a(\mathbbm{1}\{P_i\leq \alpha_i\}-\alpha_i)-\psi(\lambda_a)\alpha_i(1-\alpha_i)) \quad (i\in \mathbb{N}), \label{eq:freedman_e-val}
 \end{align}
 where $\lambda_a=\log\left(1+\frac{\kappa_a}{a}\right)$ and $\psi(\lambda_a)=\exp(\lambda_a)-\lambda_a-1$. To see that $E_i^{\text{Freed},a}$, $i\in \mathbb{N}$, are sequential e-values, define $X_i\coloneqq \mathbbm{1}\{P_i\leq \alpha_i\}-\alpha_i$ and note that for all $i\in I_0^{\mathbb{P}}$ it holds that $X_i\leq 1$, $\mathbb{E}_{\mathbb{P}}[X_i|\mathcal{F}_{i-1}]\leq 0$ and $\mathbb{E}_{\mathbb{P}}[X_i^2|\mathcal{F}_{i-1}]\leq \alpha_i(1-\alpha_i)$. Therefore, $\mathbb{E_{\mathbb{P}}}[\exp(\lambda_a X_i)|\mathcal{F}_{i-1}]\leq \exp(\psi(\lambda_a) \alpha_i(1-\alpha_i))$ (e.g., Lemma 6.7 in \citep{tropp2012user}). Now let $d_t^{\text{Freed},a}$ be the bound obtained by applying \texttt{SeqE-Guard} with $E_1^{\text{Freed},a}, E_2^{\text{Freed},a},\ldots\ .$ In the same manner as in \eqref{eq:improvement_os}, we get that
 \begin{align*}
 d_t^{\text{Freed},a}&\geq 1+ \left\lfloor -\frac{\log(1/\alpha(a))}{\lambda_a}+\sum_{i=1}^t\mathbbm{1}\{P_i\leq \alpha_i\}-\alpha_i - \frac{\psi(\lambda_a)\alpha_i(1-\alpha_i)}{\lambda_a} \right\rfloor \\
 &= 1 + \left\lfloor -\frac{\log(1/\alpha(a))+\psi(\lambda_a)B_t}{\lambda_a}+\sum_{i=1}^t\mathbbm{1}\{P_i\leq \alpha_i\}-\alpha_i\right\rfloor.
\end{align*}
Further, if $B_t\leq a$, we obtain
\begin{align*}
    &-\frac{\log(1/\alpha(a))+\psi(\lambda_a)B_t}{\lambda_a} \\&\geq  -\frac{\log(1/\alpha(a))+\psi(\lambda_a)a}{\lambda_a}\\
    &=  -\frac{\log(1/\alpha(a))+\kappa_a-\log\left(1+\frac{\kappa_a}{a}\right)(a+\kappa_a)+\log\left(1+\frac{\kappa_a}{a}\right)\kappa_a}{\log\left(1+\frac{\kappa_a}{a}\right)} \\
    &\geq -\frac{\log(1/\alpha(a))-2(\sqrt{a+\kappa_a}-\sqrt{a})^2+\log\left(1+\frac{\kappa_a}{a}\right)\kappa_a}{\log\left(1+\frac{\kappa_a}{a}\right)}, \\
\end{align*}

where the second inequality follows by Lemma~43 of \citet{iqraa2024false}. This can be further simplified by

\begin{align*}
     & -\frac{\log(1/\alpha(a))-2(\sqrt{a+\kappa_a}-\sqrt{a})^2+\log\left(1+\frac{\kappa_a}{a}\right)\kappa_a}{\log\left(1+\frac{\kappa_a}{a}\right)} \\ &= -\frac{\log(1/\alpha(a))-2(\sqrt{a+\kappa_a}-\sqrt{a})^2}{\log\left(1+\frac{\kappa_a}{a}\right)} - \kappa_a 
    \\ &= -\frac{\log(1/\alpha(a))-2\left(\sqrt{a+\sqrt{2a\log(1/\alpha(a))} + \frac{\log(1/\alpha(a))}{2}}-\sqrt{a}\right)^2}{\log\left(1+\frac{\kappa_a}{a}\right)} - \kappa_a \\
    &= -\frac{\log(1/\alpha(a))-2\left(\sqrt{\left(\sqrt{a} + \sqrt{\frac{\log(1/\alpha(a))}{2}}\right)^2}-\sqrt{a}\right)^2}{\log\left(1+\frac{\kappa_a}{a}\right)} - \kappa_a \\
    &= -\frac{\log(1/\alpha(a))-2\left(\sqrt{\frac{\log(1/\alpha(a))}{2}}\right)^2}{\log\left(1+\frac{\kappa_a}{a}\right)} - \kappa_a \\
    &= -\kappa_a,
\end{align*}
which shows that $d_t^{\text{Freed},a} \geq \bd_a^{\text{Freed}}(S_t)$ for all $a\geq B_t$ and we additionally improve \eqref{eq:freed-guarantee} by providing a nontrivial bound for $a<B_t$.  Hence, a uniform improvement of the \texttt{u-online-Freedman} procedure can either be obtained by taking the maximum of the improved bounds $d_t^{\text{Freed},a}$ over all $a=2^j$, $j\in \mathbb{N}\cup \{0\}$, or by the more powerful mean strategy described in Supplementary Material~\ref{sec:u-online-simple}. In line with Supplementary Material~\ref{sec:u-online-simple}, we call these two procedures \texttt{cu-online-Freedman} and \texttt{m-online-Freedman}, respectively.

\begin{proposition}\label{prop:improvement_freed}
    The \texttt{m-online-Freedman} procedure and the \texttt{cu-online-Freedman} procedure uniformly improve the \texttt{u-online-Freedman} method by \citet{iqraa2024false}.
\end{proposition}

Note that the sequential e-values $E_i^{\text{Freed},a}$, $i\in \mathbb{N}$, are just binary e-values depending on $\mathbbm{1}\{P_i\leq \alpha_i\}$. Hence, they are very similar to the ones defined for the \texttt{online-simple} method \eqref{eq:e-val_os}, the difference lies just in the weighting of the two cases $P_i\leq \alpha_i$ and $P_i>\alpha_i$. Therefore, the \texttt{SeqE-Guard} algorithm does not only improve these methods, it also facilitates their interpretation. We only need to ensure that the expected value of each sequential e-value is bounded by one under the null hypothesis, which is very easy to check. Indeed, one can check that $E_i^{\text{Freed},a}$ is slightly conservative as well and could be improved, we just chose the representation \eqref{eq:freedman_e-val} as it simplifies the proof of the uniform improvement. Such looseness would be difficult to detect with the procedures by \citet{katsevich2020simultaneous} and \citet{iqraa2024false}, as their proofs of validity are based on far more complicated arguments.

In Figure~\ref{fig:sim_m-Freed}, we compare the \texttt{u-online-Freed} method \citep{iqraa2024false} with our \texttt{m-online-Freed} procedure. The simulation setup is the same as described in Section~\ref{sec:sim}. The behavior of the \texttt{u-online-Freed} and \texttt{m-online-Freed} method is similar as for the \texttt{u-online-simple} and \texttt{m-online-simple} method (see Figure~\ref{fig:sim_m-os}). The \texttt{m-online-Freed} method leads to a significant uniform improvement over the \texttt{u-online-Freed} method.

\begin{figure}[tb]
\centering
\includegraphics[width=0.8\textwidth]{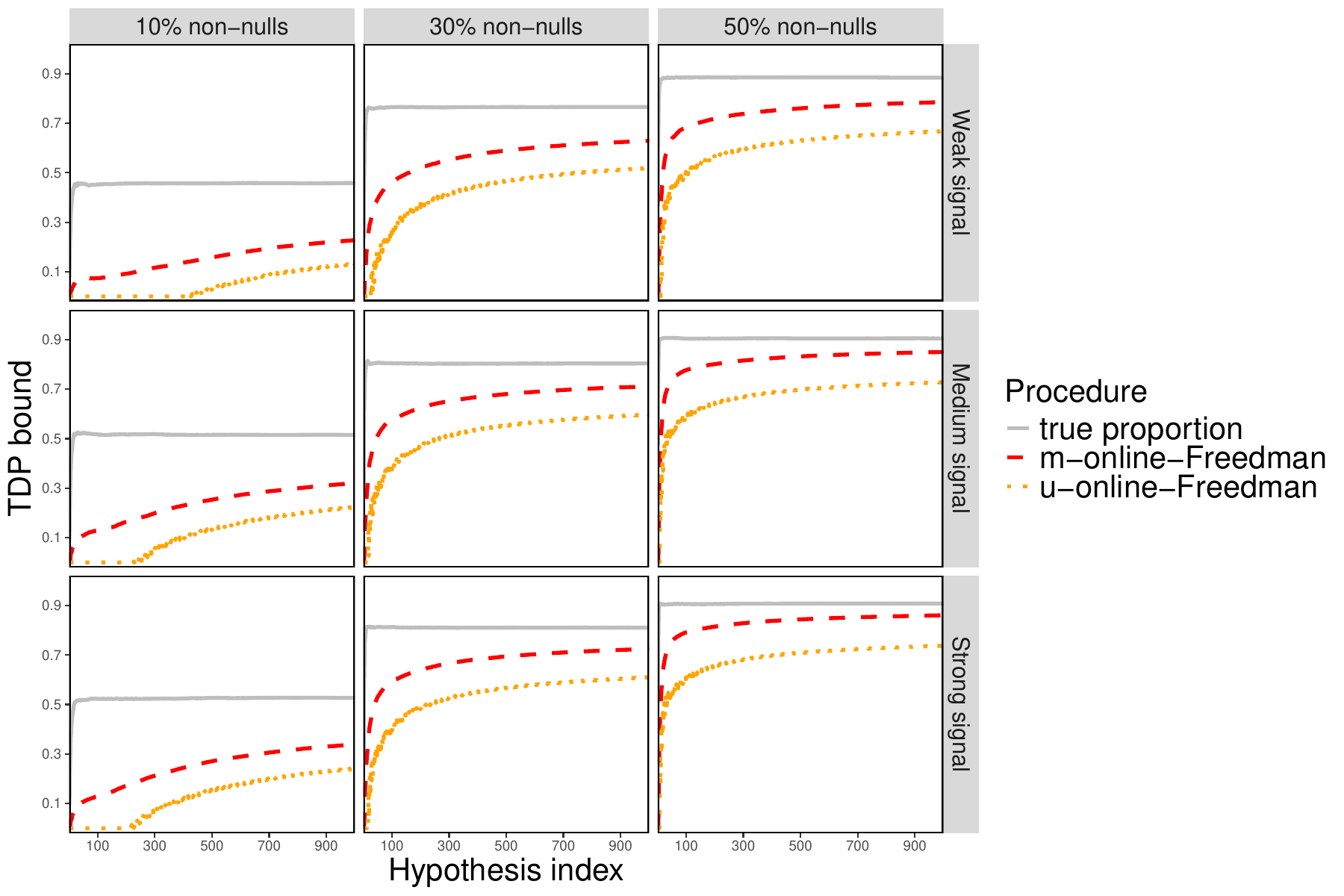}
\caption{True discovery proportion bounds obtained by the \texttt{u-online-Freed} method \citep{iqraa2024false} and our \texttt{m-online-Freed} method. The \texttt{m-online-Freed} procedure improves the \texttt{u-online-Freed} method substantially. \label{fig:sim_m-Freed} }\end{figure}

\section{Calibrating sequential p-values into e-values\label{sec:calib}}

While e-values have been of recent interest, most studies still use p-values as measure of evidence against the null hypothesis. In this case, we can \emph{calibrate} the p-values into e-values and then apply \texttt{SeqE-Guard}. In this context, a calibrator is a function from p-values to e-values, meaning it takes as an input a p-value and yields an e-value as an output. 

For example, the \texttt{online-simple} method of \citet{katsevich2020simultaneous} introduced in Section~\ref{sec:katse} is implicitly based on calibrating each p-value into a simple binary e-value. However, there are infinitely many other calibrators that could be used. Assume that p-values $(P_t)_{t\in \mathbb{N}}$ for the hypotheses $(H_t)_{t\in \mathbb{N}}$ are given. A decreasing function $f:[0,1]\rightarrow [0,\infty]$ is a calibrator if $\int_0^1 f(x)dx\leq 1$, and it is admissible if equality holds~\cite{vovk1993logic,vovk2021values}. Note that if the p-values are sequential p-values, meaning $P_t$ is measurable with respect to $\mathcal{F}_{t}$ and $\mathbb{P}(P_t\leq x|\mathcal{F}_{t-1})\leq x$ for all $x\in [0,1]$, then the calibrated e-values are sequential e-values.


An example calibrator is 
\begin{align}
h_x(p)=\exp\left(x \Phi^{-1}(1-p)-x^2/2\right), \label{eq:calibrator_duan}
\end{align}
where $\Phi$ denotes the CDF of a standard normal distribution. To see that this is a valid calibrator, note that $\Phi^{-1}(U)$, where $U\sim [0,1]$, follows a standard normal distribution and the moment generating function of a standard normal distribution for the real parameter $x> 0$ is given by $\exp(x^2/2)$. \citet{duan2020interactive} used this calibrator to derive a martingale Stouffer \citep{stouffer1949american} global test. 
This global test is based on a confidence sequence of \citet{howard2020time}. 

 We compare the application of \texttt{SeqE-Guard} with calibrated e-values and hedged GRO e-values experimentally in the next section.

\section{Boosting of sequential e-values\label{sec:boosting}}

\citet{wang2022false} introduced a way to \emph{boost} e-values before plugging them into their e-BH procedure without violating the desired FDR control. In this section, we propose a similar (but simpler) approach for \texttt{SeqE-Guard} that will improve its power for a specific query path. 

We begin by noting that whenever the bound $d_t$ is increased by one, the largest e-value with index in $A$ will not be considered in the following analysis. Hence, extremely large e-values will be excluded from the analysis anyway, which makes it possible to truncate the e-values at a specific threshold without changing its outcome. This makes the resulting truncated e-values conservative under the null, and one can improve the procedure by  multiplying the e-value by a suitable constant larger than 1 to remove its conservativeness. This truncation+multiplication operation is what is referred to as \emph{boosting} the e-value.

To this end, recall that $A_{t-1}\subseteq S_{t-1}$, $t\in \mathbb{N}$, is the index set of all previous e-values in the query set that were not already excluded by \texttt{SeqE-Guard} and $U_{t-1}\subseteq \{1,\ldots, t-1\}\setminus S_{t-1}$ the index set of all previous e-values that are not contained in the query set and that are smaller than $1$. Now, define \begin{align}
m_t:=\max\left\{\max_{i\in A_{t-1}} E_i, \frac{1}{\alpha \prod_{i\in A_{t-1}\cup U_{t-1}} E_i}\right\},
\label{eq:m_t}
\end{align}
and note that $m_t$ is predictable (measurable with respect to $\mathcal{F}_{t-1}$). Furthermore, if $E_t\geq m_t$ and $t\in S_t$, then $d_t=d_{t-1}+1$ and $E_t$ will be excluded in the further analysis. If $E_t\geq m_t$ and $t\notin S_t$, then $E_t\geq 1$, since $m_t\geq 1$ by definition, and $E_t$ won't be considered in the analysis anyway.
Hence, we define the truncation function $T_t:[0,\infty]\rightarrow [0,m_t]$ as
\begin{align}
T_t(x):=x\mathbbm{1}\{x\leq m_t\} + m_t\mathbbm{1}\{x>m_t\} \label{eq:truncation}
\end{align}
and then choose a boosting factor $b_t\geq 1$ as large as possible such that
\begin{align}
\mathbb{E}_{\mathbb{P}}[T_t(b_tE_t)|\mathcal{F}_{t-1}]\leq 1 \text{ for all } \mathbb{P}\in H_t. \label{eq:expected_boosted}
\end{align}
Note that $b_t= 1$ always satisfies \eqref{eq:expected_boosted}; so a boosting factor always exists and is always at least one. Condition \eqref{eq:expected_boosted} immediately implies that $T_t(b_tE_t)$ is a sequential e-value. Furthermore, using $b_t E_t$ in \texttt{SeqE-Guard} yields exactly the same results as using  $T_t(b_tE_t)$. Therefore, applying \texttt{SeqE-Guard} to the boosted e-values $(b_tE_t)_{t\in \mathbb{N}}$ provides simultaneous true discovery guarantee and is uniformly more powerful than with non-boosted e-values, since $b_t\geq 1$. Note that in this case $m_t$ should also be calculated based on the boosted e-values $b_1E_1,\ldots, b_{t-1}E_{t-1}$.  As also mentioned by \citet{wang2022false}, one could use different functions than $x\mapsto bx$ for some $b\geq 1$ to boost the e-values. In general, it is only required that each boosted e-value $E_t^{\mathrm{boost}}$, $t\in \mathbb{N}$, satisfies $\mathbb{E}_{\mathbb{P}}[T_t(E_t^{\mathrm{boost}})|\mathcal{F}_{t-1}]\leq 1$ for all $\mathbb{P}\in H_t$. We summarize this result in the following theorem.

\begin{proposition}\label{prop:boosting}
 Let $E_1^{\mathrm{boost}}, E_2^{\mathrm{boost}}, \ldots$ be a sequence of nonnegative random variables such that $\mathbb{E}_{\mathbb{P}}[T_t(E_t^{\mathrm{boost}})|\mathcal{F}_{t-1}]\leq 1$ for all  $\mathbb{P}\in H_t$, where $T_t$ is given by \eqref{eq:truncation} and $m_t$ is calculated as in \eqref{eq:m_t} based on $E_1^{\mathrm{boost}},\ldots, E_{t-1}^{\mathrm{boost}}$. Then, applying \texttt{SeqE-Guard} to the boosted sequential e-values $E_1^{\mathrm{boost}}, E_2^{\mathrm{boost}},\ldots$ provides simultaneous true discovery guarantee.
 \end{proposition}



It should be noted that boosting adapts the sequential e-values to the query path 
$(S_t)_{t\geq 1}$ and Proposition~\ref{prop:boosting} only holds for the query path for which the boosted e-values $E_1^{\mathrm{boost}}, E_2^{\mathrm{boost}}, \ldots$ were constructed. If multiple query paths are investigated, one must use the sequential e-values $T_1(E_1^{\mathrm{boost}}), T_2(E_2^{\mathrm{boost}}), \ldots$ for all other query paths. Therefore, it may happen that boosting weakens the performance of \texttt{SeqE-Guard} for the other query paths.

 In the following we provide several examples that illustrate how the boosting factors can be determined in specific cases and demonstrate the possible gain in efficiency.

\begin{example}\label{example:boost_GRO}
 We consider Example 3 from \citet{wang2022false} adapted to our setting. For each $t\in \mathbb{N}$, we test the simple null hypothesis $H_t:X_t|\mathcal{F}_{t-1}\sim \mathcal{N}(\mu_0,1)$ against the simple alternative $H_t^A:X_t|\mathcal{F}_{t-1}\sim \mathcal{N}(\mu_1,1)$, where $X_t$ denotes the data for $H_t$. In this case, the GRO e-value is given by the likelihood ratio between two normal distributions with variance $1$ and means $\mu_1$ and $\mu_0$
 \begin{align}
 E_t=\exp(\delta Z_t-\delta^2/2),\label{eq:log_normal_e-value}
 \end{align}
 where $\delta=\mu_1-\mu_0>0$ and $Z_t=X_t-\mu_0$ follows a standard normal distribution conditional on $\mathcal{F}_{t-1}$ under $H_t$. Hence, conditional on the past, each null e-value follows a log-normal distribution with parameter $(-\delta^2/2, \delta)$. With this, we obtain for all $t\in I_0^{\mathbb{P}}$: 
\begin{align*}
    & \mathbb{E}_{\mathbb{P}}[b_tE_t\mathbbm{1}\{b_tE_t\leq m_t\}+m_t \mathbbm{1}\{b_t E_t>m_t\}| \mathcal{F}_{t-1}]\\
    &=b_t\mathbb{E}_{\mathbb{P}}[E_t\mathbbm{1}\{b_tE_t\leq m_t\} | \mathcal{F}_{t-1}] +m_t \mathbb{P}(b_t E_t>m_t| \mathcal{F}_{t-1})\\
    &= b_t\left[1-\Phi\left(\frac{\delta}{2}-\frac{\log\left(m_t/b_t\right)}{\delta}\right)\right] +m_t \left[1-\Phi\left(\frac{\log(m_t/b_t)+\delta^2/2}{\delta}\right)\right],
\end{align*}
where $\Phi$ is the CDF of a standard normal distribution. The last expression can be set equal to $1$ and then be solved for $b_t$ numerically. For example, for $\delta=3$ and $m_t=20$, we obtain $b_t=3.494$. Hence, the e-value $E_t$ could be multiplied by $3.494$ without violating the true discovery guarantee, a substantial gain. In general, the larger $m_t$,  the smaller is the boosting factor. For example, if $m_t=5$, then $b_t=11.826$ and if $m_t=100$, then $b_t=1.774$. Nevertheless, even the latter boosting factor would increase the power of the true discovery procedure significantly and we would usually expect $m_t$ to be smaller than $100$ in most settings. If we use, as described in Section~\ref{sec:gro}, the e-value $1-\lambda_t+\lambda_tE_t^{\mathrm{GRO}}$, $\lambda_t\in (0,1)$ instead, we need to solve 
\begin{align*}
m_t+\Phi\left(\frac{\log(s_t)+\delta^2/2}{\delta}\right)[b_t(1-\lambda_t)-m_t]+b_t\lambda_t\left[1-\Phi\left(\frac{\delta}{2}-\frac{\log\left(s_t\right)}{\delta}\right)\right]=1, 
\end{align*}
for $b_t\in [1,1/(1-\lambda_t))$, where $s_t=(\lambda_t-1+m_t/b_t)/\lambda_t$. In this case, 
$\delta=3$, $\lambda_t=0.5$ and  $m_t=20$ yield a boosting factor of $b_t=1.354$. 

\end{example}

\begin{example}\label{example:calibrator}
    Suppose we observe sequential p-values $P_1,P_2,\ldots$ and want to apply the calibrator \eqref{eq:calibrator_duan}. If the p-values are uniformly distributed conditional on the past, the resulting e-value has the exact same distribution as the e-value in \eqref{eq:log_normal_e-value} under the null hypothesis for $\delta=x$, where $x$ is the freely chosen parameter for the calibrator. Hence, we can do the exact same calculations to obtain an appropriate boosting factor. If the p-values are stochastically larger than uniform, we could still use that same boosting factor, as the resulting e-values provide true discovery guarantee but might be conservative.   
\end{example}

\begin{example}
    In case of the \texttt{closed online-simple} method (Algorithm~\ref{alg:online-simple}) it is particular simple to \enquote{boost} the e-values. Since $E_i^{\mathrm{os}}$ only takes two different values, we can simply ensure $E_t^{\mathrm{os}}\leq m_t$ by choosing $\alpha_t$ such that $E_t^{\mathrm{os}}\leq m_t$ if $\mathbbm{1}\{P_t\leq \alpha_t\}=1$. Note that in case of $\alpha_t=\nu$ for all $t\in \mathbb{N}$ and some $\nu>0$ such that $\exp[\theta_c(1-c\nu)]\leq 1/\alpha$ it is not possible to improve the bounds of the \texttt{closed online-simple} method further by boosting, since we already have $E_t^{\mathrm{os}}\leq m_t$ almost surely.
\end{example}

In Figure~\ref{fig:sim_boost} we compare the \texttt{admissible online-simple} method with the \texttt{SeqE-Guard} algorithm with boosted GRO (the GRO e-values were first hedged and then boosted as described in Example~\ref{example:boost_GRO}) and calibrated e-values (using the calibrator in \eqref{eq:calibrator_duan} with $x=0.1$ and boosting as described in Example~\ref{example:calibrator}). The calibrated e-values perform worst in most cases, boosted GRO e-values perform best if the proportion of false hypotheses is small and the \texttt{admissible online-simple method} performs best if the proportion of false hypotheses is large. Also note that compared to Figure~\ref{fig:sim_best}, the bound of the hedged GRO e-values increased substantially using boosting. However, it adapts the e-values to the query paths and thus the bounds might decrease for other query paths.

\begin{figure}[tb]
\centering
\includegraphics[width=0.8\textwidth]{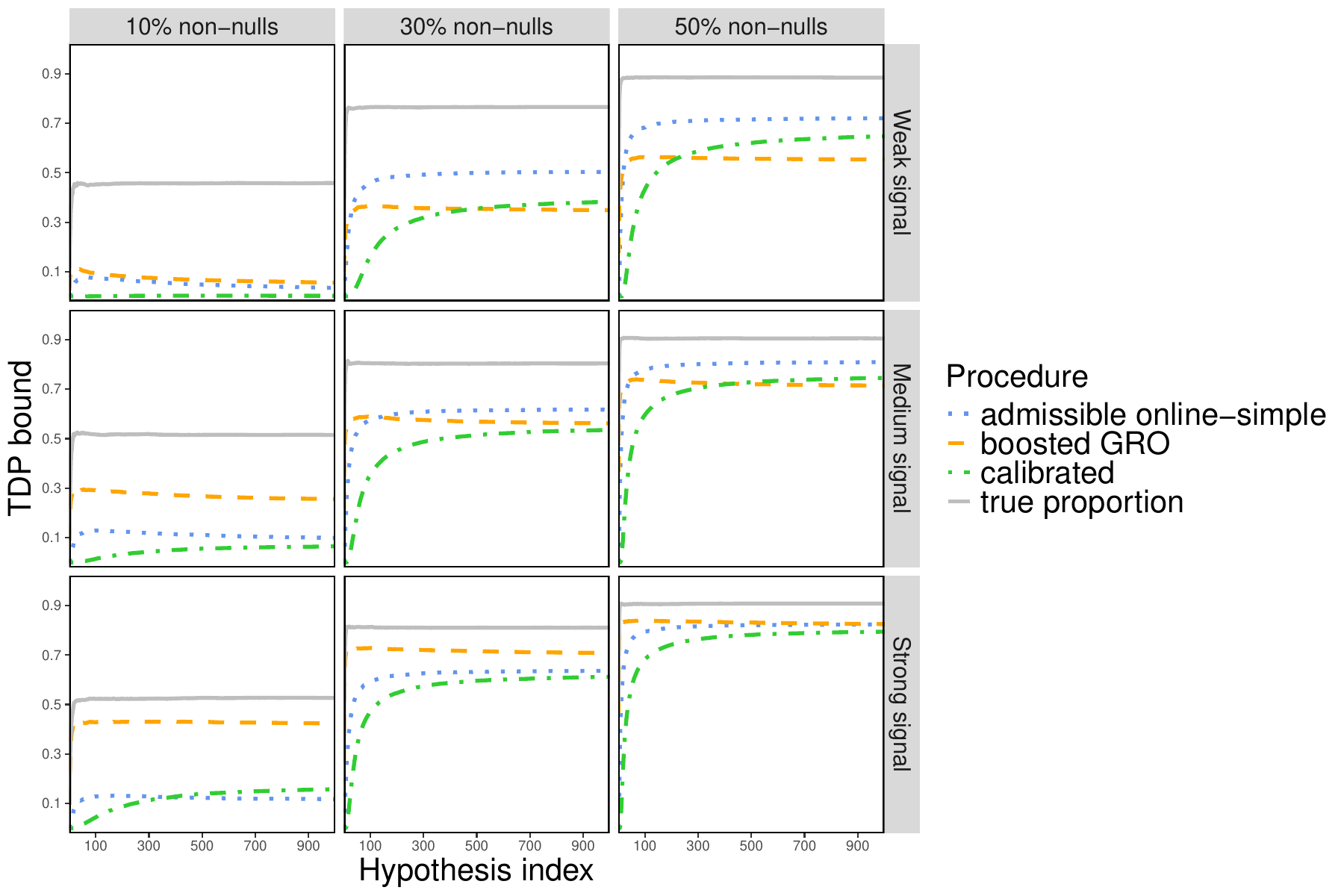}
\caption{True discovery proportion bounds obtained by the \texttt{admissible online-simple} as well as the \texttt{SeqE-Guard} algorithm with boosted GRO and (boosted) calibrated e-values. The boosted GRO e-values perform best if the proportion of false hypotheses is small and the \texttt{admissible online-simple} method performs best if the proportion is large. \label{fig:sim_boost} }\end{figure}

\section{Analysis of IMPC data\label{sec:IMPC}}
\ifdiff 
\added{The International Mouse Phenotyping Consortium (IMPC) \citep{mouse} is a large-scale effort to systematically characterize the phenotypic consequences of knocking out each protein-coding gene in the mouse genome. Since new genes are analyzed and released over time, the resulting sequential structure naturally gives rise to an online multiple testing problem, making the IMPC dataset a standard application for evaluating online procedures \citep{robertson2023online}. In this section, we apply the proposed methodology to $\numprint{172081}$ p-values that resulted from an evaluation by \citet{karp2017prevalence} and that are available at the Zenodo repository \url{zenodo.org/records/2396572} \citep{robertson2019onlinefdr}.}
\else
The International Mouse Phenotyping Consortium (IMPC) \citep{mouse} is a large-scale effort to systematically characterize the phenotypic consequences of knocking out each protein-coding gene in the mouse genome. Since new genes are analyzed and released over time, the resulting sequential structure naturally gives rise to an online multiple testing problem, making the IMPC dataset a standard application for evaluating online procedures \citep{robertson2023online}. In this section, we apply the proposed methodology to $\numprint{172081}$ p-values that resulted from an evaluation by \citet{karp2017prevalence} and that are available at the Zenodo repository \url{zenodo.org/records/2396572} \citep{robertson2019onlinefdr}.
\fi

\ifdiff \added{In our analysis, we compare the \texttt{online-simple} procedure by \citet{katsevich2020simultaneous} with our \texttt{admissible online-simple} method when applied to the query path obtained by the LORD++ algorithm \citep{javanmard2018online, ramdas2017online}. The LORD++ algorithm is an online procedure with FDR control that produces the rejection set $S_{t}=\{i\leq t: P_i\leq \alpha_i\}$, where the nonnegative thresholds $\alpha_1,\alpha_2,\ldots$ are chosen such that
$$
\widehat{\text{FDP}}_t \coloneqq \frac{\sum_{i=1}^t \alpha_i}{\max(|S_t|,1)} \leq \alpha_{\text{FDR}}.
$$
The idea is to conservatively estimate $\text{FDP}(S_t)$ by $\widehat{\text{FDP}}_t$  and then choose the levels $\alpha_1,\alpha_2,\ldots$ such that $\sup_{t\in \mathbb{N}} \widehat{\text{FDP}}_t \leq \alpha_{\text{FDR}}$, where $\alpha_{\text{FDR}}$ is the predefined nominal FDR level \citep{ramdas2017online}. Interestingly, the bound of the \texttt{online-simple} method was constructed using the estimate $\widehat{\text{FDP}}_t$ for $\text{FDP}(S_t)$ \citep{katsevich2020simultaneous}.  The LORD++ algorithm controls the FDR at level $\alpha_{\text{FDR}}$ if the p-values are independent. Please note that this assumption is not necessarily fulfilled for the IMPC dataset and the following analysis is for illustrative purposes only. The independence assumption also ensures that the p-values are valid conditional on the past and thus that the \texttt{online-simple} procedure  and the \texttt{admissible online-simple} method provide simultaneous true discovery guarantee.}
\else
In our analysis, we compare the \texttt{online-simple} procedure by \citet{katsevich2020simultaneous} with our \texttt{admissible online-simple} method when applied to the query path obtained by the LORD++ algorithm \citep{javanmard2018online, ramdas2017online}. The LORD++ algorithm is an online procedure with FDR control that produces the rejection set $S_{t}=\{i\leq t: P_i\leq \alpha_i\}$, where the nonnegative thresholds $\alpha_1,\alpha_2,\ldots$ are chosen such that
$$
\widehat{\text{FDP}}_t \coloneqq \frac{\sum_{i=1}^t \alpha_i}{\max(|S_t|,1)} \leq \alpha_{\text{FDR}}.
$$
The idea is to conservatively estimate $\text{FDP}(S_t)$ by $\widehat{\text{FDP}}_t$  and then choose the levels $\alpha_1,\alpha_2,\ldots$ such that $\sup_{t\in \mathbb{N}} \widehat{\text{FDP}}_t \leq \alpha_{\text{FDR}}$, where $\alpha_{\text{FDR}}$ is the predefined nominal FDR level \citep{ramdas2017online}. Interestingly, the bound of the \texttt{online-simple} method was constructed using the estimate $\widehat{\text{FDP}}_t$ for $\text{FDP}(S_t)$ \citep{katsevich2020simultaneous}.  The LORD++ algorithm controls the FDR at level $\alpha_{\text{FDR}}$ if the p-values are independent. Please note that this assumption is not necessarily fulfilled for the IMPC dataset and the following analysis is for illustrative purposes only. The independence assumption also ensures that the p-values are valid conditional on the past and thus that the \texttt{online-simple} procedure  and the \texttt{admissible online-simple} method provide simultaneous true discovery guarantee.
\fi


\ifdiff \added{The results of our analysis are illustrated in Figure~\ref{fig:IMPC}. The FDP bounds were simply obtained by $\bq(S_t)=1-\frac{\bd(S_t)}{|S_t|}$, where $\bd$ is the lower bound for the number of true discoveries. Both online true discovery procedures were applied with the same parameters as in Section~\ref{sec:sim} and the LORD++ algorithm with the standard values from the onlineFDR \texttt{R} package \citep{robertson2019onlinefdr} but $\gamma_i=\frac{6}{\pi^2 i^2}$. Note that the level $\alpha$ for the true discovery procedures is to be distinguished from the FDR level $\alpha_{\text{FDR}}$ used to obtain the query path $(S_t)_{t\in \mathbb{N}}$ with the LORD++ algorithm. However, in this analysis both parameters were set to $0.1$.}
\else
The results of our analysis are illustrated in Figure~\ref{fig:IMPC}. The FDP bounds were simply obtained by $\bq(S_t)=1-\frac{\bd(S_t)}{|S_t|}$, where $\bd$ is the lower bound for the number of true discoveries. Both online true discovery procedures were applied with the same parameters as in Section~\ref{sec:sim} and the LORD++ algorithm with the standard values from the onlineFDR \texttt{R} package \citep{robertson2019onlinefdr} but $\gamma_i=\frac{6}{\pi^2 i^2}$. Note that the level $\alpha$ for the true discovery procedures is to be distinguished from the FDR level $\alpha_{\text{FDR}}$ used to obtain the query path $(S_t)_{t\in \mathbb{N}}$ with the LORD++ algorithm. However, in this analysis both parameters were set to $0.1$.
\fi

\ifdiff \added{The results show that the \texttt{admissible online-simple} method substantially improves the \text{online-simple} procedure by yielding lower FDP bounds. Furthermore, the claims made by the \texttt{admissible online-simple} procedure are more informative than those made by the LORD++ algorithm for two reasons.}
\begin{enumerate}
    \item \added{Since $\bq(S_t)\approx \alpha_{\text{FDR}}=0.1$ for $|S_t|\geq 1000$ for the \texttt{admissible online-simple} procedure $\bq$, we can claim that the probability of the $\text{FDP}(S_t)$ being below $0.1$ is greater than or equal to $0.9$. In contrast, the LORD++ algorithm only allows to conclude that the expectation of $\text{FDP}(S_t)$ is bounded by $0.1$. Even though the former statement does not imply the latter (and vice versa), the high probability statement often gives stronger confidence about the true $\text{FDP}(S_t)$ than the expectation bound.}
    \item \added{The probability guarantee of the \texttt{admissible online-simple} method holds simultaneously for all $t$. In particular, this allows to stop the testing process data-adaptively while the claim made in the previous point remains valid. In contrast, the LORD++ algorithm is only proven to control the FDR at fixed times $t$, and thus the FDR could be inflated at data-adaptive stopping times. }
\end{enumerate}
\else
The results show that the \texttt{admissible online-simple} method substantially improves the \text{online-simple} procedure by yielding lower FDP bounds. Furthermore, the claims made by the \texttt{admissible online-simple} procedure are more informative than those made by the LORD++ algorithm for two reasons.
\begin{enumerate}
    \item Since $\bq(S_t)\approx \alpha_{\text{FDR}}=0.1$ for $|S_t|\geq 1000$ for the \texttt{admissible online-simple} procedure $\bq$, we can claim that the probability of the $\text{FDP}(S_t)$ being below $0.1$ is greater than or equal to $0.9$. In contrast, the LORD++ algorithm only allows to conclude that the expectation of $\text{FDP}(S_t)$ is bounded by $0.1$. Even though the former statement does not imply the latter (and vice versa), the high probability statement often gives stronger confidence about the true $\text{FDP}(S_t)$ than the expectation bound.\label{bull:IMPC_1}
    \item The probability guarantee of the \texttt{admissible online-simple} method holds simultaneously for all $t$. In particular, this allows to stop the testing process data-adaptively while the claim made in the previous point remains valid. In contrast, the LORD++ algorithm is only proven to control the FDR at fixed times $t$, and thus the FDR could be inflated at data-adaptive stopping times. 
\end{enumerate}
\fi

\begin{figure}[tb]
\centering
\includegraphics[width=0.8\textwidth]{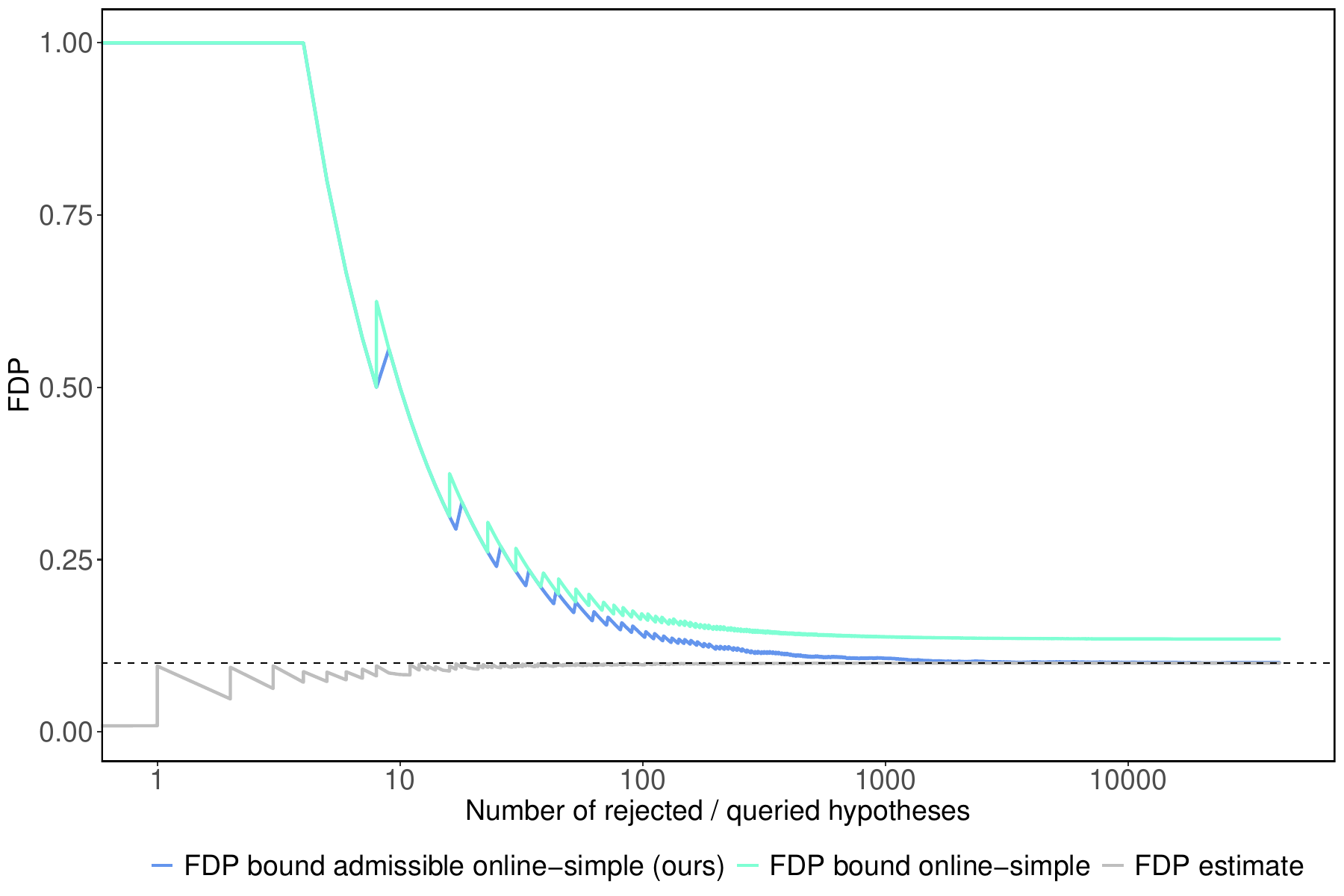}
\caption{FDP bounds obtained by the \texttt{online-simple} method \citep{katsevich2020simultaneous} and our \texttt{admissible online-simple} method applied on IMPC data. The \texttt{admissible online-simple} improves the  \texttt{online-simple} method substantially and is close to the nominal level FDR $\alpha_{\text{FDR}}$ of the LORD++ procedure. \label{fig:IMPC} }\end{figure}

\section{Proof of Theorem~\ref{theo:general_alg}\label{sec:proofs}}

\begin{proof}
 Let $d_t$, $t\in \mathbb{N}$, be the bounds of \texttt{SeqE-Guard}, $\Pi_{Z}$, $Z\subseteq \mathbb{N}$, be the product of all e-values with index in $Z$, $A_t\subseteq S_t$ \ifdiff \added{(resp. $A_{t+}\subseteq S_t$)} \else (resp. $A_{t+}\subseteq S_t$) \fi be the set $A$ at step $t$ before \ifdiff \added{(resp. after)} \else (resp. after) \fi checking whether $\Pi_{i\in A\cup U} E_i\geq 1/\alpha$, $U_t\subseteq \{1,\ldots,t\}\setminus S_t$ be the index set of e-values that are smaller than $1$. \ifdiff
 \added{Recall, that $\bd^{\bphi}(S)$ is solely determined by $\phi_I$ with $I\subseteq \{1,\ldots, \max(S)\}$, since $\bphi$ is increasing (see \eqref{eq:closed_procedure_increasing}).} \else Recall, that $\bd^{\bphi}(S)$ is solely determined by $\phi_I$ with $I\subseteq \{1,\ldots, \max(S)\}$, since $\bphi$ is increasing (see \eqref{eq:closed_procedure_increasing}). \fi
 Hence, we particularly have $d_1=\bd^{\bphi}(S_1)$. Now assume that $d_1=\bd^{\bphi}(S_1), \ldots, d_{t-1}=\bd^{\bphi}(S_{t-1})$ and $S_t=S_{t-1}\cup \{t\}$. Due to the coherence of $\bd^{\bphi}$, it holds $d_{t-1}\leq \bd^{\bphi}(S_{t})\leq d_{t-1}+1$. In the following, we show that $\Pi_{A_t \cup U_t}\geq 1/\alpha$ implies that $\bd^{\bphi}(S_{t}) \geq d_{t-1}+1$ and $\Pi_{A_t \cup U_t}< 1/\alpha$ implies that $\bd^{\bphi}(S_{t}) < d_{t-1}+1$, which proves the assertion.
 
 We first show that $\bd^{\bphi}(S_t) \geq d_{t-1}+1$, if $\Pi_{A_t \cup U_t}\geq 1/\alpha$. For this, we prove that $\Pi_{A_t \cup U_t}\geq 1/\alpha$ implies $\phi_I=1$ for all $I=V\cup W$, where $V\subseteq S_t$ and $W\subseteq \{1,\ldots,t\}\setminus S_t$ with $|V|\geq |A_t|$. Since $\bphi$ is increasing, and thus $\boldsymbol{d}^{\boldsymbol{\phi}}(S_t)$ can be written by \eqref{eq:closed_procedure_increasing}, this then implies that  $\boldsymbol{d}^{\boldsymbol{\phi}}(S_t)\geq |S_{t}|-|A_{t}|+1=d_{t-1}+1 $.
 First, note that it is sufficient to show the claim for all $I=V\cup U_t$ with $|V|\geq |A_t|$, since multiplication with e-values that are larger than or equal to $1$ cannot decrease the product. Now let such an $I$ and $V$ be fixed.
 \ifdiff \added{Because of definition \eqref{eq:intersection_test_mart} it is necessary and sufficient to find some $t'\leq t$ such that $\Pi_{I\cap \{1,\ldots,t'\}}\geq 1/\alpha$ (it will be shown that $t'=t_{\tilde{m}}$ does the trick, where $t_{\tilde{m}}$ is defined below).} \else Because of definition \eqref{eq:intersection_test_mart} it is necessary and sufficient to find some $t'\leq t$ such that $\Pi_{I\cap \{1,\ldots,t'\}}\geq 1/\alpha$ (it will be shown that $t'=t_{\tilde{m}}$ does the trick, where $t_{\tilde{m}}$ is defined below). \fi
 
 Let $t_1,\ldots, t_{m}$, where $t_m=t_{|S_t|-|A_t|+1}=t$, be the times at which $\Pi_{A_{t_i}\cup U_{t_i}}\geq 1/\alpha$ and $\tilde{m}\in \{1,\ldots, m\}$ be the smallest index such that $|V\cap \{1,\ldots, t_{\tilde{m}}\}|>|S_{t_{\tilde{m}}}|-\tilde{m}$. Note that $\tilde{m}$ always exists, because $|A_t|=|A_{t_m}|=|S_{t_m}|-m+1$. With this, we have 
 $$
 \Pi_{I\cap \{1,\ldots, t_{\tilde{m}}\}} \stackrel{(a)}{=} \Pi_{(V\cap \{1,\ldots, t_{\tilde{m}} \}) \cup U_{t_{\tilde{m}}}} \stackrel{(b)}{\geq} \Pi_{(A_{t_{\tilde{m}}}\cap \{1,\ldots, t_{\tilde{m}-1}\}) \cup (\{t_{\tilde{m}-1}+1,\ldots t_{\tilde{m}}\}\cap S_t) \cup U_{t_{\tilde{m}}}} \stackrel{(c)}{=} \Pi_{A_{t_{\tilde{m}}}\cup U_{t_{\tilde{m}}}} \stackrel{(d)}{\geq}  1/\alpha.
 $$
\ifdiff \added{As suggested by a referee, the individual steps of this derivation can be explained as follows:} \else 
As suggested by a referee, the individual steps of this derivation can be explained as follows: \fi
\begin{itemize}
    \item[(a)] \ifdiff \added{With $I=V\cup U_t$, it holds that $I\cap \{1,\ldots, t_{\tilde{m}}\}=(V\cap \{1,\ldots, t_{\tilde{m}}\}) \cup (U_t\cap \{1,\ldots, t_{\tilde{m}}\})$. Furthermore, $U_t\cap \{1,\ldots, t_{\tilde{m}}\}= U_{t_{\tilde{m}}}$.} \else With $I=V\cup U_t$, it holds that $I\cap \{1,\ldots, t_{\tilde{m}}\}=(V\cap \{1,\ldots, t_{\tilde{m}}\}) \cup (U_t\cap \{1,\ldots, t_{\tilde{m}}\})$. Furthermore, $U_t\cap \{1,\ldots, t_{\tilde{m}}\}= U_{t_{\tilde{m}}}$. \fi
    \item[(b)] \ifdiff \added{First, note that $(V\cap \{1,\ldots, t_{\tilde{m}}\})=(V\cap \{1,\ldots, t_{\tilde{m}-1}\}) \cup (V\cap \{t_{\tilde{m}-1}+1,\ldots, t_{\tilde{m}}\})$. For the second part of the union, we can write $V\cap \{t_{\tilde{m}-1}+1,\ldots, t_{\tilde{m}}\}=S_t\cap \{t_{\tilde{m}-1}+1,\ldots, t_{\tilde{m}}\}$. To see this, note that $V\subseteq S_t$ and, due to the definition of $\tilde{m}$, $|V\cap \{t_{\tilde{m}-1}+1,\ldots, t_{\tilde{m}}\}|=|V\cap \{1,\ldots, t_{\tilde{m}}\}|-|V\cap \{1,\ldots, t_{\tilde{m}-1}\}|\geq |S_{t_{\tilde{m}}}|-|S_{t_{\tilde{m}-1}}|=|S_t\cap \{t_{\tilde{m}-1}+1,\ldots, t_{\tilde{m}}\}|$. Now consider $J=V\cap \{1,\ldots, t_{\tilde{m}-1}\}$. By the definition of $\tilde{m}$, we have that $|J\cap \{1,\ldots,t_i\}|\leq |S_{t_i}|-i$ for all $i\in \{1,\ldots,\tilde{m}-1\}$, with an equality for $i=\tilde{m}-1$. Among all sets $J\subseteq S_t$ that satisfy the aforementioned constraints, the set $A$ at step $t_{\tilde{m}-1}$ \textit{after} completing line 13 of the algorithm, i.e. $A_{(t_{\tilde{m}-1})+}$, minimizes the product of the e-values. This is because $|(S_t\setminus J)\cap \{1,\ldots,t_i\}|\geq i$, $i\in \{1,\ldots,\tilde{m}-1\}$, and the \texttt{SeqE-Guard} algorithm exactly removes those indices from $S_t$ that correspond to the largest possible e-values. Since no indices were removed between the times $t_{\tilde{m}-1}$ and $t_{\tilde{m}}$, it follows that $A_{(t_{\tilde{m}}-1)+}\cap \{1,\ldots,t_{\tilde{m}-1}\}=A_{t_{\tilde{m}}}\cap \{1,\ldots,t_{\tilde{m}-1}\}$, showing the claim.} \else First, note that $(V\cap \{1,\ldots, t_{\tilde{m}}\})=(V\cap \{1,\ldots, t_{\tilde{m}-1}\}) \cup (V\cap \{t_{\tilde{m}-1}+1,\ldots, t_{\tilde{m}}\})$. For the second part of the union, we can write $V\cap \{t_{\tilde{m}-1}+1,\ldots, t_{\tilde{m}}\}=S_t\cap \{t_{\tilde{m}-1}+1,\ldots, t_{\tilde{m}}\}$. To see this, note that $V\subseteq S_t$ and, due to the definition of $\tilde{m}$, $|V\cap \{t_{\tilde{m}-1}+1,\ldots, t_{\tilde{m}}\}|=|V\cap \{1,\ldots, t_{\tilde{m}}\}|-|V\cap \{1,\ldots, t_{\tilde{m}-1}\}|\geq |S_{t_{\tilde{m}}}|-|S_{t_{\tilde{m}-1}}|=|S_t\cap \{t_{\tilde{m}-1}+1,\ldots, t_{\tilde{m}}\}|$. Now consider $J=V\cap \{1,\ldots, t_{\tilde{m}-1}\}$. By the definition of $\tilde{m}$, we have that $|J\cap \{1,\ldots,t_i\}|\leq |S_{t_i}|-i$ for all $i\in \{1,\ldots,\tilde{m}-1\}$, with an equality for $i=\tilde{m}-1$. Among all sets $J\subseteq S_t$ that satisfy the aforementioned constraints, the set $A$ at step $t_{\tilde{m}-1}$ \textit{after} completing line 13 of the algorithm, i.e. $A_{(t_{\tilde{m}-1})+}$, minimizes the product of the e-values. This is because $|(S_t\setminus J)\cap \{1,\ldots,t_i\}|\geq i$, $i\in \{1,\ldots,\tilde{m}-1\}$, and the \texttt{SeqE-Guard} algorithm exactly removes those indices from $S_t$ that correspond to the largest possible e-values. Since no indices were removed between the times $t_{\tilde{m}-1}$ and $t_{\tilde{m}}$, it follows that $A_{(t_{\tilde{m}}-1)+}\cap \{1,\ldots,t_{\tilde{m}-1}\}=A_{t_{\tilde{m}}}\cap \{1,\ldots,t_{\tilde{m}-1}\}$, showing the claim. \fi
    \item[(c)] \ifdiff \added{Again, since no indices were removed between the times $t_{\tilde{m}-1}$ and $t_{\tilde{m}}$, we have $A_{t_{\tilde{m}}}\cap \{t_{\tilde{m}-1}+1,\ldots, t_{\tilde{m}}\}=S_t\cap \{t_{\tilde{m}-1}+1,\ldots, t_{\tilde{m}}\}$.} \else Again, since no indices were removed between the times $t_{\tilde{m}-1}$ and $t_{\tilde{m}}$, we have $A_{t_{\tilde{m}}}\cap \{t_{\tilde{m}-1}+1,\ldots, t_{\tilde{m}}\}=S_t\cap \{t_{\tilde{m}-1}+1,\ldots, t_{\tilde{m}}\}$. \fi
    \item[(d)] \ifdiff \added{This follows immediately from the fact that $\Pi_{A_{t_i}\cup U_{t_i}}\geq 1/\alpha$ for all $i\in \{1,\ldots,m\}$.} \else This follows immediately from the fact that $\Pi_{A_{t_i}\cup U_{t_i}}\geq 1/\alpha$ for all $i\in \{1,\ldots,m\}$. \fi
\end{itemize}


\sloppy Hence, it remains to show that $\bd^{\bphi}(S_{t}) < d_{t-1}+1$, if $\Pi_{A_t\cup U_t}<1/\alpha$. Since $\Pi_{(A_t\cap \{1,\ldots, i\}) \cup U_t}< 1/\alpha$ for all $i\in \{1,\ldots,t-1\}$, $\Pi_{A_t\cup U_t}<1/\alpha$ implies that $\phi_{A_t \cup U_t}=0$. Furthermore, since $|S_t\setminus A_t|=d_{t-1}$, the claim follows.
\end{proof}

\putbib
\end{bibunit}

\end{document}